\documentclass[11pt]{article}

\usepackage[utf8]{inputenc}
\usepackage{mathrsfs}
\usepackage{amscd}
\usepackage{amsmath}
\usepackage{amssymb}
\usepackage{amstext}
\usepackage{amsthm}
\usepackage{bbold}
\usepackage{bm}
\usepackage{color}
\usepackage{framed}
\usepackage[dvips,letterpaper,margin=1in]{geometry}
\usepackage{graphicx}
\usepackage{hyperref}
\usepackage{mathtools}
\usepackage{colonequals}
\usepackage{rotating}
\usepackage{setspace}
\usepackage{verbatim}
\usepackage[dvipsnames]{xcolor}
\usepackage[sort,numbers]{natbib}
\DeclareSymbolFont{bbold}{U}{bbold}{m}{n}
\DeclareSymbolFontAlphabet{\mathbbold}{bbold}

\newtheorem{theorem}{Theorem}[section]
\newtheorem{remark}[theorem]{Remark}
\newtheorem{lemma}[theorem]{Lemma}

\newtheorem{definition}[theorem]{Definition}

\newtheorem{proposition}[theorem]{Proposition}
\newtheorem{corollary}[theorem]{Corollary}
\newtheorem{conjecture}[theorem]{Conjecture}

\newcommand{\Erdos}{Erd\H{o}s}
\newcommand{\Renyi}{R\'{e}nyi}

\newcommand{\ra}{\rangle}
\newcommand{\la}{\langle}

\newcommand{\RR}{\mathbb{R}}

\newcommand{\ZZ}{\mathbb{Z}}

\newcommand{\PP}{\mathbb{P}}
\newcommand{\Px}{\mathop{\mathbb{P}}}
\newcommand{\EE}{\mathbb{E}}
\newcommand{\Ex}{\mathop{\mathbb{E}}}

\newcommand{\bD}{\bm D}

\newcommand{\bL}{\bm L}

\newcommand{\bP}{\bm P}

\newcommand{\bW}{\bm W}
\newcommand{\bX}{\bm X}

\newcommand{\bZ}{\bm Z}

\newcommand{\bx}{\bm x}

\newcommand{\by}{\bm y}
\newcommand{\bz}{\bm z}

\newcommand{\sA}{\mathcal{A}}

\newcommand{\sE}{\mathcal{E}}

\newcommand{\sM}{\mathcal{M}}
\newcommand{\sN}{\mathcal{N}}

\newcommand{\sP}{\mathcal{P}}
\newcommand{\sQ}{\mathcal{Q}}

\newcommand{\Tr}{\mathrm{Tr}}

\newcommand{\diag}{\mathsf{diag}}

\newcommand{\cc}{\mathsf{cc}}

\DeclareMathOperator*{\argmax}{arg\,max}
\DeclareMathOperator*{\argmin}{arg\,min}

\newcommand{\Unif}{\mathsf{Unif}}

\newcommand{\aug}{\mathsf{aug}}
\newcommand{\dist}{\mathsf{dist}}
\newcommand{\pprod}{\mathsf{prod}}
\newcommand{\Cyc}{\mathsf{Cyc}}

\newcommand{\what}{\widehat}

\newcommand{\oa}{\bar{a}}
\newcommand{\ob}{\bar{b}}
\newcommand{\oc}{\bar{c}}
\newcommand{\oj}{\bar{j}}
\newcommand{\ok}{\bar{k}}
\newcommand{\ol}{\bar{\ell}}

\newcommand\numberthis{\addtocounter{equation}{1}\tag{\theequation}}

\hypersetup{
  colorlinks = true
}

\title{Strong recovery of geometric planted matchings}
\date{July 12, 2021}

\author{}
\usepackage{authblk}
\author[1]{Dmitriy Kunisky\thanks{Email: \textit{kunisky@cims.nyu.edu}. Partially supported by NSF grants DMS-1712730 and DMS-1719545.}}
\author[1]{Jonathan Niles-Weed\thanks{Email: \textit{jnw@cims.nyu.edu}. Partially supported by NSF grant DMS-2015291.}}

\affil[1]{Department of Mathematics, Courant Institute of Mathematical Sciences, New York University}

\begin{document}

\maketitle
\thispagestyle{empty}

\begin{abstract}
We study the problem of efficiently recovering the matching between an unlabelled collection of $n$ points in $\mathbb{R}^d$ and a small random perturbation of those points.
    We consider a model where the initial points are i.i.d.\ standard Gaussian vectors, perturbed by adding i.i.d.\ Gaussian vectors with variance $\sigma^2$.
    In this setting, the maximum likelihood estimator (MLE) can be found in polynomial time as the solution of a linear assignment problem.
    We establish thresholds on $\sigma^2$ for the MLE to perfectly recover the planted matching (making no errors) and to strongly recover the planted matching (making $o(n)$ errors) both for $d$ constant and $d = d(n)$ growing arbitrarily.
    Between these two thresholds, we show that the MLE makes $n^{\delta + o(1)}$ errors for an explicit $\delta \in (0, 1)$.
    These results extend to the geometric setting a recent line of work on recovering matchings planted in random graphs with independently-weighted edges.
    Our proof techniques rely on careful analysis of the combinatorial structure of partial matchings in large, weakly dependent random graphs using the first and second moment methods.
    
    \begin{comment}
    Using this, we describe a stylized application to tracking the motion of particles undergoing Brownian motion by recording their positions periodically and computing the MLE matching between consecutive ``snapshots.''
    We predict using our results and verify numerically that the benefits of more frequent snapshots are almost negligible when $d \in \{1, 2\}$, but are substantial once $d \geq 3$.
    \end{comment}
    
    \begin{comment}
    Finally, we compare the MLE to two variants of greedy matching algorithms, finding that each of the three algorithms suffers drawbacks in different regimes of $d$ and $\sigma^2$.
    \end{comment}
    
    \begin{comment}
     When $d = d(n) = o(\log n)$, we show that our strong recovery threshold is tight, in that for greater $\sigma^2$ the MLE makes $\Omega(n)$ errors if $d$ is constant, and $n^{1 - o(1)}$ errors if $1 \ll d \ll \log n$; when $d \gtrsim \log n$, we conjecture the same and provide both theoretical and numerical evidence.
    \end{comment}
\end{abstract}

\clearpage

\tableofcontents
\thispagestyle{empty}

\clearpage

\section{Introduction}
\pagenumbering{arabic}
\setcounter{page}{1}
%Consider a set of $n$ independent vectors $\bx_1, \dots, \bx_n \sim \sN(\bm 0, \bm I_d)$.
%
%
%For some fixed $\sigma^2 > 0$, perturb the vectors by addi 
%, then draw noise vectors $\bz_1, \dots, \bz_n \sim \sN(\bm 0, \sigma^2 \bm I_d)$ independently (of one another and the $\bx_i$) and set $\by_i \colonequals \bx_i + \bz_i$.
%We then draw a hidden permutation $\pi^{\star} \sim \Unif(S_n)$ (as in the independent model) and observe the tuple $(\bx_1, \dots, \bx_n, \by_{\pi^{\star}(1)}, \dots, \by_{\pi^{\star}(n)})$.
%The goal is to estimate the planted permutation $\pi^{\star}$ from this observation.

Consider a set of $n$ unlabelled particles $\{\bx_1, \dots, \bx_n\}$ in $\RR^d$ undergoing random motion.
A short time later, the particles are observed at new locations $\{\by_1, \dots, \by_n\}$.
Is it possible to ascertain which particles correspond to which?
This problem---known as \emph{multitarget tracking}---was proposed for theoretical analysis by \cite{CKKVZ-2010-ParticleTrackingBP}, and has a wide range of applications in many scientific contexts where it is useful to infer the trajectories of objects from a succession of still images.

For concreteness, we formalize this question as follows:
fix a dimension $d \in \ZZ_+$, a sample size $n \in \ZZ_+$, and a noise variance $\sigma^2 \in \RR_+$.
We first draw $\bx_1, \dots, \bx_n \sim \sN(\bm 0, \bm I_d)$ independently, then draw noise vectors $\bz_1, \dots, \bz_n \sim \sN(\bm 0, \sigma^2 \bm I_d)$ independently (of one another and the $\bx_i$) and set $\by_i \colonequals \bx_i + \bz_i$.
We then draw a hidden permutation $\pi^{\star} \sim \Unif(S_n)$ and observe the tuple $(\bx_1, \dots, \bx_n, \by_{\pi^{\star}(1)}, \dots, \by_{\pi^{\star}(n)})$.
The goal is to estimate the planted permutation $\pi^{\star}$ from this observation.

While this model is quite natural, rigorously analyzing its statistical and computational properties has proven challenging, chiefly because the pairwise distances $\{\|\bx_i - \by_j\|^2\}_{i, j = 1}^n$ are not independent.
In the interest of identifying a mathematically tractable alternative, \cite{CKKVZ-2010-ParticleTrackingBP} suggested to study a simpler model where independent random variables are substituted for these distances.
Under this simplified model, we observe a matrix $\bW \in \RR^{n \times n}$ where, for a random hidden permutation $\pi^{\star}$, the entries $W_{ij}$ are drawn from a distribution $\sP$ when $\pi^{\star}(i) = j$, and another distribution $\sQ$ otherwise, all independently.

Models of this type have attracted significant recent interest in the computer science and statistics communities, and precise results are now known in a number of different settings \cite{MMX-2019-PlantedMatching,SSZ-2020-SparsePlantedMatching,DWXY-2021-PlantedMatchingInfiniteOrder}.
Despite this progress, however, the original problem of recovering planted \emph{geometric} matchings to our knowledge has not received any attention since its proposal by \cite{CKKVZ-2010-ParticleTrackingBP}.

In this work, we make progress on this original question.
We precisely characterize the performance of a natural recovery procedure based on the linear assignment problem, and establish thresholds on $\sigma^2$ for this procedure to recover the planted matching with various amounts of error.
Our results also suggest new conjectures about the performance of a natural online algorithm for multitarget tracking which has been proposed in the signal processing literature~\cite{PSHBH-2006-MultiObjectTracking,BGORU-2016-OnlineTracking,SA-2016-IterativeHungarianTracking}.
Taken as a whole, our results indicate regimes in which it is possible to recover geometric planted matchings to high accuracy in polynomial time.
%
%\paragraph{Geometric planted matching model}
%
%We are given a dimension $d \in \ZZ_+$, a sample size $n \in \ZZ_+$, and a noise variance $\sigma^2 \in \RR_+$.
%We first draw $\bx_1, \dots, \bx_n \sim \sN(\bm 0, \bm I_d)$ independently, then draw noise vectors $\bz_1, \dots, \bz_n \sim \sN(\bm 0, \sigma^2 \bm I_d)$ independently (of one another and the $\bx_i$) and set $\by_i \colonequals \bx_i + \bz_i$.
%We then draw a hidden permutation $\pi^{\star} \sim \Unif(S_n)$ (as in the independent model) and observe the tuple $(\bx_1, \dots, \bx_n, \by_{\pi^{\star}(1)}, \dots, \by_{\pi^{\star}(n)})$.
%The goal is to estimate the planted permutation $\pi^{\star}$ from this observation.

\paragraph{Maximum likelihood estimation}
We will focus on the \emph{maximum likelihood estimator (MLE)} of $\pi^{\star}$ from the observations, which is given by
\begin{align*}
    \what{\pi} \colonequals \argmax_{\pi \in S_n} \, \exp\left(-\frac{1}{2\sigma^2} \sum_{i= 1}^n \|\bx_i -\by_{\pi(i)}\|^2\right) = \argmin_{\pi \in S_n} \sum_{i = 1}^n \|\bx_i - \by_{\pi(i)}\|^2.
\end{align*}
One advantage of this estimator is that it does not depend on the variance $\sigma^2$, which may not be known in practice.
Crucially, despite being given as the solution to an optimization problem over $S_n$, the estimator can be computed in polynomial time, since it is an instance of the \emph{linear assignment problem}.
Solutions may therefore be computed efficiently either by an exact relaxation to a linear program over doubly stochastic matrices, or with specialized combinatorial algorithms such as the Hungarian algorithm~\cite{Kuhn-1955-HungarianAlgorithm,Bertsekas-1990-AuctionAlgorithm}.

We note that though the MLE is a canonical choice of estimator, it is not the only available polynomial-time approach.
Another natural approach is to estimate $\pi^\star$ by greedily matching each point $\bx_i$ to its nearest neighbor.
One can show that this algorithm is competitive with the MLE in some regimes, but is strictly dominated by the MLE when the dimension is large.
We discuss this algorithm and a similar greedy algorithm which seeks to maximize the correlation between $\bx_i$ and its matched point in Appendix~\ref{app:greedy}.

We assess the error incurred by the MLE by counting how many indices of $[n]$ it matches incorrectly.
We define the (random) set of such errors,
\begin{equation}
    \sE = \{i \in [n]: \what{\pi}(i) \neq \pi^{\star}(i)\}.
\end{equation}
We will primarily be concerned with the behavior of the random variable $|\sE|$.
Its law is unchanged by fixing $\pi^{\star}$, so we assume without loss of generality that $\pi^{\star}$ is the identity permutation.
We lastly introduce some standard jargon.
We say $\what{\pi}$ achieves \emph{strong recovery} (of $\pi^{\star}$) if $|\sE| = o(n)$, achieves \emph{perfect recovery} if $|\sE| = 0$, and achieves \emph{near-perfect recovery} or \emph{sublinear error} if $0 < |\sE| \leq o(n)$.
In contrast, we say $\what{\pi}$ makes a \emph{macroscopic} number of errors if $|\sE| = \Omega(n)$.

Most prior work on planted matching problems has focused on establishing when strong recovery is or is not achieved.
We will partly address this question, but we will also study the \emph{polynomial error rate} given by $\frac{\log(1 \vee |\sE|)}{\log n}$.
As we show below in Section~\ref{sec:tracking}, this finer control is valuable in applications to multitarget tracking over time.

%\paragraph{Tracking Brownian motions}

\paragraph{Related work}

The limits of recovering planted matchings under independent weights are increasingly well understood.
These models exhibit a phase transition in the recoverability of $\pi^{\star}$, which was conjectured by \cite{CKKVZ-2010-ParticleTrackingBP}, proved in a special case by \cite{MMX-2019-PlantedMatching}, and studied in greater detail and generality by \cite{SSZ-2020-SparsePlantedMatching,DWXY-2021-PlantedMatchingInfiniteOrder}.
The approach of \cite{MMX-2019-PlantedMatching} in particular may be viewed as an extension to the planted setting of an earlier line of work studying optimal matchings under i.i.d.\ weights, the so-called \emph{random assignment model} \cite{MP-1987-RandomAssignment,Aldous-1992-AsymptoticsRandomAssignment,Parisi-1998-RandomAssignment,Aldous-2001-ZetaRandomAssignment}.
Despite the sophistication of these results, their techniques rely heavily on the independence assumption, and many of their conclusions remain conjectural in the geometric matching setting.

More broadly, various problems of estimating combinatorial structures from noisy observations have received much attention in recent years.
As in our case, the models making strong independence assumptions have been the most amenable to analysis; notable examples include the stochastic block model \cite{DKMZ-2011-SBM,Moore-2017-SBMReview,Abbe-2017-SBMReview} and the planted clique model \cite{Jerrum-1992-LargeCliques,AKS-1998-PlantedClique,BHKKMP-2019-PlantedClique}, both of which may be viewed as models of \emph{community detection} in networks.
One of the remarkable phenomena that such models exhibit is the \emph{statistical-to-computational gap}, where in a range of model parameters it is possible to estimate the planted object, but (conjecturally) only with prohibitively costly algorithms (see, e.g., \cite{BPW-2018-GapsNotes}).
There is not yet evidence that planted matching problems ever have such gaps, but it is an interesting open question to determine if this in fact ever occurs.
We note also that the difference between independent planted matching models and our geometric planted matching model is analogous to the difference between the stochastic block model of network community structure and the stochastic ball model \cite{ABCKVW-2015-RelaxRoundClustering,IMPV-2017-CertifiableKMeans} and similar Gaussian mixture models \cite{MVW-2017-ClusteringSubgaussianMixtures,LLLSW-2020-KMeansProximity} analyzed more recently in the community detection literature.

Finally, the question of optimally matching i.i.d.~random points is a classical topic in probability theory and computational geometry~\cite{AST-2019-PDEApproach,Ledoux-2019-OptimalMatching,Ledoux-2018-OptimalMatchingII,Talagrand-2018-NonStandardMatching,CLPS-2014-HypothesisEuclideanMatching,Tal14,AjtKomTus84,ShoYuk91, LeiSho89}.
This line of work studies a natural \emph{null model} counterpart to ours, where all $2n$ points $\bx_1, \dots, \bx_n, \by_1, \dots, \by_n$ are i.i.d.
This model is the geometric analogue of the random assignment problem, and it would be interesting to understand whether the optimal transport techniques developed for analyzing matchings of i.i.d.\ points (such as the PDE approach of \cite{CLPS-2014-HypothesisEuclideanMatching,AST-2019-PDEApproach}) can be imported to the study of geometric planted matching models, in the same way that \cite{MMX-2019-PlantedMatching} imported the techniques of \cite{Aldous-1992-AsymptoticsRandomAssignment,Aldous-2001-ZetaRandomAssignment} related to local weak convergence from the random assignment problem to their independent planted matching model.

\subsection{Notation}
Throughout, we focus on the $n \to \infty$ limit and let $d = d(n)$ and $\sigma^2 = \sigma^2(n)$ scale at various rates with $n$.
The asymptotic symbols $o(\cdot), O(\cdot), \omega(\cdot), \Omega(\cdot), \Theta(\cdot), \ll, \sim,$ and $\gg$ will have their usual meanings with reference to the limit $n \to \infty$, and events which occur with probability $1 - o(1)$ are said to hold ``with high probability."

We also introduce some further notation for the MLE.
We define two \emph{cost matrices} $\bW^{(0)}, \bW \in \RR^{n \times n}$ with entries
\begin{align}
    W_{ij}^{(0)} &\colonequals \|\bx_i - \by_j\|^2, \\
    W_{ij} &\colonequals \langle \bx_i, \by_j \rangle,
\end{align}
and note that, writing $\bP_{\pi}$ for the permutation matrix of a permutation $\pi$, the MLE is equivalently
\begin{equation}\label{eq:W-mle}
    \what{\pi} = \argmin_{\pi \in S_n} \langle \bW^{(0)}, \bP_{\pi} \rangle = \argmax_{\pi \in S_n} \langle \bW, \bm P_{\pi} \rangle,
\end{equation}
since, upon expanding the squared distances, each $\|\bx_i\|^2$ and $\|\by_j\|^2$ occurs exactly once for any $\pi$.

For $a, b \in \RR$, we write $a \vee b$ for the maximum of $a$ and $b$ and $a \wedge b$ for their minimum.
Given $x > 0$, we let $\log_+(x) \colonequals 0 \vee \log(x)$.

\subsection{Main Results}
To state our results, we consider three different regimes: the low-dimensional regime where $d = o(\log n)$, the logarithmic regime where $d = \Theta(\log n)$, and the high-dimensional regime where $d = \omega(\log n)$.
In each, we identify the behavior of $|\sE|$ as a function of $\sigma^2$.
As our proofs make clear, the difference between these regimes is justified by the fact that the quantity
\begin{equation*}
\frac{d}{\log n} \log(1 + \sigma^{-2})
\end{equation*}
plays the role of a signal-to-noise ratio for our problem, which suggests that the correct scaling of~$\sigma$ is $\sigma^2 = \Theta(n^{-\xi/d})$ for some $\xi > 0$ in the low-dimensional regime, $\sigma^2 = \Theta(1)$ in the logarithmic regime, and $\sigma^2 = \Theta(\frac d {\log n})$ in the high-dimensional regime.
Our main results verify these claims.

In the low-dimensional regime, we are able to resolve the thresholds between perfect recovery, strong recovery, and macroscopic error.
%Though we are not able to rigorously establish all of these thresholds in the logarithmic and high-dimensional regimes, our results nevertheless allow us to identify ranges of $\sigma^2$ in which perfect and strong recovery hold. 

\begin{theorem}[Low-dimensional regime]
    \label{thm:d-ll-logn}
    Suppose that $d = o(\log n)$.
    \begin{enumerate}
        \item (Perfect recovery) If $\sigma^2 = o(n^{-4/d})$, then $|\sE| = 0$ with high probability.
        \item (Constant error) If $\sigma^2 = \Theta(n^{-4/d})$, then $\EE|\sE|$ is bounded; in particular $|\sE| \leq f(n)$ for any $f(n) = \omega(1)$.
        \item (Sublinear error) If $n^{-4/d} \ll \sigma^2 \ll n^{-2/d}$, then there exists an absolute constant $c > 0$ such that, for any $f(n) = \omega(1)$,
        \begin{equation}
            \frac{c}{\sqrt{d}}\sigma^d n^2 \leq |\sE| \leq f(n) \sigma^d n^2.
        \end{equation}
            In particular, if $\frac{d}{\log n} \log(1+\sigma^{-2}) \to \xi \in [2, 4]$, then the following convergence in probability holds as $n \to \infty$:
    \begin{equation}
        \frac{\log(1 \vee |\sE|)}{\log n} \to 2 - \frac{\xi}{2}.
    \end{equation}
        \item (Linear or nearly-linear error) If $\sigma^2 \geq an^{-2/d}$ for some $a > 0$, then there exists $c = c(a)$ such that $|\sE| \geq e^{-cd}n$ with high probability.
    \end{enumerate}

%    In particular, for arbitrary $\sigma^2 = \sigma^2(n) > 0$,     \begin{equation}
%        \frac{\log(1 \vee |\sE|)}{\log n} \to 0 \vee \left(2 - \frac{d}{2\log n}\log(1 + \sigma^{-2})\right).
%    \end{equation}
\end{theorem}
\noindent
%We note that, if $\sigma^2 = \Theta(n^{-\xi / d})$ for some $\xi \in [2, 4]$, then the error rate convergence may be simplified to $\frac{\log(1 \vee |\sE|)}{\log n} \to 2 - \frac{\xi}{2}$.
Note that when $\sigma^2 = \Omega(n^{-2/d})$ and $d$ is a constant not depending on $n$, Theorem~\ref{thm:d-ll-logn} implies that $|\sE| = \Omega(n)$ with high probability; this is the only regime where we are able to show that the MLE actually incurs macroscopic error.
When $1 \ll d \ll \log n$ with the same scaling of $\sigma^2$, we find the nearly macroscopic $|\sE| = \Omega(n^{1 - o(1)})$.

In the logarithmic regime we obtain similar results, except that the range of $\sigma^2$ yielding sublinear errors appears to end at a point when $|\sE| = \Theta(n^{\delta})$ for some $\delta < 1$.
In fact, in Conjecture~\ref{conj:non-recovery} below we predict the existence of a discontinuity in the limiting value of $\frac{\log(1 \vee |\sE|)}{\log n}$, where the error rate jumps sharply from $|\sE| = \Theta(n^\delta)$ to $|\sE| = \Omega(n)$.
\begin{theorem}[Logarithmic regime]
    \label{thm:d-theta-logn}
    Suppose that $d \sim a \log n$ for some $a > 0$, and that $\sigma^2$ is constant not depending on $n$.
    \begin{enumerate}
        \item (Perfect recovery) If
        \begin{equation}
            \sigma^2 < \frac{1}{e^{4/a} - 1},
        \end{equation}
        then $|\sE| = 0$ with high probability.
        \item (Sublinear error) If
        \begin{equation}\label{eq:sublinear_width}
            \frac{1}{e^{4/a} - 1} \leq \sigma^2 < \frac{1}{(2e^{1/a} - 1)^2 - 1},
        \end{equation}
        then the following convergence in probability holds:
        \begin{equation}\label{eq:log_convergence}
            \frac{\log(1 \vee |\sE|)}{\log n} \to 2 - \frac{a}{2}\log(1 + \sigma^{-2}).
        \end{equation}
    \end{enumerate}
\end{theorem}
\noindent
The quantity on the right side of~\eqref{eq:log_convergence} equals zero at the lower limit $\sigma^2 = \frac{1}{e^{4/a} - 1}$, and equals $2 - a\log(2e^{1/a} - 1) \in (0, 1)$ at the upper limit $\sigma^2 = \frac{1}{(2e^{1/a} - 1)^2 - 1}$ for any $a > 0$.
As $a \to \infty$, the width of the sublinear error regime given in~\eqref{eq:sublinear_width} is $\frac{1}{(2e^{1/a} - 1)^2 - 1} - \frac{1}{e^{4/a} - 1} = \frac{1}{8} + o(1)$, so this is indeed a non-trivial range of $\sigma^2$ on the critical scale $\sigma^2 = \Theta(1)$.

Next, we treat the remaining high-dimensional regime.
Here our results only describe perfect recovery; however, Conjecture~\ref{conj:non-recovery} will again predict that on the scale of $\sigma^2$ indicated below, greater noise results in macroscopic error.
\begin{theorem}[High-dimensional regime]
    \label{thm:d-omega-logn}
    Suppose that $d = \omega(\log n)$.
    If for some $\epsilon > 0$
    \begin{equation}
        \sigma^2 \leq \left(\frac{1}{4} - \epsilon\right)\frac{d}{\log n},
    \end{equation}
    then $|\sE| = 0$ with high probability.
\end{theorem}

Finally, we state a supplementary conjecture, which we will discuss in greater detail in Section~\ref{sec:techniques}, where we show how it is suggested by the first moment combinatorics of augmenting cycles.
If true, this conjecture would complete the high-level picture described by our results, in each regime of $d$ showing that for the remaining $\sigma^2$ not covered by our results, the MLE makes a macroscopic number of errors.
\begin{conjecture}
    \label{conj:non-recovery}
    Suppose that any of the following conditions holds:
    \begin{enumerate}
        \item $1 \ll d \ll \log n$ and, for some $\epsilon > 0$, $\sigma^2 \geq n^{-(2 - \epsilon)/d}$.
        \item $d \sim a\log n$ and, for some $\epsilon > 0$, $\sigma^2 \geq \frac{1}{(2e^{1/a} - 1)^2 - 1} + \epsilon$.
        \item $d = \omega(\log n)$ and, for some $\epsilon > 0$, $\sigma^2 \geq (\frac{1}{4} + \epsilon)\frac{d}{\log n}$.
    \end{enumerate}
    Then, for some $c = c(\epsilon) > 0$, $|\sE| \geq cn$ with high probability.
\end{conjecture}
\begin{comment}
Indeed, it seems unlikely that the constant $c$ should depend on $\epsilon > 0$, so it is perhaps also reasonable to expect the stronger result to hold with $c = 1 - \epsilon^{\prime}$ for any $\epsilon^{\prime} > 0$, i.e., a ``total'' failure of recovery for the MLE in this regime.
\end{comment}
\noindent
If true, Conjecture~\ref{conj:non-recovery} together with Theorem~\ref{thm:d-theta-logn} would surprisingly imply a discontinuity in the value of $\frac{\log(1 \vee |\sE|)}{\log n}$ as a function of $\sigma^2$ when $d = a\log n$ at $\sigma^2 = \frac{1}{(2e^{1/a} - 1)^2 - 1}$: from the left this quantity would tend to a limit $2 - a \log(2e^{1/a} - 1)$ strictly smaller than 1, while from the right it would equal 1.
As $a \to 0$, the size of this jump would shrink, recovering in the limit the continuous behavior of the $d \ll \log(n)$ case.
We illustrate these error curves and the predicted jump in Figure~\ref{fig:error-numerics}; see also Section~\ref{sec:techniques} for discussion of theoretical evidence for this prediction.

\begin{comment}
When $d = \omega(\log n)$, the stronger version of Conjecture~\ref{conj:non-recovery} stated above would imply an ``all-or-nothing'' behavior for the MLE around the critical $\sigma^2 = \frac{1}{4} \frac{d}{\log n}$: any smaller $\sigma^2$ on the same scale would result in perfect recovery, while any larger $\sigma^2$ would result in $|\sE| = n - o(n)$ with high probability.
\end{comment}

\subsection{Stylized Application: Online Tracking of Brownian Motions}
\label{sec:tracking}
As an application of our results, we consider a stylized motion tracking model, similar to the one proposed by \cite{CKKVZ-2010-ParticleTrackingBP}. 
Suppose that $\bx_1(t), \dots, \bx_n(t) \in \RR^d$ are independent standard Brownian motions in dimension $d = O(1)$, started from $\bx_i(0)$ independent standard Gaussian vectors.
We view these Brownian motions as the evolution of indistinguishable particles, whose motion we would like to track over time: for some fixed $\delta > 0$, we observe this collection of particles (but not their labels) at times $t = k \delta$ for each integer $k  \geq 0$. On the basis of these observations, we would like to track the identities of each particle over some large interval $t \in [0, T]$ as accurately as possible.
 
A natural approach is an iterative matching algorithm: having observed the point set $X_k = \{\bx_1(k\delta), \dots, \bx_n(k\delta)\}$ for each integer $k \geq 0$, repeatedly compute the MLE matching $\what{\pi}_k$ between $X_{k - 1}$ and $X_k$ for $k \geq 1$. 
Then, the composition $\what{\pi} = \what{\pi}_1 \cdots \what{\pi}_K$ gives a plausible matching between $X_0$ and $X_K$, which attempts to track the Brownian motions up to time $T = K\delta$.
In fact, this approach is frequently used in practical engineering applications in concert with various preprocessing and filtering pipelines \cite{PSHBH-2006-MultiObjectTracking,BGORU-2016-OnlineTracking,SA-2016-IterativeHungarianTracking}.
We illustrate a small example in Figure~\ref{fig:tracking-schematic}.
How large can we make this $T$ while having the final matching correctly identify at least, say, half of the particles, i.e., having $\what{\pi}$ fix at least half of the points of $[n]$?\footnote{All manner of quantities describing the approach of $\what{\pi}$ to a uniformly random permutation, such as total variation distance in the style of results on Markov chain mixing times, would be interesting to consider; we restrict our discussion to the number of fixed points for the sake of simplicity.}
Let us define the expectation of this time,
\begin{equation}
    T_{\max} = T_{\max}(\delta, n) \colonequals \delta \cdot \EE \min\{K: \what{\pi}_1 \cdots \what{\pi}_K \text{ has fewer than } n / 2 \text{ fixed points}\}.
\end{equation}

Clearly we expect decreasing $\delta$---taking snapshots more frequently---to increase $T_{\max}$.
We can use our results for $d$ constant to make an informal prediction as to the behavior of this tradeoff.
The displacement of a Brownian motion in time $\delta$ has law $\sN(0, \delta \bm I_d)$, so each time step looks like our earlier setup with $\sigma^2 = \delta$.
Thus suppose $n^{-4/d} \lesssim \delta \lesssim n^{-2/d}$.
Then, we expect the error incurred by $\what{\pi}_k$ to be roughly $\delta^{d / 2}n^2$ for each $k$.
Supposing that these errors affect different indices in each time step, we then expect to make $\Omega(n)$ errors in total once $K > n / (\delta^{d / 2}n^2) = \delta^{-d/2} / n$.
Thus, we expect $T_{\max} \sim \delta K = \delta^{1 - d/2} / n$.

One case to which this argument certainly does \emph{not} apply is $d = 1$: in this case, the difference between the positions of any two particles is itself a Brownian motion which will eventually cross zero (meaning that the particles will collide), and by a standard argument of time inversion of Brownian motion will in fact cross zero infinitely many times in the vicinity of any such crossing (meaning that the particles will collide infinitely many times immediately following their first collision).
Indeed, we illustrate in Figure~\ref{fig:tracking-error} below that, when $d = 1$, the error of tracking appears to be driven by such collisions and does not depend at all on the sampling interval $\delta$.
However, we conjecture that the above heuristic is sound for larger dimension.

\begin{conjecture}
    \label{conj:tracking}
    Suppose that $d \geq 2$ and $\delta = n^{-\xi / d}$ for some $\xi \in [2, 4]$.
    Then, $T_{\max} \sim \frac{\delta^{1 - d/2}}{n} f(n) = n^{\xi/2 - \xi/d - 1} f(n)$ for some $1 / \mathsf{polylog}(n) \leq f(n) \leq \mathsf{polylog}(n)$.
\end{conjecture}
\noindent
A surprising consequence of this conjecture would be that, when $d = 2$, there is a large range of $\delta$ over which the improvement in $T_{\max}$ gained for decreasing $\delta$ is only logarithmic in $\delta$---the situation is hardly better than $d = 1$---while once $d \geq 3$ this improvement becomes polynomial in $\delta$.
This criticality of $d = 2$ seems to resemble similar phenomena in the structure of optimal matchings of i.i.d.\ points in the null model \cite{AjtKomTus84,Ledoux-2019-OptimalMatching,Ledoux-2018-OptimalMatchingII,Talagrand-2018-NonStandardMatching}.
While it is difficult to make $n$ sufficiently large to overcome finite-size effects and resolve the exponents we are interested in numerically, as alternative evidence we plot the number of errors over time for a fixed small $n$ and various $\delta$ and $d$ in Figure~\ref{fig:tracking-error}.
We observe something qualitatively similar to the Conjecture: when $d = 2$ the error changes logarithmically over several orders of magnitude of $\delta$, while once $d = 3$ the error changes much more rapidly, plausibly polynomially.

Proving Conjecture~\ref{conj:tracking} would require several improvements over our current results, and represents an interesting question for future work.
At a minimum, doing so would require better understanding of the concentration properties of $|\sE|$ in the low-dimensional regime.
Obtaining stronger concentration bounds would also open the door to understanding what happens when $\delta \ll n^{-4/d}$, when each time step is in our ``perfect recovery'' regime and most time steps do not introduce new errors.
%We remark that our results do not immediately prove a useful bound on $T_{\max}$: if the number of errors the MLE made were highly concentrated around its expectation, our first moment computations would give a lower bound; however, the upper bounds of our ``with high probability'' statements come only from Markov's inequality, and thus correspond to quite weak concentration inequalities not suitable for a union bound over many time steps.
%Moreover, it would be interesting to describe what happens once $\delta \ll n^{-4/d}$, when each time step is in our ``perfect recovery'' regime.
%This would again require a precise understanding of how the small probabilities of making various numbers of errors scale in the perfect recovery regime.

\begin{figure}[p]
    \centering
    \begin{tabular}{c@{\hskip.0cm}c@{\hskip.0cm}c@{\hskip.0cm}}
        \includegraphics[scale=0.35]{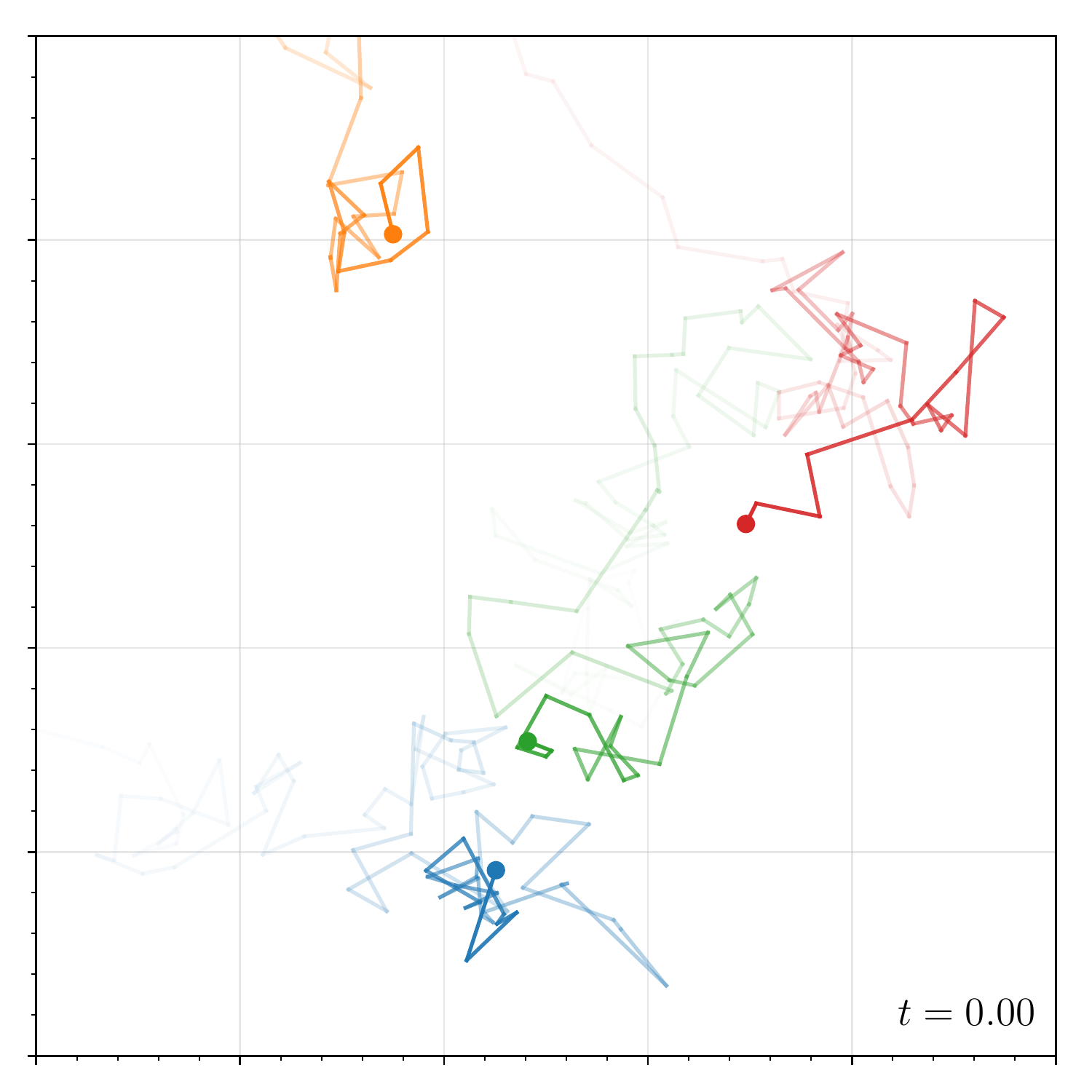} & 
        \includegraphics[scale=0.35]{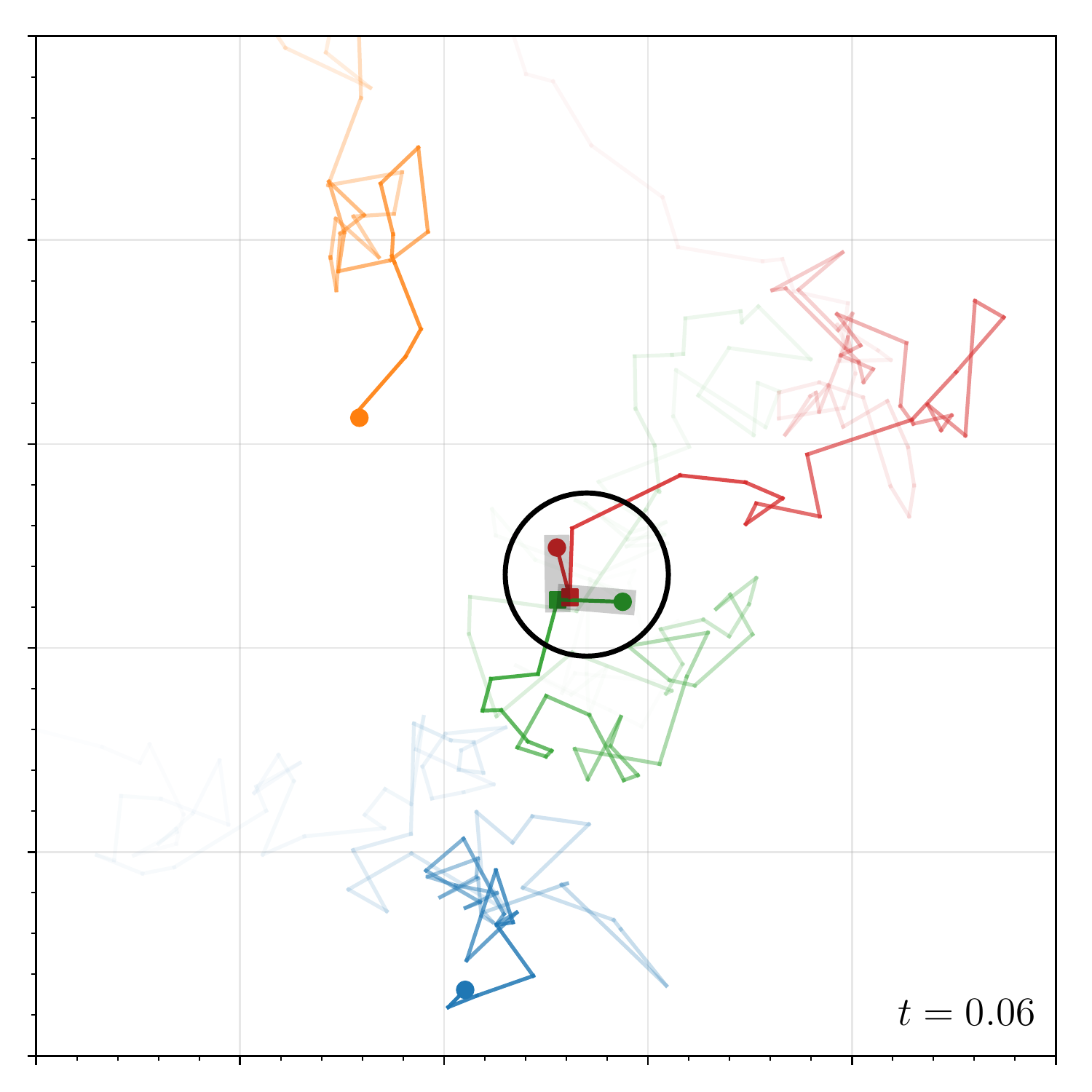} &
        \includegraphics[scale=0.35]{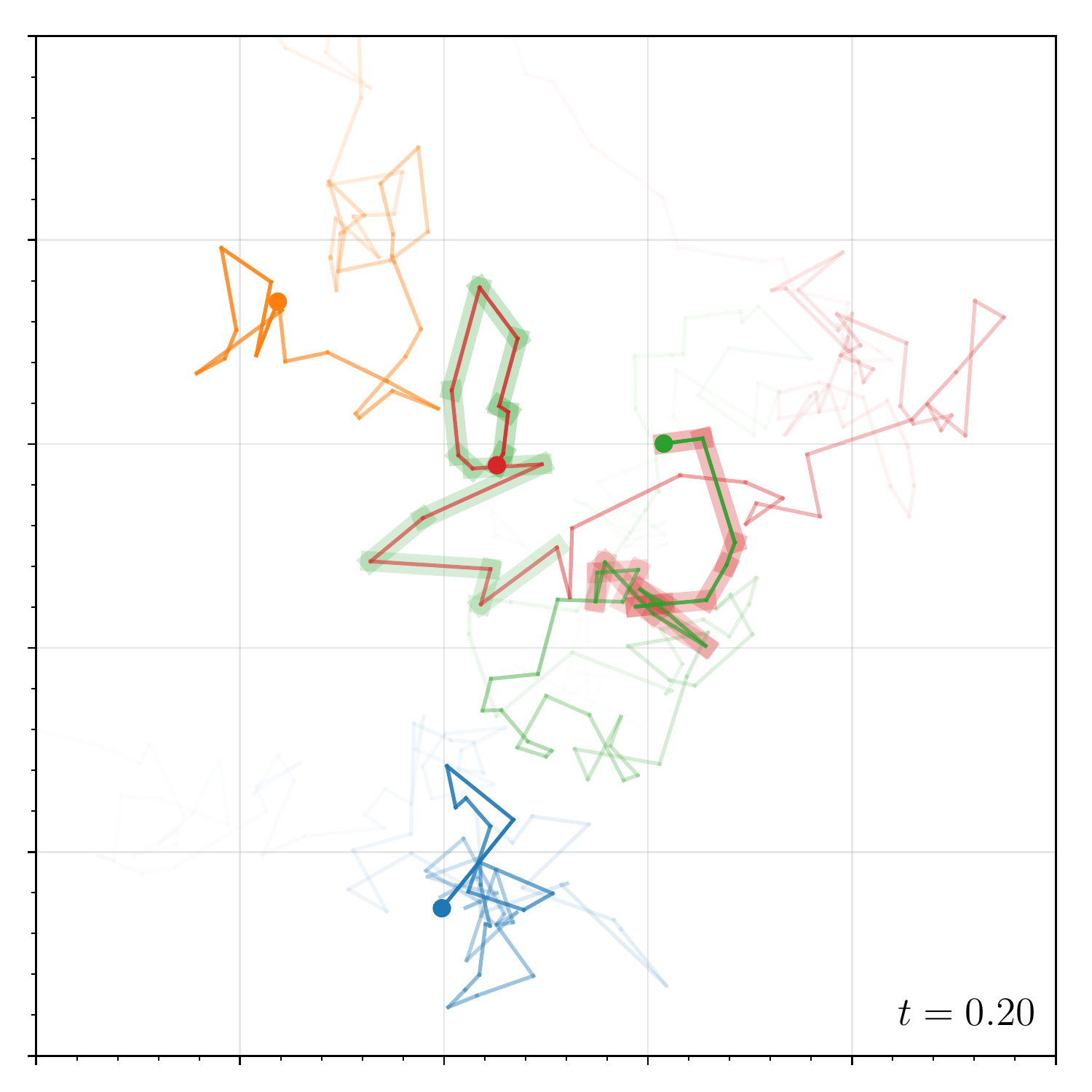}
    \end{tabular}
    \begin{tabular}{c@{\hskip.0cm}c}
        \includegraphics[scale=0.35]{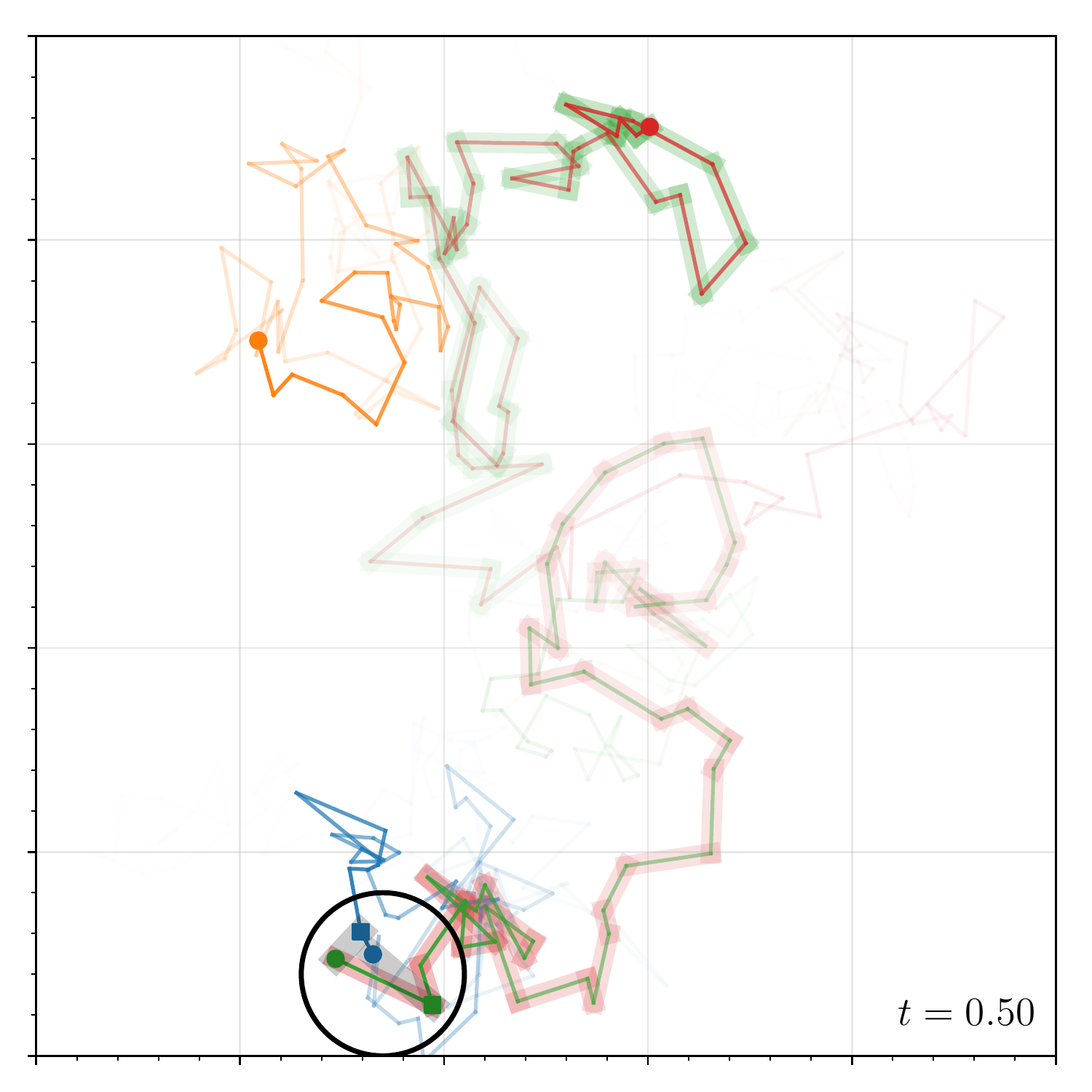} &
        \includegraphics[scale=0.35]{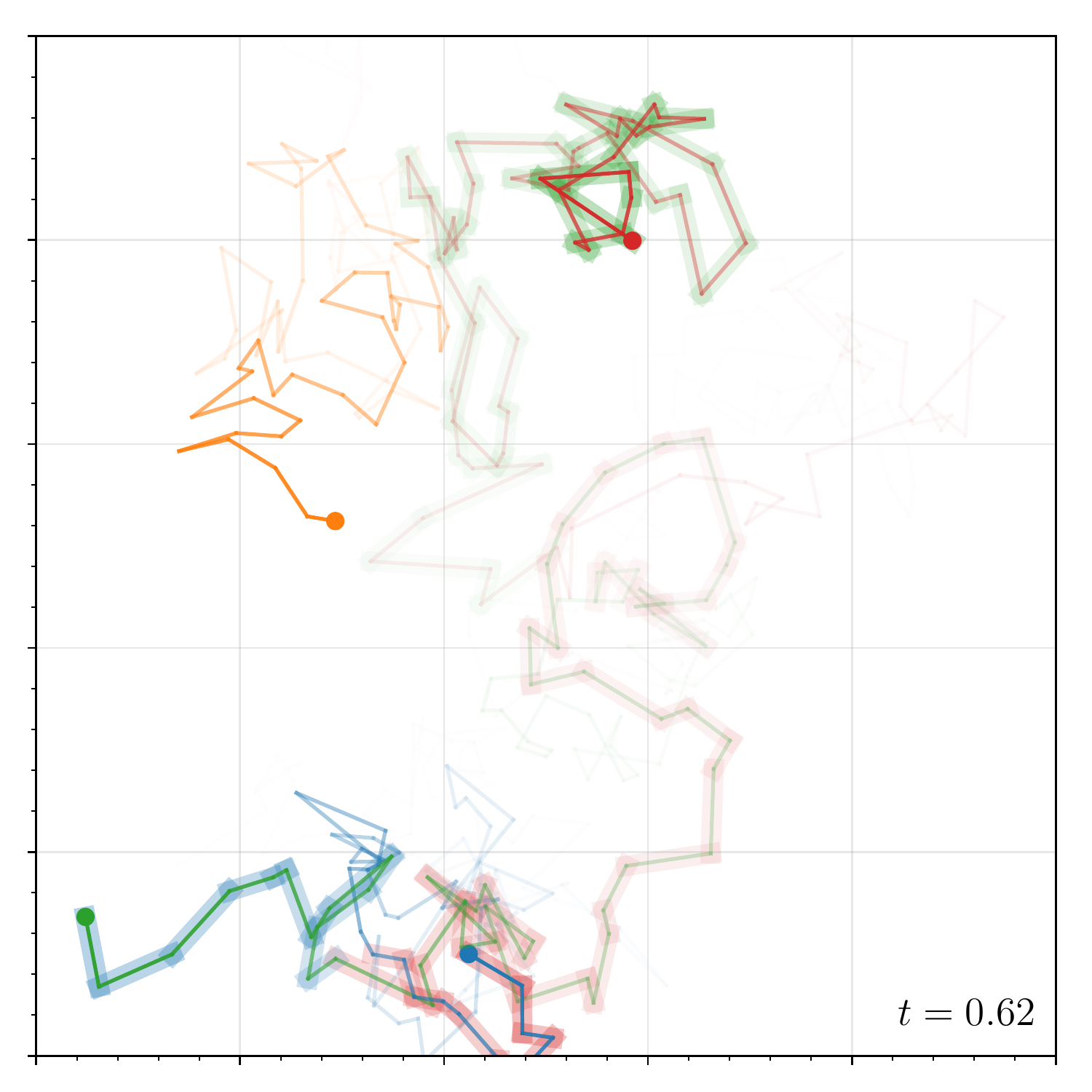}
    \end{tabular}
    \caption{\textbf{Online MLE tracking of Brownian motions.} We illustrate how errors accrue in tracking particles by iteratively computing the MLE. We plot the random walks formed by snapshots of four Brownian motions in $\RR^2$, and indicate by a circle two times when the permutation produced by the iterated MLE undergoes a transposition from the true labeling. For erroneously labelled points, we show their true label in the thin inner line, and their label by the iterated MLE in the thick outer line. If the points colored orange, red, green, and blue are respectively labelled $1, 2, 3, 4$ at the beginning, then the estimated permutation changes first to $1, 3, 2, 4$, and then to $1, 3, 4, 2$.} 
    \label{fig:tracking-schematic}
\end{figure}

\begin{figure}
    \hspace{-0.5cm}
    \begin{tabular}{c@{\hskip.0cm}c@{\hskip.0cm}c}
   \includegraphics[scale=0.55]{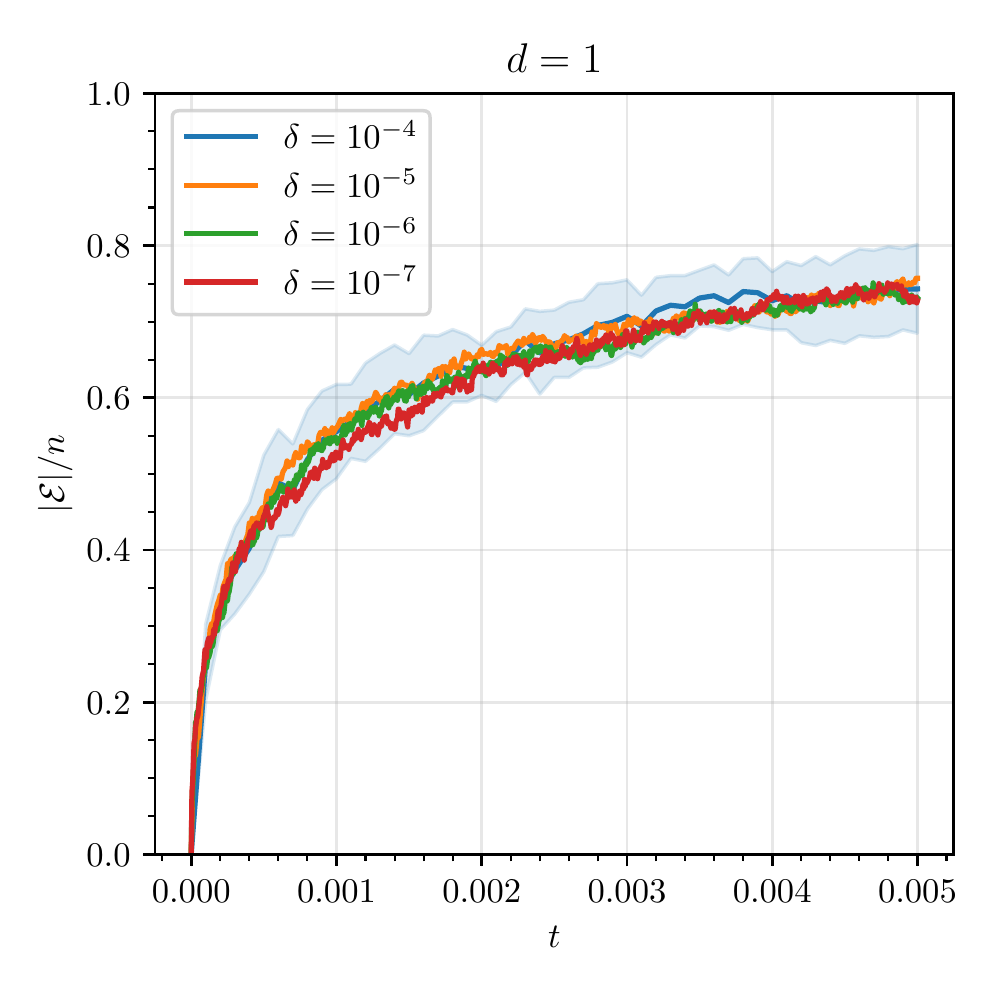} & \includegraphics[scale=0.55]{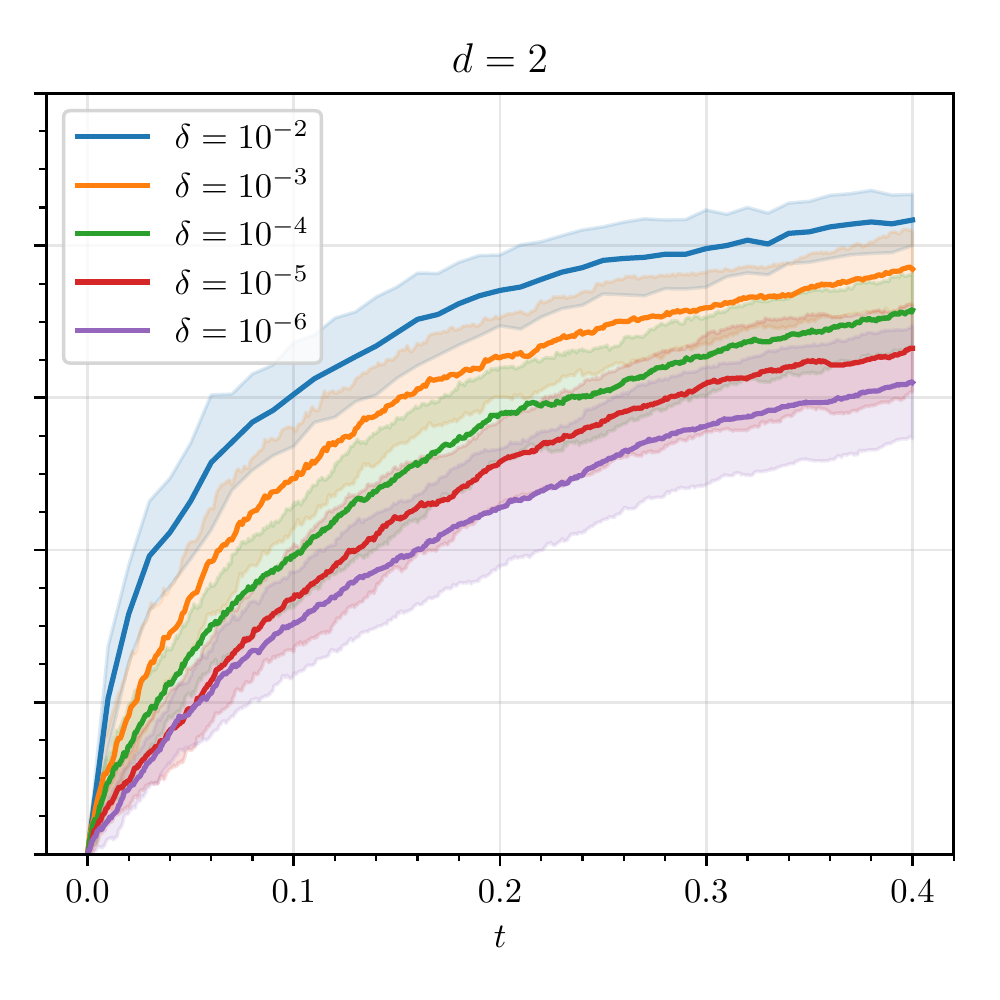} & \includegraphics[scale=0.55]{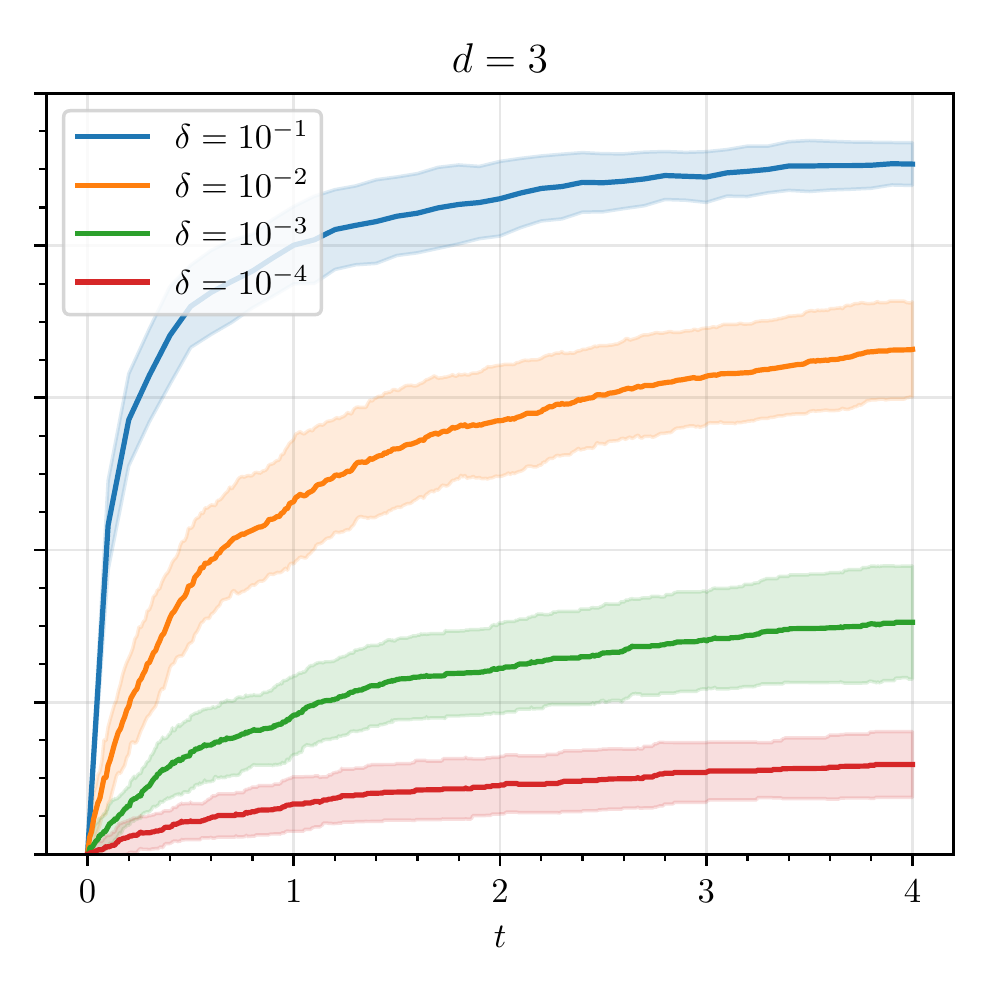}
    \end{tabular}
    \vspace{-0.5cm}
    \caption{\textbf{Dimension-dependent error scaling of MLE tracking.} We plot the error incurred by the iterated MLE estimator over time for tracking $n = 100$ independent Brownian motions in dimensions $d = 1, 2$, and 3, illustrating the differing dependences on the sampling interval $\delta$. Each curve plots an average of 20 independent trials and an error bar of one standard deviation.}
    \label{fig:tracking-error}
\end{figure}

 \begin{figure}
        \centering
        \includegraphics[scale=0.5]{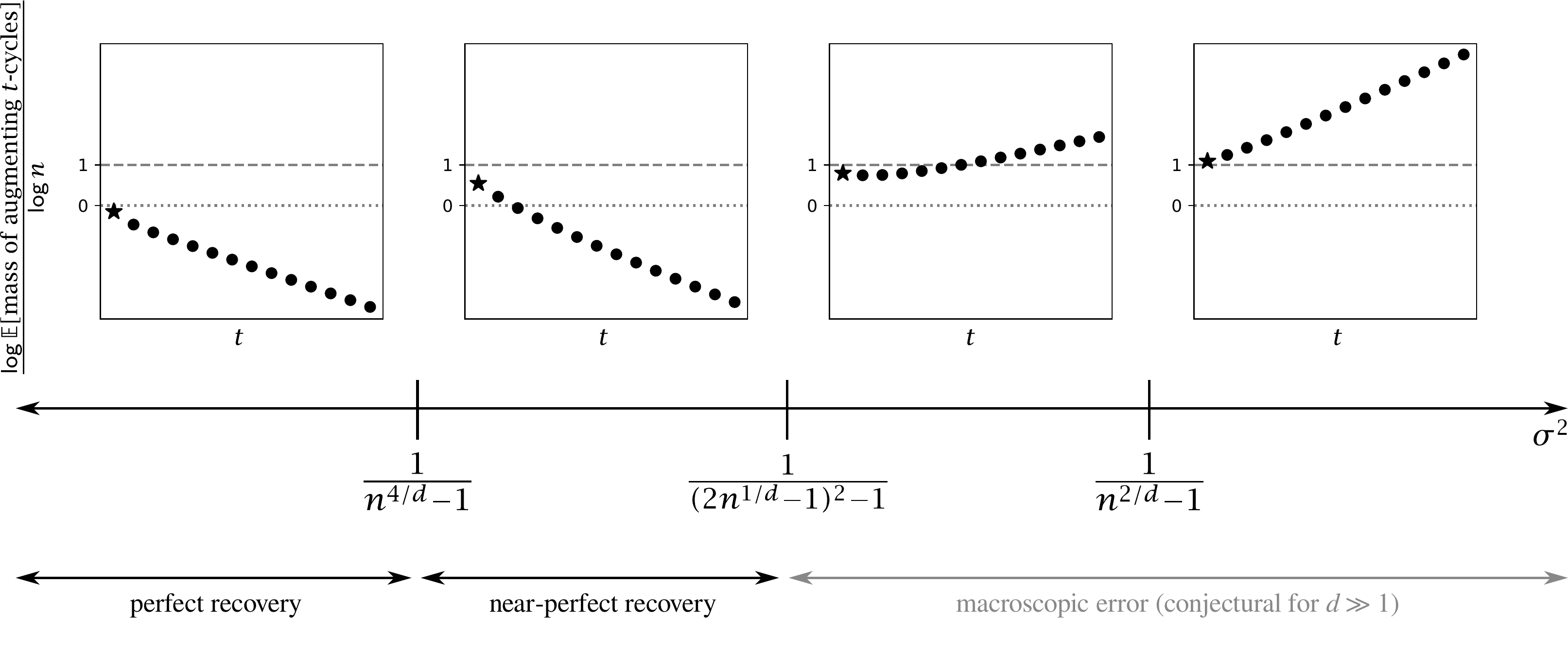}
        \caption{\textbf{First moments of augmenting cycle counts.} We illustrate our results and the associated thresholds, giving a schematic illustration of the polynomial rate of growth of the total mass of augmenting cycles of various sizes in each regime of the noise parameter $\sigma^2$. Regimes marked in black are those described by our results; the one in gray is conjectural. In each plot, a star marks the point plotting the expected mass of augmenting 2-cycles, whose analysis drives our lower bounds on $|\sE|$.}
        \label{fig:cycle-counts}
    \end{figure}

\begin{figure}
    \centering
    \includegraphics[scale=0.8]{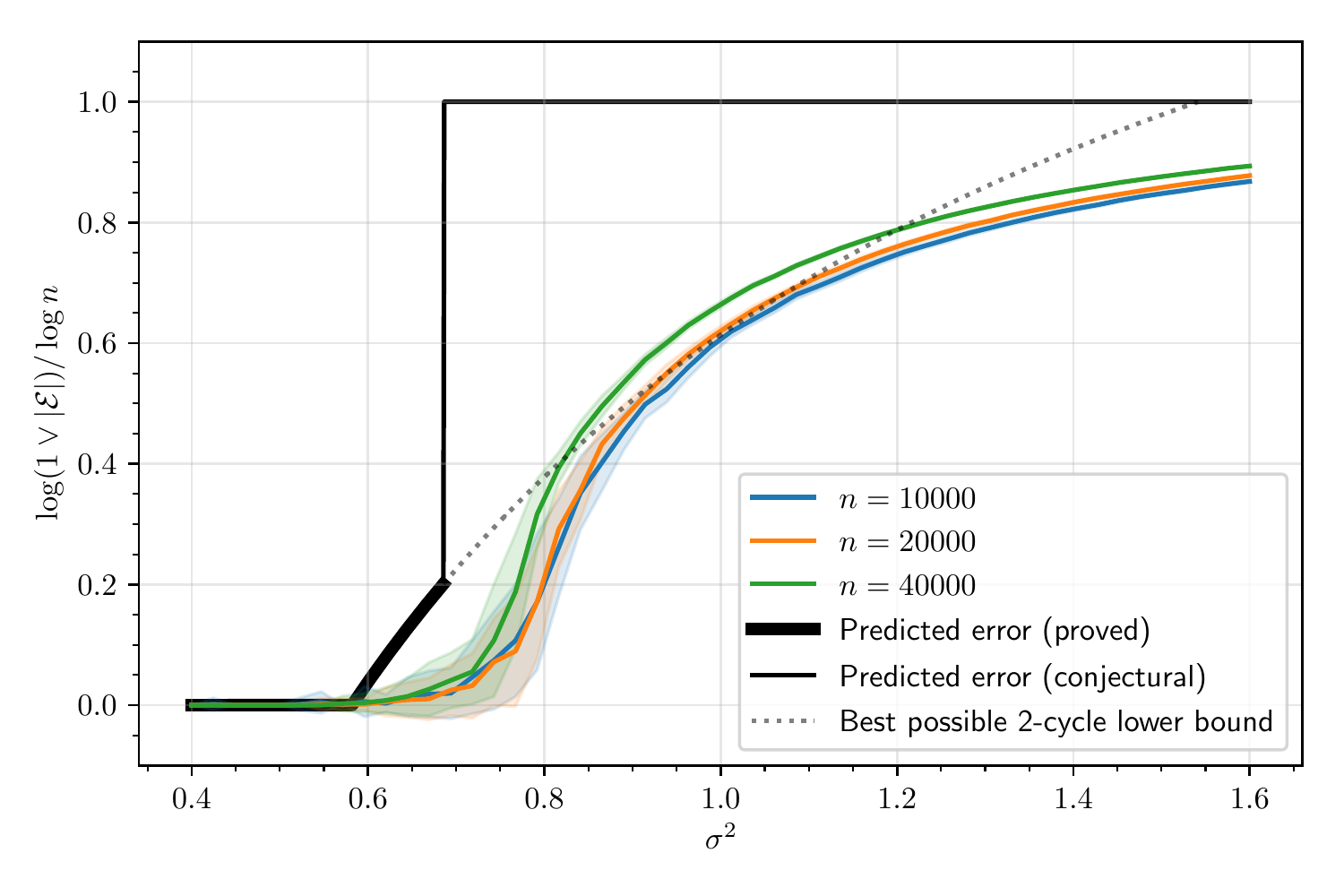}
    \vspace{-1em}
    \caption{\textbf{Discontinuity in polynomial error rate.} We show the predicted jump in the MLE error rate when $d = a\log n$ with $a = 4$ (bold solid line from Theorem~\ref{thm:d-theta-logn} and thin solid line from Conjecture~\ref{conj:non-recovery}) contrasted with the best possible lower bound that could be proved by analyzing only augmenting 2-cycles (dotted line). For increasing $n$, we also plot the average and one standard deviation error bars for 50 random trials of the MLE at regularly spaced $\sigma^2$. Though convergence is very slow with $n$, the fact that these curves cross the dotted line implies that there is non-trivial contribution to the total error from augmenting cycles of length greater than $2$, supporting Conjecture~\ref{conj:non-recovery} in the $d \sim \log n$ regime.}
    \label{fig:error-numerics}
\end{figure}

\subsection{Proof Techniques}
\label{sec:techniques}

We briefly discuss our proof techniques, with the aim of giving a heuristic theoretical justification of Conjecture~\ref{conj:non-recovery} above.
The following is the key structural property obeyed by $\sE$: the indices of $\sE$ belong to a disjoint union of cycles in $\what{\pi}$, and each such cycle $(i_1, \dots, i_t)$ is \emph{augmenting}, meaning that, performing index arithmetic modulo $t$,
\begin{equation}
    \sum_{k = 1}^t W_{i_k i_{k + 1}} \geq \sum_{k = 1}^t W_{i_k i_k},
\end{equation}
the reason being simply that the objective value of $\what{\pi}$ in \eqref{eq:W-mle} must not be increased by replacing any cycle of $\what{\pi}$ with the identity mapping.\footnote{Often the term ``augmenting cycle'' instead refers to an even cycle alternating between rows and columns of $\bW$, a cycle in the weighted bipartite graph on $2n$ vertices whose weights are given by $\bW$. However, we will find it more intuitive to think of cycles as permutations on $[n]$ instead, as described here.}
Our analysis is based on considering how many augmenting cycles of various sizes on $[n]$ exist.

There are $\binom{n}{t}(t - 1)! \approx n^t / t$ possible $t$-cycles on $[n]$ (the approximation holding for $t \ll n$), so the total ``mass'' or sum of the lengths of these cycles is $\approx n^t$.
We show that the probability that any given cycle is augmenting is related to the Riemann sum of a particular function $f(\sigma^2, x)$, thus obtaining that
\begin{align}
    \PP[t\text{-cycle is augmenting}] &\leq \exp\left(-\frac{d}{2}\sum_{j = 1}^{t - 1} f\left(\sigma^2, \frac{j}{t}\right)\right), \\
    \EE[\text{mass of augmenting } t\text{-cycles}] &\leq \exp\left(t\log n -\frac{d}{2}\sum_{j = 1}^{t - 1} f\left(\sigma^2, \frac{j}{t}\right)\right) \equalscolon n^{c(t)}.
\end{align}

We will show that these Riemann sums have a \emph{discrete concavity} property (see Section~\ref{sec:riemann-sums}), and that consequently $c(t)$ is a convex function of $t$, as we illustrate in Figure~\ref{fig:cycle-counts}.
The threshold that Conjecture~\ref{conj:non-recovery} predicts for strong recovery is the location where $\lim_{t \to \infty} c(t) / t$ changes sign from negative to positive, i.e. where the limiting slope of the curves in Figure~\ref{fig:cycle-counts} changes from negative to positive.

When this limiting slope is negative, then in fact the entire curve of $c(t)$ is decreasing, so the dominant contribution is made by augmenting 2-cycles.
In this case, we may analyze the number of errors the MLE makes by counting augmenting 2-cycles with the first and second moment methods.
When the limiting slope is positive, we expect substantial contributions to be made by $t$-cycles with large $t$, which our techniques here do not handle.
There is a third threshold when there are $\Omega(n)$ augmenting 2-cycles, the rightmost threshold in Figure~\ref{fig:cycle-counts}, beyond which in principle our second moment method might be improved to show that the MLE makes $\Omega(n)$ errors.
There are technical obstructions due to correlations in the second moment method that prevent us from carrying this out; moreover, as we emphasize in Figure~\ref{fig:error-numerics} for the case $d = \Theta(\log n)$, we do \emph{not} expect this analysis alone to prove the correct strong recovery threshold---for that, it appears necessary to argue the existence of larger augmenting cycles.

Finally, we remark that this latter threshold is a natural one for \emph{greedy algorithms} that attempt to find a good matching in the matrix $\bW$ row by row.
In Appendix~\ref{app:greedy}, we show that the greedy algorithm applied to $\bW$ in fact achieves strong recovery below this third threshold $\sigma^2 = \frac{1}{n^{2/d} - 1}$, which is asymptotically greater than the strong recovery threshold of the MLE $\sigma^2 = \frac{1}{(2n^{1/d} - 1)^2 - 1}$ once $d = \omega(\log n)$ (the former is $\sim \frac{1}{2}\frac{d}{\log n}$, while the latter is $\sim \frac{1}{4}\frac{d}{\log n}$).
On the other hand, this algorithm fails completely for $d = o(\log n)$; by contrast, a greedy algorithm applied to $\bW^{(0)}$ performs similarly to the MLE in that regime but can be worse outside the low-dimensional regime.
Across all $d = d(n)$ the three algorithms are generally incomparable.
We refer the reader to Appendix~\ref{app:greedy} for further discussion of these algorithms.

\begin{comment}
Carrying out the second moment method, we need to account for correlations between different matchings being augmenting.
That amounts to considering probabilities of subgraphs occurring in $G^{\aug}$ that consist of paths and cycles (the connected components that can arise in a union of two matchings) of various sizes.
Let us focus for the moment on the case of cycles occurring in this graph.
As we will show, the probability that a cycle of augmenting transpositions occurs is comparable to the probability that the cycle itself is augmenting.
That means that our analogy to an \Erdos-\Renyi\ graph is not entirely accurate: the probability that any given edge belongs to $G^{\aug}$ is roughly $\exp(-\frac{d}{2}f(\sigma^2, \frac{1}{2}))$, while the probability that any given $t$-cycle belongs to $G^{\aug}$ is, for large $t$, roughly $\exp(-\frac{d}{2}t\int_0^1 f(\sigma^2, x)dx)$, and we have $\int_0^1 f(\sigma^2, x)dx < f(\sigma^2, \frac{1}{2})$.
That is, while in an \Erdos-\Renyi\ graph only matchings sharing edges are correlated, in $G^{\aug}$ matchings sharing vertices are correlated as well, with the ``cost'' of each shared vertex relating to the difference $f(\sigma^2, \frac{1}{2}) - \int_0^1 f(\sigma^2, x)dx$ and ``amplified'' by $d$.
Once $d \gtrsim \log n$ this difference is amplified enough---the correlations between large matchings are great enough---that our second moment method no longer shows that $G^{\aug}$ contains a matching of nearly-linear size when $\sigma^2 > \frac{1}{(2n^{1/d} -1)^2 - 1}$.
\end{comment}

\subsection{Open Questions}

We conclude with several open questions on the estimation of geometric planted matchings that we find promising for future research.
\begin{enumerate}
    \item Establish the strong recovery threshold for $d \gg 1$, i.e., prove Conjecture~\ref{conj:non-recovery}.
    \item Establish the error curve for constant dimension $d$: what is the function $e(a, d)$ such that, when $\sigma^2 = an^{-2/d}$, then $\EE |\sE| / n \to e(a, d)$?
    \item Are algorithms other than the MLE (including the greedy algorithms we discuss in Appendix~\ref{app:greedy}, algorithms computing matchings corresponding to Wasserstein distances $W_p$ with $p \neq 2$, algorithms computing entropy-regularized relaxations of the linear assignment problem~\cite{Cuturi-2013-SinkhornLightspeed}, and the belief propagation algorithm proposed by \cite{CKKVZ-2010-ParticleTrackingBP}) more effective in certain regimes of $d$ and $\sigma^2$?
    \item Establish the dimension-dependent scaling of the time for which online MLE tracking can consistently track $n$ particles given in Conjecture~\ref{conj:tracking}, and determine what happens for small time intervals $\delta \ll n^{-4/d}$.
    \item More generally, what are effective algorithms for the motion tracking application proposed in Section~\ref{sec:tracking}? Is there an offline algorithm (processing the entire set of snapshots concurrently) that is superior to the kind of online algorithm we propose?
    \item What are the statistics of permutations obtained by computing optimal matchings between a collection of points and their evolution under Brownian motion for some period of time (either just once or with an iterated MLE or greedy algorithm)? How quickly do such permutations converge to the uniform distribution?
\end{enumerate}

\section{Preliminaries}

\subsection{Graph Laplacians and Spectra}

Given a graph $G = (V, E)$, we write $\bL^G \in \RR^{V \times V}$ for the \emph{graph Laplacian} of $G$, the symmetric matrix with quadratic form
\begin{equation}
    \bx^{\top}\bL^G \bx = \sum_{\{v, w\} \in E} (x_v - x_w)^2.
\end{equation}
We will particularly be interested in the path and cycle graphs.
We write $P_t$ and $C_t$ for the path or cycle, respectively, on $t$ vertices, where we require $t \geq 3$ for $C_t$ to be defined.
The following gives the spectra of their respective Laplacians (see, e.g., Example 8.8 for cycles and the discussion following Lemma 10.18 for paths in \cite{Nica-2016-SpectralGraphTheory}).
\begin{proposition}
    \label{prop:lap-evals}
    The eigenvalues of $\bL^{P_t}$ are $2(1 - \cos(\frac{\pi k}{t}))$ for $k = 0, \dots, t - 1$, and the eigenvalues of $\bL^{C_t}$ are $2(1 - \cos(\frac{2\pi k}{t})) = 4\sin^2(\frac{\pi k}{t})$ for $k = 0, \dots, t - 1$.
\end{proposition}

\subsection{Riemann Sums}
\label{sec:riemann-sums}

We have indicated in Section~\ref{sec:techniques}, and will see more precisely below, that probabilities of cycles being augmenting for the MLE give rise to expressions of the form $\Tr\log(1 + (4\sigma^2)^{-1} \bL^{C_t})$.
Per Proposition~\ref{prop:lap-evals}, these may in turn be viewed as Riemann sums of a certain periodic function, and the asymptotic probability of being augmenting for large cycles is therefore related to the integral of this function.
Below we set some notation for these objects and present the properties of theirs that we will use.
\begin{definition}
    For any $t \geq 2$, $\sigma^2 > 0$ define
    \begin{align}
        f(\sigma^2, x) &\colonequals \log\left(1 + \frac{1}{2\sigma^2}(1 - \cos(2\pi x))\right) = \log\left(1 + \frac{1}{\sigma^2}\sin^2(\pi x)\right), \\
        I(\sigma^2) &\colonequals \int_0^1 f(\sigma^2, x)\, dx, \\
        S(\sigma^2, t) &\colonequals \sum_{j = 1}^{t - 1}f\left(\sigma^2, \frac{j}{t}\right).
    \end{align}
\end{definition}
\noindent
In fact, it is possible to evaluate $I(\sigma^2)$ in closed form.
\begin{proposition}
    \label{prop:I-eval}
    For all $\sigma^2 > 0$,
    \begin{equation}
        I(\sigma^2) = 2\log\left(\frac{1 + \sqrt{1 + \sigma^{-2}}}{2}\right).
    \end{equation}
\end{proposition}
\noindent
We give the proof in Appendix~\ref{app:I} by translating the real integral to a complex contour integral.

By elementary real analysis, as $f(\sigma^2, \cdot)$ is continuous on $[0, 1]$, we have the following convergence.
\begin{proposition}
    For any $\sigma^2 > 0$, we have
    \begin{equation}
        \lim_{t \to \infty} \frac{S(\sigma^2, t)}{t} = I(\sigma^2).
    \end{equation}
\end{proposition}

We will, however, need to be substantially more precise for our applications.
The following are the main technical results that much of our analysis will rely on, a discrete analog of concavity for the Riemann sums of $f(\sigma^2, \cdot)$ as well as a matching opposite bound, which together allow us to formulate linear lower bounds on the $S(\sigma^2, t)$.
\begin{lemma}[Riemann sum discrete concavity]
    \label{lem:S-lower}
    For $\sigma^2 > 0$, $S(\sigma^2, t) - S(\sigma^2, t - 1)$ is strictly decreasing in $t \geq 3$ and approaches $I(\sigma^2)$ as $t \to \infty$.
    In particular, $S(\sigma^2, t) - S(\sigma^2, t - 1) > I(\sigma^2)$ for all $t \geq 3$.
\end{lemma}

\begin{lemma}[Riemann sum upper bound]
    \label{lem:S-upper}
    For $t \geq 2$ and $\sigma^2 > 0$, $S(\sigma^2, t) < tI(\sigma^2)$.
\end{lemma}

\begin{corollary}[Riemann sum lower bound]
    \label{cor:riemann-sum-cvx}
    For all $t_0 \geq 2$ and $t > t_0$, we have
    \begin{equation}
        S(\sigma^2, t) > S(\sigma^2, t_0) + (t - t_0)I(\sigma^2) = tI(\sigma^2) - (t_0 I(\sigma^2) - S(\sigma^2, t_0)),
    \end{equation}
    where the constant term satisfies $t_0 I(\sigma^2) - S(\sigma^2, t_0) > 0$.
\end{corollary}

\noindent
The third result follows immediately from the first two.
We give the proofs of the first two results in Appendix~\ref{app:riemann}.
The proofs rely on a combinatorial relationship between the sums $S(\sigma^2, t)$ and the \emph{Lucas polynomials}, which solve a Fibonacci-like recurrence that allows very precise asymptotics via a polynomial-valued analogue of Binet's formula.

\section{Upper Bounds and First Moment Method}

\subsection{Counting Augmenting Cycles}

To prove upper bounds on $|\sE|$, we use the first moment method and bound $\EE |\sE|$ by counting the numbers of augmenting cycles of various sizes.
First, we bound the probability that a cycle of a given size is augmenting.

\begin{proposition}
    \label{prop:cycle-aug-prob}
    Let $C$ be any fixed $t$-cycle in $[n]$.
    Then,
    \begin{equation}
        \PP[C \text{ is augmenting}] \leq \exp\left(-\frac{d}{2}S(\sigma^2, t)\right).
    \end{equation}
\end{proposition}
\begin{proof}
    Without loss of generality we may suppose that $C = (1, \dots, t)$.
    Let us consider the cases $t = 2$ and $t \geq 3$ separately.
    If $t = 2$, then $C$ is augmenting if and only if 
    \begin{equation}
        W_{1, 2} + W_{2, 1} \geq W_{1, 1} + W_{2, 2},
    \end{equation}
    which in turn holds if and only if
    \begin{equation}
        \langle \bz_1, \bx_2 - \bx_1 \rangle + \langle \bz_2, \bx_1 - \bx_2 \rangle \geq \|\bx_1 - \bx_2\|^2.
    \end{equation}
    Here, conditional on the $\bx_i$, the law of the left-hand side is $\sN(0, 2\sigma^2 \|\bx_1 - \bx_2\|^2)$ since $\bz_1$ and $\bz_2$ are i.i.d.\ with law $\sN(0, \sigma^2 \bm I_d)$.
    Therefore, we compute
    \begin{align*}
        \PP[C \text{ augmenting}]
        &= \Ex_{\bx_1, \bx_2} \Px_{g \sim \sN(0, 2\sigma^2 \|\bx_1 - \bx_2\|^2)}[g \geq \|\bx_1 - \bx_2\|^2] \\
        &= \Ex_{\bx_1, \bx_2} \Px_{g \sim \sN(0, 1)}\left[g \geq \sqrt{\frac{\|\bx_1 - \bx_2\|^2}{2\sigma^2}}\right] \\
        &\leq \Ex_{\bx_1, \bx_2} \exp\left(-\frac{\|\bx_1 - \bx_2\|^2}{4\sigma^2}\right)
        \intertext{To evaluate the remaining expectation, we must understand the spectrum of the quadratic form involved.
        Writing $\bx$ for the concatenation of $\bx_1$ and $\bx_2$, we may write $\|\bx_1 - \bx_2\|^2 = \bx^{\top} (\bL^{P_2} \otimes \bm I_d) \bx$, where $\bL^{P_2} \in \RR^{2 \times 2}$ is the Laplacian of the path graph on two vertices, using the notation of Proposition~\ref{prop:lap-evals}.
        By the Proposition, the eigenvalues of $\bL^{P_2}$ are 0 and 2.
    Therefore, continuing by applying an orthogonal change of basis diagonalizing the quadratic form and evaluating the $\chi^2$ moment generating function that appears, we find}
    &= \det\left( \bm I_{2d} + \frac{1}{2\sigma^2}(\bL^{P_2} \otimes \bm I_d)\right)^{-1/2} \\
    &= \det\left( \bm I_{2} + \frac{1}{2\sigma^2}\bL^{P_2}\right)^{-d/2} \\
    &= \left(1 + \frac{1}{\sigma^2}\right)^{-d/2} \\
    &= \exp\left(-\frac{d}{2}\log\left(1 + \frac{1}{\sigma^2}\right)\right) \\
    &= \exp\left(-\frac{d}{2}S(\sigma^2, 2)\right), \numberthis
    \end{align*}
    as claimed.
    
    Now, suppose $t \geq 3$.
    Then $C$ is augmenting if and only if
    \begin{equation}
        W_{t, 1} + \sum_{i = 1}^{t - 1} W_{i,i + 1} \geq \sum_{i = 1}^t W_{i,i},
    \end{equation}
    which in turn holds if and only if
    \begin{equation}
        \la \bz_1, \bx_t - \bx_1 \ra + \sum_{i = 2}^{t} \la \bz_i, \bx_{i - 1} - \bx_{i} \ra \geq \frac{1}{2}\bigg(\|\bx_t - \bx_1\|_2^2 + \sum_{i = 2}^t \|\bx_{i - 1} - \bx_{i}\|_2^2\bigg).
    \end{equation}
    Again, let $\bx$ be the concatenation of the $\bx_i$.
    Then, we have
    \begin{equation}
        \|\bx_t - \bx_1\|_2^2 + \sum_{i = 2}^t \|\bx_{i - 1} - \bx_{i}\|_2^2 = \bx^{\top} (\bL^{C_t} \otimes \bm I_d) \bx,
    \end{equation}
    where $\bL^{C_t}$ is the Laplacian of the cycle graph $C_t$ on $t$ vertices.
    Thus the law of the left-hand side above conditional on the $\bx_i$ is $\sN(0, \sigma^2 \bx^{\top} (\bL^{C_t} \otimes \bm I_d) \bx)$, while the right-hand side is $\frac{1}{2} \bx^{\top} (\bL^{C_t} \otimes \bm I_d) \bx$.
    (We note the two differences from the case $t = 2$: the path graph is replaced by the cycle graph, and an extra factor of $\frac{1}{2}$ appears on the right-hand side.)
    An analogous computation to before gives
    \begin{align*}
    \PP\left[C \text{ augmenting}\right] 
    &= \Ex_{\bx_1, \dots, \bx_t} \Px_{g \sim \sN(0, \sigma^2 \bx^{\top} (\bL^{C_t} \otimes \bm I_d) \bx)}\left[g \geq \frac{\bx^{\top} (\bL^{C_t} \otimes \bm I_d) \bx}{2} \right] \\
    &= \Ex_{\bx_1, \dots, \bx_t} \Px_{g \sim \sN(0, 1)}\left[g \geq \sqrt{\frac{\bx^{\top} (\bL^{C_t} \otimes \bm I_d) \bx}{4\sigma^2}} \, \right] \\
    &\leq \Ex_{\bx_1, \dots, \bx_t} \exp\left(-\frac{\bx^{\top} (\bL^{C_t} \otimes \bm I_d) \bx}{8\sigma^2}\right) \\
    &= \det\left( \bm I_{dt} + \frac{1}{4\sigma^2}\bL^{C_t} \otimes \bm I_d\right)^{-1/2} \\
    &= \det\left( \bm I_{t} + \frac{1}{4\sigma^2}\bL^{C_t}\right)^{-d/2}
    \intertext{and substituting in the eigenvalues of $\bL$ from Proposition~\ref{prop:lap-evals}, we have}
    &= \left(\prod_{j = 0}^{t - 1}\left\{1 + \frac{1}{2\sigma^2}\left(1 - \cos\left(\frac{2\pi j}{t}\right)\right)\right\}\right)^{-d/2} \\
    &= \exp\left(-\frac{d}{2}\sum_{j = 0}^{t - 1}\log\left(1 + \frac{1}{2\sigma^2}\left(1 - \cos\left(\frac{2\pi j}{t}\right)\right)\right)\right) \\
    &= \exp\left(-\frac{d}{2}S(\sigma^2, t)\right), \numberthis
    \end{align*}
    again giving the result.
\end{proof} 
%We remark that both cases $t = 2$ or $t \geq 3$ yield the same type of Riemann sum $S(\sigma^2, \cdot)$ only thanks to a rather subtle coincidence: the extra factor of 2 appearing in the expressions governing a cycle's being augmenting when $t = 2$ exactly cancels the factor of 2 by which the greatest eigenvalue of the path of length 2 (equaling 2) and any even cycle (equaling 4) differ.

\begin{corollary}
    \label{cor:first-moment}
    For any $d, n, \sigma^2$,
    \begin{equation}
        \EE |\sE| \leq \sum_{t = 2}^n \exp\left(t\log n - \frac{d}{2}S(\sigma^2, t)\right).
    \end{equation}
\end{corollary}
\begin{proof}
    $\sE$ is a disjoint union of augmenting cycles, so $|\sE|$ is at most the sum of the lengths of all augmenting cycles.
    The result then follows from linearity of expectation and applying that the number of $t$-cycles in $[n]$ is $\leq n^t / t$ and the probability bound of Proposition~\ref{prop:cycle-aug-prob}.
\end{proof}

With these expressions for the expected masses of augmenting cycles of various sizes in hand, we may describe more precisely why the situation presented in Figure~\ref{fig:cycle-counts} arises: the limiting exponent above as $t \to \infty$ is $\sim t\log n(1 - \frac{d}{2\log n} I(\sigma^2))$, thus the transition around $I(\sigma^2) = 2\log(\frac{1 + \sqrt{1 + \sigma^{-2}}}{2}) = \frac{2\log n}{d}$, or $\sigma^2 = \frac{1}{(2n^{1/d} - 1)^2 - 1}$, determines whether the expected mass of large augmenting cycles diverges or not, which we conjecture is the correct strong recovery threshold.
Moreover, it will turn out that when strong recovery is possible, then the dominant contribution is by augmenting 2-cycles, whose exponent is $2 - \frac{d}{2}S(\sigma^2, 2) = 2\log n - \frac{d}{2}\log(1 + \sigma^{-2})$, and this changes sign at $\sigma^2 = \frac{1}{n^{4/d} - 1}$, which is the perfect recovery threshold.

\subsection{Perfect Recovery}

In this section we give a sufficient condition for perfect recovery, which proves Part 1 of Theorem~\ref{thm:d-ll-logn}, Part 1 of Theorem~\ref{thm:d-theta-logn}, and Theorem~\ref{thm:d-omega-logn}.

\begin{lemma}
    Let $s_0 \colonequals 2^{1/d}$, and suppose that
    \begin{equation}
        \sigma^2 \leq \frac{1}{s_0^{\omega(1)}n^{4/d} - 1}.
    \end{equation}
    Then, $\EE |\sE| \to 0$, so, in particular, $|\sE| = 0$ with high probability.
\end{lemma}
\noindent
Before proceeding with the proof, let us indicate how this implies the claimed results for specific scalings of $d$.
When $d \ll \log n$ is constant, then $s_0$ is bounded and the denominator in the bound above goes to infinity as $n \to \infty$, so the condition is satisfied whenever $\sigma^2 \ll n^{-4/d}$, giving Part 1 of Theorem~\ref{thm:d-ll-logn}.

When $d = a \log n$, then $n^{4/d} = e^{4/a}$, and there exists $f(n) = \omega(1)$ such that $s_0^{f(n)} \to 1$.
Thus the condition is satisfied whenever $\sigma^2$ is bounded below $\frac{1}{e^{4/a} - 1}$, giving Part~1 of Theorem~\ref{thm:d-theta-logn}.

Finally, when $d = \omega(\log n)$, then for any $\epsilon > 0$ again we may choose $f(n) = \omega(1)$ such that $s_0^{f(n)} = 2^{f(n) / d} \leq n^{\epsilon / d}$.
Thus the condition is satisfied whenever $\sigma^2 \leq \frac{1}{n^{(4 + \epsilon) / d} - 1} \sim \frac{1}{4 + \epsilon} \frac{d}{\log n}$, giving Theorem~\ref{thm:d-omega-logn}.

\begin{proof}
Rearranging the assumption on $\sigma^2$, we have
\begin{equation}
    2 - \frac{d}{2\log n}S(\sigma^2, 2) =2 - \frac{d \log(1 + \sigma^{-2})}{2\log n} \leq -\omega\left(\frac{d\log s_0}{\log n}\right) = -\omega\left(\frac{1}{\log n}\right).
\end{equation}
Also, since by Lemma~\ref{lem:S-upper} we have $S(\sigma^2, 2) < 2I(\sigma^2)$, we further have
\begin{equation}
    2 - \frac{d}{2\log n}S(\sigma^2, 2) > 2 - \frac{d}{2\log n}2I(\sigma^2) = 2\left(1 - \frac{d}{2\log n}I(\sigma^2)\right).
\end{equation}
Towards bounding the exponents appearing in Corollary~\ref{cor:first-moment}, we manipulate
\begin{align*}
    t\log n - \frac{d}{2}S(\sigma^2, t)
    &= \log n\left(t - \frac{d}{2\log n}S(\sigma^2, t)\right) \\
    &\leq \log n\left(t - \frac{d}{2\log n}(S(\sigma^2, 2) + (t - 2)I(\sigma^2))\right) \tag{by Corollary~\ref{cor:riemann-sum-cvx} with $t_0 = 2$} \\
    &= \log n\left(2 - \frac{d}{2\log n}S(\sigma^2, 2) + (t - 2)\left(1 - \frac{d}{2\log n}I(\sigma^2)\right)\right)
    \intertext{and substituting in our bounds from above,}
    &\leq t\log n\left(1 - \frac{d}{2\log n}I(\sigma^2)\right) \\
    &\leq - \omega(1) \cdot t . \numberthis
\end{align*}
Applying this to Corollary~\ref{cor:first-moment}, we find
\begin{equation}
    \EE |\sE| \leq \sum_{t = 2}^n (e^{-\omega(1)})^t = o(1),
\end{equation}
and the second result follows by Markov's inequality.
\end{proof}

\subsection{Constant Error Upper Bound}

We next prove a similar result to the above that gives Part~2 of Theorem~\ref{thm:d-ll-logn} and the upper bound for the case of Part~2 of Theorem~\ref{thm:d-theta-logn} where $\sigma^2$ takes its lower bound, $\sigma^2 = \frac{1}{e^{4/a} - 1}$.

\begin{lemma}
    Let $s_0 \colonequals 2^{1/d}$, and suppose that
    \begin{equation}
        \sigma^2 \leq \frac{1}{s_0^{O(1)}n^{4/d} - 1}.
    \end{equation}
    Then, $\EE |\sE| = O(1)$, so, in particular, for any $f(n) = \omega(1)$, we have $|\sE| \leq f(n)$ with high probability.
\end{lemma}
\noindent
The argument from the previous proof applies verbatim with $\omega(\cdot)$ replaced by $O(\cdot)$ throughout, and shows that $\EE |\sE| = O(1)$, whereby the result again follows by Markov's inequality.

\subsection{Sublinear Error Upper Bound}

Finally we give an upper bound on $|\sE|$ that holds in the sublinear error regime. 
This implies the upper bound of Part~3 of Theorem~\ref{thm:d-ll-logn} and the remainder of the upper bound of Part~2 of Theorem~\ref{thm:d-theta-logn} not covered by the previous proof.

\begin{lemma}
    Let $s_0 \colonequals 2^{1/d}$, and suppose that
    \begin{equation}
    \sigma^2 \leq \frac{1}{(2s_0^{\omega(1)}n^{1 / d} - 1)^2 - 1}.
    \end{equation}
    Then,
    \begin{equation}
        \EE |\sE| = O\left(\left(1 + \frac{1}{\sigma^2}\right)^{-d/2} n^2\right),
    \end{equation}
    so in particular for any $f(n) = \omega(1)$ we have, with high probability,
    \begin{equation}
        |\sE| \leq f(n)\left(1 + \frac{1}{\sigma^2}\right)^{-d/2} n^2.
    \end{equation}
\end{lemma}
\begin{proof}
    Rearranging the assumption on $\sigma^2$, we have
\begin{equation}
    1 - \frac{d}{2\log n}I(\sigma^2) = 1 - \frac{d}{\log n}\log\left(\frac{1 +  \sqrt{1 + \sigma^{-2}}}{2}\right) = -\omega\left(\frac{1}{\log n}\right),
\end{equation}
as before (the difference with the above settings being that such a bound no longer holds for $2 - \frac{d}{2\log n}S(\sigma^2, 2)$).
Following the previous argument applied to Corollary~\ref{cor:first-moment}, we find
\begin{equation}
    \EE |\sE| \leq \sum_{t = 2}^n n^{2 - \frac{d}{2\log n}S(\sigma^2, 2)} (e^{-\omega(1)})^{t - 2} = O(n^{2 - \frac{d}{2\log n}S(\sigma^2, 2)}) = O\left(\left(1 + \frac{1}{\sigma^2}\right)^{-d/2} n^2\right),
\end{equation}
and the second result again follows by Markov's inequality.
\end{proof}

\section{Lower Bounds and Second Moment Method}

To prove lower bounds on $|\sE|$, we will apply the second moment method to show that there exists a large number of vertex-disjoint augmenting 2-cycles.
That is, we will study the random variable
\begin{equation}
    M \colonequals \text{maximum number of vertex-disjoint augmenting 2-cycles in } [n].
\end{equation}
The following shows that $M$ being large guarantees a large number of errors in the MLE.
\begin{proposition}
    $|\sE| \geq M$.
\end{proposition}
\begin{proof}
    It is impossible for $(i, j)$ to be an augmenting transposition and to have both $\what{\pi}(i) = i$ and $\what{\pi}(j) = j$, since then $\pi$ formed by composing the transposition $(i, j)$ with $\what{\pi}$ would have a higher likelihood than $\what{\pi}$.
    Thus, for every pair in a maximal collection of $M$ augmenting 2-cycles, at least one of its vertices must be labelled incorrectly by $\what{\pi}$, and the result follows.
\end{proof}

Conveniently, this quantity admits a graph-theoretic interpretation.
Namely, the set of augmenting 2-cycles may be described by a graph on $[n]$:
\begin{equation}
    G^{\aug} \colonequals (V = [n], E = \{\{i, j\}: (i, j) \text{ is an augmenting 2-cycle}\}).
\end{equation}
With this notation, $M$ is the size of the largest matching in this graph:
\begin{equation}
    M = \text{number of edges in the largest matching in } G^{\aug}.
\end{equation}
Thus our task is to show that a large matching exists in a random graph.
In particular, we will want to show that there exists a matching of size $\Omega(|E| \wedge n)$, i.e., a matching of size asymptotically as large as possible subject to the basic constraints that it can exceed neither the number of vertices nor the number of edges.

There is an extensive literature on similar questions for \Erdos-\Renyi\ (ER) random graphs; however, most of these results analyze concrete algorithms for finding large matchings rather than using the second moment method \cite{KS-1981-MathingSparseRandomGraphs,DFP-1993-AverageGreedyMatching,AFP-1998-KarpSipserRevisited,ZM-2006-MatchingsRandomGraphs}.
Indeed, to the best of our knowledge no previous work has tried to show the existence of large matchings in random graphs using the second moment method---perhaps thanks to the success of analyzing algorithms and to the ``effectiveness'' of such results, which provide an algorithm in addition to an existence proof.
However, our graph $G^{\aug}$ has a more complicated dependence structure, so the second moment method is more convenient, and we draw inspiration from a line of work applying an adjusted second moment method to other extremal problems in ER random graphs, especially the chromatic number and independence number \cite{SS-1987-SharpConcentrationChromaticNumber,Frieze-1990-IndependenceNumberRandomGraphs,McDiarmid-1989-BoundedDifferences}.

\begin{remark}
    When the degree of all vertices in $G^{\aug}$ is bounded with high probability by some $d_{\max}$, then algorithmic techniques do show that a large matching exists, since a greedy algorithm matching vertices arbitrarily until no more can be matched will produce a matching of at least $|E| / 2d_{\max}$ edges.
    One may control the maximum degree in our case by appealing to the probability bounds of Proposition~\ref{prop:aug-per-edge-general} for the star graph.
    However, this no longer applies in the critical regime where the average degree is constant (when we expect a nearly-linear number of errors in the MLE), in which case in an ER graph the largest degree is of logarithmic order, and we expect a similar behavior for $G^{\aug}$.
\end{remark}

\subsection{Statistics of \texorpdfstring{$G^{\protect\aug}$}{Gaug}}

We will think of $G^{\aug}$ as being well-approximated by an ER random graph, albeit with some stronger dependencies among various subgraphs.
We begin by precisely describing the probability of any particular edge belonging to $G^{\aug}$, which is the edge probability of the analogous ER graph.

\begin{proposition}[Edge probability in $G^{\aug}$]
    \label{prop:pair-aug-lb}
    Define
    \begin{align}
        p &\colonequals \PP[\{i, j\} \in E(G^{\aug})], \\
        \what{p} &\colonequals \frac{p}{\exp(-\frac{d}{2}S(\sigma^2, 2))},
    \end{align}
    which do not depend on $i, j \in [n]$ distinct.
    Then, for all $n$, $d$, and $\sigma^2 \leq \frac{1}{40}d$, 
    \begin{equation}
        \frac{1}{1000}\sqrt{\frac{1 + \sigma^2}{d}} \leq \what{p} \leq 1.
    \end{equation}
\end{proposition}
\noindent
We give the proof, an application of bounds on Gaussian Mills' ratios, in Appendix~\ref{app:pf:prop:pair-aug-lb}.

Next, we control more coarsely the probability that a given graph occurs as a subgraph of $G^{\aug}$.
The following is a general parametrized bound, which relates these probabilities to Laplacians with weighted edges.
\begin{proposition}
    \label{prop:aug-per-edge-general}
   % If $G$ and $H$ are vertex-disjoint graphs on subsets of $[n]$, then the events $\{G \subseteq G^{\aug}\}$ and $\{H \subseteq G^{\aug}\}$ are independent.
    Suppose $G = (V, E)$ for some $V \subseteq [n]$.
    Let $\bm\Delta \in \RR^{E \times V}$ be the edge-vertex incidence matrix for $G$, i.e., the matrix having non-zero entries $\Delta_{\{i, j\}, k}$ only when $i = k$ or $j = k$, with one of these equaling 1 and the other equaling $-1$ (chosen arbitrarily) for each row index $\{i, j\}$.
    Note that $\bm\Delta^{\top}\bm\Delta = \bm L$, the graph Laplacian.
    Then, for any diagonal matrix $\bD \succeq \bm 0$,
    \begin{align*}
        \PP[G \subseteq G^{\aug}]
        &\leq \det\left( \bm I_V + 2\bm\Delta^{\top}\bD \bm\Delta - \sigma^2(\bm\Delta^{\top}\bD \bm\Delta)^2\right)^{-d / 2} \\
        &= \exp\left(-\frac{d}{2}\sum_{i = 1}^{|V|} \log\left(1 + 2\lambda_i(\bm\Delta^{\top}\bD \bm\Delta) - \sigma^2 \lambda_i(\bm\Delta^{\top}\bD \bm\Delta)^2\right)\right). \numberthis
    \end{align*}
\end{proposition}
\begin{proof}
    The event that $G \subseteq G^{\aug}$ is the same as that, for all $\{i, j\} \in E$, we have $W_{i, j} + W_{j, i} \leq W_{i, i} + W_{j, j}$.
    Rewriting, this is the event that, for all $\{i, j\} \in E$,
    \begin{equation}
        -\langle \bz_i - \bz_j, \bx_i - \bx_j \rangle \geq \|\bx_i - \bx_j \|^2. 
    \end{equation}
    Let $\bX, \bZ \in \RR^{V \times d}$ have the $\bx_i$ and the $-\bz_i$ as their rows, respectively.
    Then, the system above may be rewritten with the help of $\bm\Delta$ as
    \begin{equation}
        \diag(\bm\Delta \bZ (\bm\Delta \bX)^{\top}) \geq \diag(\bm \Delta \bX(\bm \Delta \bX)^{\top}).
    \end{equation}
    Whenever this is true, then we also have
    \begin{equation}
        \langle \bD, \bm\Delta \bZ (\bm\Delta \bX)^{\top} \rangle \geq \langle \bD, \bm \Delta \bX(\bm \Delta \bX)^{\top}\rangle,
    \end{equation}
    or, rewriting to isolate $\bZ$,
    \begin{equation}
        \langle \bZ, \bm\Delta^{\top}\bD \bm\Delta \bX \rangle \geq \langle \bX\bX^{\top}, \bm\Delta^{\top}\bD \bm\Delta \rangle.
    \end{equation}
    
    Since the entries of $\bZ$ are i.i.d.\ with law $\sN(0, \sigma^2)$, taking a Chernoff bound and evaluating the Gaussian moment generating function yields
    \begin{align*}
        \PP[G \subseteq G^{\aug}]
        &\leq \EE_{\bX}\frac{\EE_{\bZ} \exp\left(\langle \bZ, \bm\Delta^{\top}\bD \bm\Delta \bX \rangle\right)}{\exp\left(\langle \bX\bX^{\top}, \bm\Delta^{\top}\bD \bm\Delta \rangle\right)} \\
        &= \EE_{\bX}\exp\left(\frac{\sigma^2}{2}\|\bm\Delta^{\top}\bD \bm\Delta \bX\|_F^2 - \langle \bX\bX^{\top}, \bm\Delta^{\top}\bD \bm\Delta \rangle\right)
        \intertext{and, noting that $\|\bm\Delta^{\top}\bD \bm\Delta \bX\|_F^2 = \Tr(\bX^{\top}(\bm\Delta^{\top} \bD \bm\Delta)^2 \bX) = \langle \bX\bX^{\top}, (\bm\Delta^{\top} \bD \bm\Delta)^2\rangle$, we find}
        &= \EE_{\bX}\exp\left(\left\langle\bX\bX^{\top}, \frac{\sigma^2}{2}(\bm\Delta^{\top} \bD \bm\Delta)^2 - \bm\Delta^{\top}\bD\bm\Delta\right\rangle\right)
        \intertext{and evaluating this as a $\chi^2$ moment generating function after an orthogonal change of basis diagonalizing the matrix on the right, we obtain}
        &= \det\left( \bm I_V + 2\bm\Delta^{\top}\bD \bm\Delta - \sigma^2(\bm\Delta^{\top}\bD \bm\Delta)^2\right)^{-d / 2}, \numberthis
    \end{align*}
    as claimed.
\end{proof}
\noindent
It is an interesting question to optimize the choice of $\bD$ in this bound.
For our purposes, it suffices to use a simple version for $G$ a path or cycle.
\begin{proposition}
    \label{prop:aug-per-edge-path-cycle}
    For any $G = P_t$ with $t \geq 2$ or $G = C_t$ with $t \geq 3$,
    \begin{equation}
        \PP[G \subseteq G^{\aug}] \leq \exp\left(-\frac{d}{2}S(\sigma^2, t)\right).
    \end{equation}
\end{proposition}
\noindent
In words, this shows that the probability that a path or cycle in $G^{\aug}$ on $t$ vertices has augmenting 2-cycles for all of its edges is at most our bound (Proposition~\ref{prop:cycle-aug-prob}) on the probability that a cycle on $t$ vertices is augmenting.
\begin{proof}
    For $G = P_2$ the result follows from Proposition~\ref{prop:cycle-aug-prob}.
    We first note that, since $P_t$ is a subgraph of $C_t$, $\PP[C_t \subseteq G^{\aug}] \leq \PP[P_t \subseteq G^{\aug}]$ for all $t \geq 3$ (since the event that $P_t \subseteq G^{\aug}$ contains the event that $C_t \subseteq G^{\aug}$ for suitable labellings of the two graphs), so it suffices to consider $G = P_t$.
    For this case, we choose $\bD = \frac{1}{2\sigma^2}\bm I_{t - 1}$ in Proposition~\ref{prop:aug-per-edge-general}.
    That gives
    \begin{align*}
        \PP[ P_t \subseteq G^{\aug}]
        &\leq \exp\left(-\frac{d}{2}\sum_{i = 1}^{t} \log\left(1 + \frac{1}{\sigma^2}\lambda_i(\bL^{P_t}) - \frac{1}{4\sigma^2} \lambda_i(\bL^{P_t})^2\right)\right) \\
        &= \exp\left(-\frac{d}{2}\sum_{k = 1}^{t - 1} \log\left(1 + \frac{2}{\sigma^2}\left(1 - \cos\left(\frac{\pi k}{t}\right)\right) - \frac{1}{\sigma^2} \left(1 - \cos\left(\frac{\pi k}{t}\right)\right)^2\right)\right) \\
        &= \exp\left(-\frac{d}{2}\sum_{k = 1}^{t - 1} \log\left(1 + \frac{1}{\sigma^2}\sin^2\left(\frac{\pi k}{t}\right)\right)\right) \\
        &= \exp\left(-\frac{d}{2}S(\sigma^2, t)\right), \numberthis
    \end{align*}
    completing the proof.
\end{proof}

\begin{remark}
    While this approach to bounding $\PP[G \subseteq G^{\aug}]$ may seem rather naive, there is reason to believe it is close to optimal up to constant factors in $\sigma^2$: we know from the proof of Proposition~\ref{prop:cycle-aug-prob} for $t = 2$ that $\PP[\{i, j\} \in E(G^{\aug}) \mid \bx_i, \bx_j] \approx \exp(-\frac{1}{4\sigma^2} \|\bx_i - \bx_j\|^2)$, so if we heuristically suppose that the edges of $G^{\aug}$ occur independently conditional on the $\bx_i$, then we find
    \begin{align*}
        \PP[G \subseteq G^{\aug}] 
        &\approx \Ex_{\bx_i} \prod_{\{i, j\} \in E(G)} \PP[\{i, j\} \in E(G^{\aug}) \mid \bx_i, \bx_j] \\
        &\approx \Ex_{\bx_i} \exp\left(-\frac{1}{4\sigma^2} \bx^{\top} (\bL^G \otimes \bm I_d) \bx\right) \\
        &= \det\left(\bm I_V + \frac{1}{2\sigma^2} \bL^G\right)^{-d/2}
        \intertext{and if, for instance, $G = C_t$ then following the computations in Proposition~\ref{prop:cycle-aug-prob} for $t \geq 3$ we would find}
        &= \exp\left(-\frac{d}{2}S\left(\frac{\sigma^2}{2}, t\right)\right), \numberthis
    \end{align*}
    differing only by a factor of 2 in $\sigma^2$ from the bound of Proposition~\ref{prop:aug-per-edge-path-cycle}.
\end{remark}

\subsection{Concentration-Enhanced Second Moment Method}
\label{sec:second-moment}

We next review a version of the second moment method that can sometimes improve a weak result of the ordinary method---showing an object exists with quite low probability---to a strong result with high probability by combining it with a concentration inequality.
Below, Part~(b) is the typical result of a second moment method that has not succeeded in showing that a random variable is positive with high probability, instead only giving a lower bound of exponentially small probability.
Part~(a) is a concentration inequality, which in our case will come from a martingale argument, showing that the random variable also enjoys concentration around its mean with Gaussian tails.
Exploiting the interplay of these two inequalities, we may in fact ``repair'' the ineffective second moment, as follows.
\begin{lemma}
    \label{lem:frieze}
    Suppose $X \geq 0$ is a random variable and $m > 0$ are such that the following two statements hold, for some constants $0 < \beta < \alpha$:
    \begin{enumerate}
        \item[(a)] $\PP[X - \EE X \leq - t] \vee \PP[X - \EE X \geq t] \leq \exp(-\alpha t^2 / m)$ for all $t > 0$.
        \item[(b)] $\PP[X \geq m] \geq \exp(-\beta m)$.
    \end{enumerate}
    Then, for any $0 < \gamma < 1 - \sqrt{\beta / \alpha}$,
    \begin{equation}
        \PP[X > \gamma m] \geq 1 - \exp\left(-\alpha\left(1 - \sqrt{\frac{\beta}{\alpha}} - \gamma\right)^2 m\right).
    \end{equation}
\end{lemma}
\begin{proof}
    Suppose $\delta \in (0, 1)$.
    Then, whenever $\EE X \leq (1 - \delta) m$, we have
    \begin{align*}
        \exp(-\beta m)
        &\leq \PP[X \geq m] \\
        &\leq \PP[X \geq \EE X + \delta m] \\
        &\leq \exp\left(-\frac{\alpha (\delta m)^2}{m}\right) \\
        &= \exp(-\alpha \delta^2 m), \numberthis
    \end{align*}
    whereby $\delta \leq \sqrt{\beta / \alpha}$.
    Thus, by contrapositive, $\EE X > (1 - \delta)m$ for all $\delta > \sqrt{\beta/\alpha}$, so $\EE X \geq (1 - \sqrt{\beta/\alpha})m$.
    
    Now, for all $0 < \gamma < 1 - \sqrt{\beta / \alpha}$, we have
    \begin{align*}
        \PP[X \leq \gamma m]
        &\leq \PP\left[X \leq \EE X - \left(1 - \sqrt{\frac{\beta}{\alpha}} - \gamma\right)m\right] \\
        &\leq \exp\left(-\alpha\left(1 - \sqrt{\frac{\beta}{\alpha}} - \gamma\right)^2 m\right), \numberthis
    \end{align*}
    as claimed.
\end{proof}

\noindent
Our formulation here is very similar to that of Frieze in \cite{Frieze-1990-IndependenceNumberRandomGraphs}, who treats the largest independent set in an ER graph; a similar idea also appeared earlier in \cite{SS-1987-SharpConcentrationChromaticNumber} for the chromatic number of an ER graph.
See also \cite{McDiarmid-1989-BoundedDifferences} for a survey of related methods.

\subsection{Type (a) and (b) Inequalities}

We now proceed to the main computations for using the concentration-enhanced second moment method, which we state as general claims for all dimensions $d$.
In the following sections we will derive specific consequences for different scalings of $d$.

Unfortunately, applying our method directly to the random variable $M$ does not afford us sufficient flexibility to adjust the constants $\alpha$ and $\beta$ such that the condition $\beta < \alpha$ is satisfied.
Instead, we will proceed by applying Lemma~\ref{lem:frieze} to the following adjustment of $M$, which is also directly analogous to the approach of Frieze in \cite{Frieze-1990-IndependenceNumberRandomGraphs}, there credited to Luczak, to the existence of independent sets.
Given $r \in \ZZ_+$, let $n^{\prime} \colonequals \lfloor n / r \rfloor$, and let $A_k = \{(k - 1)r + 1, \dots, kr\}$ for $k \in [n^{\prime}]$.
Then, we call a matching \emph{$r$-good} if all of its vertices belong to $A_1 \cup \cdots \cup A_{n^{\prime}}$, and it contains at most one vertex in each $A_k$.
We then work with the random variable
\begin{equation}
    M^{(r)} \colonequals \text{number of edges in the largest $r$-good matching in } G^{\aug}.
\end{equation}
Clearly, $M \geq M^{(r)}$.

\begin{lemma}[Type (a) inequality]
    \label{lem:second-moment-a}
    For all $t > 0$,
    \begin{equation}
        \PP[M^{(r)} - \EE M^{(r)} \leq - t] \vee \PP[M^{(r)} - \EE M^{(r)} \geq t] \leq \exp\left(-\frac{t^2}{2n^{\prime}}\right).
    \end{equation}
    That is, inequality (a) of Lemma~\ref{lem:frieze} holds for $M^{(r)}$ for any $m > 0$ with
    \begin{equation}
        \alpha = \frac{m}{2n^{\prime}}.
    \end{equation}
\end{lemma}
\begin{proof}
    For an arbitrary graph $G$ on vertex set $[n]$, let $M^{(r)}(G)$ denote the number of edges in the largest $r$-good matching in $G$.
    
    We first claim that, if there exists some $k \in [n^{\prime}]$ such that $G$ and $G^{\prime}$ differ only on edges incident with $A_k$, then $|M^{(r)}(G) - M^{(r)}(G^{\prime})| \leq 1$.
    Indeed, if the largest matching in $G^{\prime}$ contains no edge incident with $A_k$, then the same matching exists in $M^{(r)}(G)$, so $M^{(r)}(G) \geq M^{(r)}(G^{\prime})$.
    If the largest matching in $G^{\prime}$ does contain an edge incident with $A_k$, then the matching formed by removing that edge exists in $M^{(r)}(G)$, so $M^{(r)}(G) \geq M^{(r)}(G^{\prime}) - 1$.
    Thus $M^{(r)}(G^{\prime}) - M^{(r)}(G) \leq 1$, and symmetrically $M^{(r)}(G) - M^{(r)}(G^{\prime}) \leq 1$.
    
    Now, view $M^{(r)} = M^{(r)}(G^{\aug})$ as a function of $\bx_1, \bz_1, \dots, \bx_n, \bz_n$.
    Form the Doob's martingale $M^{(r)}_k \colonequals \EE[M^{(r)} \mid \{\bx_i\}_{i \in A_1 \cup \cdots \cup A_k} \cup \{\bz_i\}_{i \in A_1 \cup \cdots \cup A_k}]$ for $k = 0, 1, \dots, n^{\prime}$, for which $M^{(r)}_0 = \EE M^{(r)}$ and $M^{(r)}_{n^{\prime}} = M^{(r)}$.
    By the above claim, $|M^{(r)}_k - M^{(r)}_{k - 1}| \leq 1$ for all $k$, and the result then follows from the Azuma-Hoeffding inequality (see Lemma 1.2 of \cite{McDiarmid-1989-BoundedDifferences}).
\end{proof}

Our type (b) inequality involves the multinomial entropy function $H$, defined for $x_1, \dots, x_k \geq 0$ satisfying $x_1 + \cdots + x_k \leq 1$ as
\begin{equation}
    H(x_1, \dots, x_k) \colonequals -\sum_{i = 1}^k x_i \log x_i - \left(1 - \sum_{i = 1}^k x_i\right)\log\left(1 - \sum_{i = 1}^k x_i\right).
\end{equation}
We use the slightly non-standard notation of omitting what is usually the last argument $1 - \sum_{i = 1}^k x_i$ to shorten the expressions that arise below; this is, however, in agreement with the standard notation $H(x) = -x\log x - (1 - x)\log(1 - x)$ for the binomial entropy.

We give a coarsely-bounded exponential rate function below; this will suffice for our purposes and we make no efforts to optimize our analysis at the level of constants on the exponential scale in $m$.
More precise expressions are mentioned in our proof to follow.
\begin{lemma}[Type (b) inequality]
    \label{lem:second-moment-b}
    Suppose $n^{\prime} \geq 4m$ and $\sigma^2 \leq \frac{d}{40}$.
    Define as before $p \colonequals \PP[\{i, j\} \in E(G^{\aug})]$.
    Then,
    \begin{equation}
        \PP[M^{(r)} \geq m] \geq \exp\left(-m \sup_{\bx \in \sA} F(\bx) - O(\log n^{\prime})\right),
    \end{equation}
    where, for an absolute positive constant $K$ (e.g., one may take $K = 50$),
    \begin{align}
        \sA &\colonequals \bigg\{ (\oa, \ob, \oc, \oj, \ok, \ol) \in  [0, 1]^6: 2\oa + \ob + 2\oc + \oj + \ok + \ol \leq 1\bigg\}, \\
        R_1 &\colonequals K + \log\left(\frac{n^{\prime^2}}{pn^2 m}\right), \\
        R_2 &\colonequals K + d(S(\sigma^2, 2) - I(\sigma^2)) + 4\log_+\left(\frac{d}{1 + \sigma^2}\right) - 2\log r, \\
        F(\oa, \ob, \oc, \oj, \ok, \ol) &\colonequals 7H\left(\oa, \oa, \ob, \oc, \oc, \oj, \ok, \ol\right) + \left(\oa + \ob + \oc + \oj + \ok + \ol\right)(R_1 \vee R_2).
    \end{align}
    That is, if $n^{\prime}$ and $m$ are functions of $n \to \infty$ with $n^{\prime} = e^{o(m)}$ then, for any $\epsilon > 0$, for all sufficiently large $n$, inequality (b) of Lemma~\ref{lem:frieze} holds with
    \begin{equation}
        \beta = \sup_{\bx \in \sA} F(\bx) + \epsilon.
    \end{equation}
\end{lemma}
\noindent
The basic idea of the remaining analysis will be to choose $r$ and $m$ to ensure that $R_1$ and $R_2$ are very negative, forcing $\oa, \ob, \oc, \oj, \ok, \ol$ to be small at the maximizing point.
In $R_1$, we will accomplish this by taking $n^{\prime} = Cm$ for some fixed $C$, and $m = cpn^2$ for some sufficiently small $c$.
Then, the first term of $F$ is also small, so $\sup F$ and therefore $\beta$ may be made arbitrarily small by lowering $c$.
On the other hand, $\alpha = m / 2n^{\prime} = 1/2C$, so we may ensure $\beta < \alpha$ and apply Lemma~\ref{lem:frieze}, finding that with high probability $M \geq M^{(r)} \geq c^{\prime}pn^2$ for some $0 < c^{\prime} < c$.

The following is the main technical preliminary to our proof, which bounds the moment generating function of the number of connected components in the union of two random edge-disjoint perfect matchings.
\begin{proposition}[Cycle moment generating function]
    \label{prop:cycles-mgf}
    Let $\ell$ be even and let $K_{\ell}$ be the complete graph on vertex set $[\ell]$.
    Let $Q_1$ be any perfect matching in $K_{\ell}$, and let $Q_2$ be a uniformly random perfect matching in $K_{\ell}$ with the edges of $Q_1$ removed.
    Write $X_{\ell}$ for the random variable giving the number of connected components in $Q_1 \cup Q_2$.
%    Then, for all $\ell \geq 4$ and all $a \geq 10\ell$,
%    \begin{equation}
%        \EE a^{X_{\ell}} \leq \left(\frac{24 a}{\ell}\right)^{\ell / 4}.
%    \end{equation}
Then, for all $\ell \geq 4$ and all $a \geq \ell$,
    \begin{equation}\label{mgf}
        \EE a^{X_{\ell}} \leq \frac{(\phi^2a)^{\ell/4}}{(\ell/2)!!} \leq \left(\frac{e^3 a}{\ell}\right)^{\ell/4},
    \end{equation}
    where $\phi = (1+\sqrt 5)/2$ denotes the golden ratio.
    \end{proposition}
\begin{proof}
    %Note that $Q_1 \cup Q_2$ as described is a uniformly random 2-regular graph on $[\ell]$. JNW: doesn't seem true? a 2 regular graph can contain odd cycles.
    We prove our bound inductively.
    For any fixed $a$, write $m_{\ell} \colonequals \EE a^{X_{\ell}}$ for each even $\ell$, where we take $m_0 = 1$ and $m_2 = 0$.
    We will prove that
	\begin{equation}\label{eq:m_recurrence}
		m_\ell = \frac{a}{\ell - 3} m_{\ell - 4} + \left(1 - \frac{1}{\ell - 3}\right) m_{\ell - 2}\,.
	\end{equation}
	
	Let us assume that~\eqref{eq:m_recurrence} holds for now and show how to derive the claim.
	Clearly the first inequality in~\eqref{mgf} holds for $\ell = 0, 2, 4$.
    For the inductive step, suppose the bound holds for all values smaller than a given $\ell \geq 6$.
    The induction hypothesis then implies
    \begin{align*}
    m_\ell & \leq \frac{a}{\ell -3} \cdot \frac{(\phi^2a)^{\ell/4 - 1}}{(\ell/2 - 2)!!}  + \left(1 - \frac{1}{\ell - 3}\right) \frac{(\phi^2a)^{\ell/4 - 1/2}}{(\ell/2 - 1)!!} \\
    & \leq \frac{a}{\ell/2} \cdot \frac{(\phi^2a)^{\ell/4 - 1}}{(\ell/2 - 2)!!} + \frac{(\phi^2a)^{\ell/4 - 1/2} a^{1/2}}{(\ell/2)!!} \\
    & = \frac{(\phi^2a^{\ell/4})}{(\ell/2)!} (\phi^{-2} + \phi^{-1})\,,
    \end{align*}
    where we have used that $\ell - 3 \geq \ell/2$ for all $\ell \geq 6$ and that $(\ell/2 - 1)!! \geq (\ell/2)!!/\sqrt{\ell} \geq (\ell/2)!!/\sqrt{a}$.
    Since $\phi^{-2} + \phi^{-1} = 1$, this completes the induction and proves the first inequality in~\eqref{mgf}, and the second is an immediate consequence.

    All that is left is to establish the promised recurrence~\eqref{eq:m_recurrence}.
        Note that $(Q_1, Q_2)$ as described are two uniformly random perfect matchings on $[\ell]$ conditioned to be edge disjoint.
    Each connected component of $Q_1 \cup Q_2$ is a cycle whose edges alternate between $Q_1$ and $Q_2$.
    Let us condition on the size of the component containing the vertex $1$.
    Write $i$ for the neighbor of $1$ in $Q_1$, and $j$ and $k$ for the neighbors of $1$ and $i$, respectively, in $Q_2$.
    Since $Q_1$ and $Q_2$ are edge-disjoint perfect matchings, $1$, $i$, $j$, and $k$ are distinct.
    
    If $1$ lies in a 4-cycle, then $\{j, k\} \in Q_1$, and removing the vertices $\{1, i, j, k\}$, and corresponding edges from $Q_1$ and $Q_2$ yields two uniformly random, edge-disjoint perfect matchings on $\ell - 4$ vertices, with one fewer connected component than $Q_1 \cup Q_2$.
    Since $k$ is a uniform random vertex from $[\ell] \setminus \{1, i, j\}$, this situation occurs with probability $\frac{1}{\ell - 3}$.
    This gives the first term of~\eqref{eq:m_recurrence}.
     
    On the other hand, if $1$ lies in a cycle of length greater than $4$, then $\{j, k\} \notin Q_1$.
    Removing the vertices $1$ and $i$ as well as the edges $\{1, i\}$ from $Q_1$ and $\{1, j\}$ from $Q_2$ and replacing the edge $\{i, k\}$ by $\{j, k\}$ in $Q_2$ yields two uniformly random, edge-disjoint perfect matchings on $\ell - 2$ vertices, with the same number of connected components as $Q_1 \cup Q_2$.
    Since this occurs with probability $1 - \frac{1}{\ell - 3}$, this yields the second term of~\eqref{eq:m_recurrence}.

\end{proof}

We will also use the following inequalities among the various functions of $\sigma^2$, whose proofs we defer to Appendix~\ref{app:pf:prop:eta-zeta-bounds}.
\begin{proposition}
    \label{prop:eta-zeta-bounds}
    For $\sigma^2 > 0$, define
    \begin{align}
        \eta_1 = \eta_1(\sigma^2) &\colonequals \frac{3}{4}S(\sigma^2, 2) - \frac{1}{4}S(\sigma^2, 4), \\
        \eta_2 = \eta_2(\sigma^2) &\colonequals S(\sigma^2, 2) -\frac{1}{2}S(\sigma^2, 3), \\
        \eta_3 = \eta_3(\sigma^2) &\colonequals \frac{1}{2}S(\sigma^2, 2) - \frac{1}{2}I(\sigma^2).
    \end{align}
    Then, we have
    \begin{align}
        \eta_i &\leq \eta_3 \leq \frac{3}{2 + 8\sigma^2} \text{ for each } i \in \{1, 2, 3\}.
    \end{align}
\end{proposition}
\noindent
We remark that these results are qualitatively sharp, in that the given quantities indeed approach positive constants as $\sigma^2 \to 0$, and decay as $O(\sigma^{-2})$ as $\sigma^2 \to \infty$; proofs of matching opposite bounds follow from similar elementary manipulations to those we give in the proof.

Finally, we will use the following standard properties of the multinomial entropy function $H$.
We note that we adopt the same convention for multinomial coefficients of omitting the last argument as we do for $H$:
\begin{equation}
    \binom{m}{a_1, \dots, a_k} \colonequals \frac{m!}{a_1! \cdots a_k! (m - a_1 - \cdots - a_k)!}.
\end{equation}
\begin{proposition}
    \label{prop:entropy-properties}
    The function $H$ satisfies the following properties:
    \begin{enumerate}
        \item $H(x_1, \dots, x_k) \leq \log(k + 1)$.
        \item For any $x \in (0, 1)$, $tH(x / t)$ is a strictly increasing function of $t$.
        \item For any $x_1, \dots, x_k \geq 0$ with $x_1 + \cdots + x_k \leq 1$, $H(x_1 + x_2, x_3, \dots, x_k) \leq H(x_1, x_2, x_3, \dots, x_k)$, and for any $k^{\prime} < k$, $H(x_1, \dots, x_{k^{\prime}}) \leq H(x_1, \dots, x_k)$.
        \item A multinomial coefficient is bounded by the entropy as 
        \begin{equation}
            \exp\left(mH\left(\frac{a_1}{m}, \cdots, \frac{a_k}{m}\right) - O_k(\log m)\right) \leq \binom{m}{a_1, \dots, a_k} \leq \exp\left(mH\left(\frac{a_1}{m}, \cdots, \frac{a_k}{m}\right)\right).
        \end{equation}
    \end{enumerate}
\end{proposition}

\begin{proof}[Proof of Lemma~\ref{lem:second-moment-b}]
    Define the random variable
    \begin{equation}
        N \colonequals \#\{r\text{-good matchings on } 2m \text{ vertices of } G^{\aug}\}.
    \end{equation}
    We then have
    \begin{equation}
        \PP[M^{(r)} \geq m] = \PP[N > 0],
    \end{equation}
    and we will bound the latter from below by the second moment method.
    
    Let $\sM$ denote the set of $r$-good matchings of $2m$ vertices of the complete graph on $[n]$, whose cardinality is
    \begin{equation}
        |\sM| =  \binom{n^{\prime}}{2m} r^{2m} (2m - 1)!!.
    \end{equation}
    We then have by linearity of expectation that
    \begin{equation}
        \EE N = p^{m} |\sM|.
    \end{equation}
    Let $Q_0$ be a fixed $r$-good matching of $m$ elements in the complete graph on $[n]$ (say, the graph with edges $\{1, 2\}, \{3, 4\}, \dots, \{2m - 1, 2m\}$).
    By symmetry, we have
    \begin{equation}
        \EE N^2 = |\sM| \sum_{Q \in \sM} \PP[Q_0 \cup Q \subseteq G^{\aug}],
    \end{equation}
    and therefore the moment ratio may be written as an average,
    \begin{equation}
        \frac{\EE N^2}{(\EE N)^2} = \frac{1}{|\sM|} \sum_{Q \in \sM} \frac{\PP[Q_0 \cup Q \subseteq G^{\aug}]}{p^{2m}}.
    \end{equation}
    
    Given a graph $G$, write $\cc(G)$ for the set of its connected components, $\cc_2(G)$ for the set of its connected components isomorphic to the path on two vertices, $\cc_3(G)$ for the set of those isomorphic to the path on three vertices, and $\cc_{\geq 4}(G)$ for the set of the remaining connected components.
    Let us abbreviate $G \colonequals Q_0 \cup Q$.
    Note that all components of $G$ are then either cycles of even length at least 4 or paths.
    Then, by Proposition~\ref{prop:aug-per-edge-path-cycle}, for any connected component $H \in \cc_{\geq 4}(G)$ we have
    \begin{align}
        \PP[H \subseteq G^{\aug}] 
        &\leq \exp\left(-\frac{d}{2}S(\sigma^2, |V(H)|)\right)
        \intertext{and, applying Lemma~\ref{lem:S-lower} with $t_0 = 4$, we have}
        &\leq \exp\left(-\frac{d}{2}(|V(H)| I - J)\right),
    \end{align}
    where $I = I(\sigma^2)$ and $J = J(\sigma^2) = 4I(\sigma^2) - S(\sigma^2, 4) > 0$.
    For the remainder of this proof, let us follow the above convention abbreviating $I = I(\sigma^2)$ and $J = J(\sigma^2)$, and also writing $S_t = S(\sigma^2, t)$.
    Using this bound for connected components on at least four vertices and Proposition~\ref{prop:aug-per-edge-path-cycle} for connected components on two and three vertices, we then have
    \begin{align*}
        &\PP[G \subseteq G^{\aug}] \\
        &= \prod_{H \in \cc(G)} \PP[H \subseteq G^{\aug}] \\
        &\leq \exp\left(|\cc_2(G)| \log p - |\cc_3(G)| \frac{d}{2}S_3 - \frac{d}{2}\sum_{H \in \cc_{\geq 4}(G)}(|V(H)|I - J)\right) \\
        &= \exp\left(|\cc_2(G)| \log p - |\cc_3(G)| \frac{d}{2}S_3 - (|V(G)| - 2|\cc_2(G)| - 3|\cc_3(G)|)\frac{d}{2}I + |\cc_{\geq 4}(G)|\frac{d}{2}J\right) \\
        &= \exp\left(-|V(G)|\frac{d}{2}I + |\cc_2(G)| \left(dI + \log p\right) + |\cc_3(G)|\left(\frac{3d}{2}I - \frac{d}{2}S_3\right) + |\cc_{\geq 4}(G)|\frac{d}{2}J\right). \numberthis
    \end{align*}
    
    \begin{figure}
        \centering
        \includegraphics[scale=0.8]{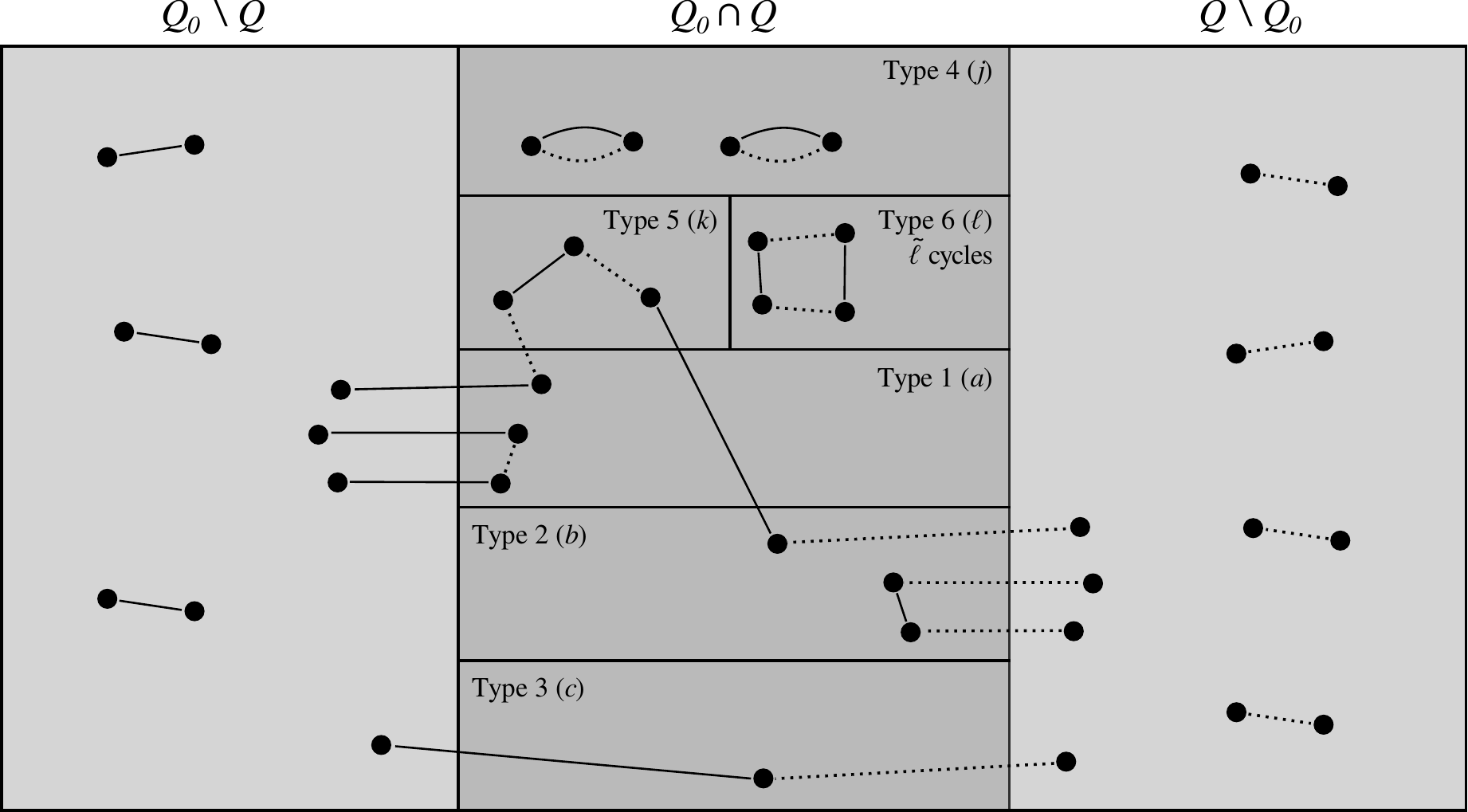}
        \caption{\textbf{Union of two matchings.} We illustrate the decomposition of two partial matchings of $[n]$ from the proof of Lemma~\ref{lem:second-moment-b}. The center region contains all vertices of $V(Q_0) \cap V(Q)$, solid lines indicate edges of $E(Q_0)$, and dotted lines indicate edges of $E(Q)$.}
        \label{fig:matching-types}
    \end{figure}
    
    We now divide the sum over $Q \in \sM$ into portions over which we may uniformly bound this probability.
    To do this, for any given $Q$ we introduce the following classification of the vertices of $Q_0 \cap Q$, into ``types'' 1, 2, 3, 4, 5, and 6.
    Next to each type, we give the letter that will denote the number of vertices of this type, $a, b,c,j, k,$ and $\ell$, respectively:
    \begin{enumerate}
        \item $a$ vertices whose neighbor in $Q_0$ lies in $Q_0 \setminus Q$ and whose neighbor in $Q$ lies in $Q_0 \cap Q$.
        \item $b$ vertices whose neighbor in $Q_0$ lies in $Q_0 \cap Q$ and whose neighbor in $Q$ lies in $Q \setminus Q_0$.
        \item $c$ vertices whose neighbor in $Q_0$ lies in $Q_0 \setminus Q$ and whose neighbor in $Q$ lies in $Q \setminus Q_0$.
        \item $j$ vertices whose neighbors in both $Q_0$ and $Q$ are equal.
        \item $k$ vertices which belong to a path connected component and whose neighbors in $Q_0$ and $Q$ are different but both lie in $Q_0 \cap Q$.
        \item $\ell$ vertices which belong to a cycle connected component and whose neighbors in $Q_0$ and $Q$ are different but both lie in $Q_0 \cap Q$.
    \end{enumerate}
    We also denote by $\widetilde{\ell}$ the number of cycle connected components in $Q_0 \cup Q$ (which all consist of Type 6 vertices).
    See Figure~\ref{fig:matching-types} for an illustration of this decomposition.
    With these notations, recalling that $|V(Q_0)| = |V(Q)| = 2m$, we have
    \begin{align*}
        |V(Q_0) \cap V(Q)| &= a + b + c + j + k + \ell, \numberthis \\
        |V(Q_0 \cup Q)| &= |V(Q_0)| + |V(Q)| - |V(Q_0) \cap V(Q)| \\
        &= 4m - a - b - c - j - k - \ell, \numberthis \\
        |\cc_2(Q_0 \cup Q)| &= \frac{j}{2} + \frac{2m - 2a - b - 2c - j - k - \ell}{2} + \frac{2m - a - 2b - 2c - j - k - \ell}{2} \\
        &= 2m - \frac{3}{2}a - \frac{3}{2}b - 2c - \frac{1}{2}j - k - \ell, \numberthis \\
        |\cc_3(Q_0 \cup Q)| &= c, \numberthis \\
        |\cc_{\geq 4}(Q_0 \cup Q)| &= \frac{1}{2}a + \frac{1}{2}b + \widetilde{\ell}, \numberthis
    \end{align*}
    the final claim following because every path component of length 4 or greater contains exactly two internal vertices of Type 1 or Type 2.
    
    Let $\sM_{a, b, c, j, k, \ell, \widetilde{\ell}}$ be the set of $Q$ such that $Q_0 \cup Q$ has the specified number of vertices of each type, and $\widetilde{\ell}$ cycle connected components.
    We note that this set is empty unless $j$, $k + a$, and $k + b$, and $\ell$ are all even, since these sets of vertices must admit perfect matchings (from restrictions of both $Q_0$ and $Q$, $Q$, $Q_0$, and both $Q_0$ and $Q$, respectively).
    For $Q \in \sM_{a, b, c, j, k, \ell, \widetilde{\ell}}$, we have
    \begin{align*}
        &\frac{\PP[Q_0 \cup Q \subseteq G^{\aug}]}{p^{2m}} \\
        &\leq \exp\bigg(-(4m - a - b - c - j - k - \ell) \frac{d}{2}I \\
        &\hspace{2cm} + \left(2m - \frac{3}{2}a - \frac{3}{2}b - 2c - \frac{1}{2}j - k - \ell\right)(dI + \log p) \\
        &\hspace{2cm} + c\left(\frac{3d}{2}I - \frac{d}{2}S_3\right) \\
        &\hspace{2cm} + \left(\frac{1}{2}a + \frac{1}{2}b + \widetilde{\ell}\right)\frac{d}{2}J \\
        &\hspace{2cm} - 2m\log p\bigg) \\
        &= \exp\bigg(a\left(-dI + \frac{d}{4}J - \frac{3}{2}\log p\right) + b\left(-dI + \frac{d}{4}J - \frac{3}{2}\log p\right) + c\left(-\frac{d}{2}S_3 - 2\log p\right) \\
        &\hspace{2cm} + j\left(-\frac{1}{2}\log p\right) + k\left(-\frac{d}{2}I - \log p\right) + \ell\left(-\frac{d}{2}I - \log p\right) + \widetilde{\ell}\left(\frac{d}{2}J\right)\bigg)
        \intertext{and, using that $\log p = -\frac{d}{2}S_2 + \log \what{p}$,}
        &= \exp\bigg((a + b)\left(\eta_1 d - \frac{3}{2}\log\what{p}\right) + c\left(\eta_2 d - 2\log\what{p}\right) + (k + \ell)\left(\eta_3 d - \log \what{p}\right) \\
        &\hspace{2cm} + j\left(-\frac{1}{2}\log p\right) + \widetilde{\ell}\left(\frac{d}{2}J\right)\bigg), \numberthis
    \end{align*}
    with $\eta_i$ as defined in Proposition~\ref{prop:eta-zeta-bounds},
    \begin{align}
        \eta_1 &= -I + \frac{1}{4}J + \frac{3}{4}S_2 = \frac{3}{4}S_2 - \frac{1}{4}S_4, \\
        \eta_2 &= S_2 -\frac{1}{2}S_3, \\
        \eta_3 &= \frac{1}{2}S_2 - \frac{1}{2}I.
    \end{align}
    By Proposition~\ref{prop:eta-zeta-bounds} we have $\eta_i \leq \eta_3$, and by Proposition~\ref{prop:pair-aug-lb} we have $\what{p} \geq \frac{1}{1000}\sqrt{\frac{1 + \sigma^2}{d}}$.
    Here and in the remainder of the proof, let $K$ be a large constant that may vary line to line.
    Substituting these bounds,
    \begin{align*}
        &\frac{\PP[Q_0 \cup Q \subseteq G^{\aug}]}{p^{2m}} \\
        &\leq \exp\bigg((a + b + c + k)\left(K + \eta_3d + \log\left(\frac{d}{1 + \sigma^2}\right)\right) \\
        &\hspace{2cm} + j\left(-\frac{1}{2}\log p\right) + \ell\left(K + \eta_3 d + \log\left(\frac{d}{1 + \sigma^2}\right)\right) + \widetilde{\ell}\left(\frac{d}{2}J\right)\bigg). \numberthis \label{eq:second-moment-prob-bound}
    \end{align*}
    
    We must also control the size of the subsets $\sM_{a, b, c, j, k, \ell, \widetilde{\ell}}$, which is the content of the following technical lemma, whose proof we defer to the conclusion of this section.
    \begin{lemma}\label{lem:M-bound}
    Given $\ell$ and $\widetilde{\ell}$, let $R_0$ be a perfect matching of $[\ell]$, and write $\Cyc(\ell, \widetilde{\ell})$ for the number of perfect matchings $R$ of $[\ell]$ edge-disjoint from $R_0$ such that $R_0 \cup R$ contains exactly $\widetilde{\ell}$ cycles.
    Define normalizations $\oa \colonequals a / 2m$ and likewise $\ob, \oc, \oj, \ok,$ and $\ol$.
Then $\sM_{a, b, c, j, k, \ell, \widetilde{\ell}}$ satisfies
\begin{multline}\label{eq:second-moment-ent-bound}
\frac{|\sM_{a, b, c, j, k, \ell, \widetilde{\ell}}|}{|\sM|} \leq \exp\bigg(m\bigg[5H\left(\oa, \oa, \ob, \oc, \oc, \oj, \ok, \ol\right) + (\oa + \oc)\log 4\bigg] + O(\log n^{\prime})\bigg) \cdot \\ r^{- a - b - c - j - k - \ell} \frac{1}{(j - 1)!!} \frac{\Cyc(\ell, \widetilde{\ell})}{(\ell - 1)!!}.
\end{multline}
    \end{lemma}

    Putting the bounds \eqref{eq:second-moment-prob-bound} and \eqref{eq:second-moment-ent-bound} together,
    \begin{align*}
        \frac{\EE N^2}{(\EE N)^2}
        &\leq \sum_{a, b, c, j, k, \ell, \widetilde{\ell}} \frac{|\sM_{a, b, c, j, k, \ell, \widetilde{\ell}}|}{|\sM|} \max_{Q \in \sM_{a, b, c, j, k, \ell, \widetilde{\ell}}} \frac{\PP[Q_0 \cup Q \subseteq G^{\aug}]}{p^{2m}} \\
        &\leq \sum_{a, b, c, j, k, \ell}\exp\bigg(m\bigg[5H\left(\oa, \oa, \ob, \oc, \oc, \oj, \ok, \ol\right) \\
        &\hspace{3cm} + (\oa + \ob + \oc + \ok)\left(K + 2\eta_3d + 2\log_+\left(\frac{d}{1 + \sigma^2}\right) - 2\log r\right) \\
        &\hspace{3cm} + \oj\left(-\log p - 2\log r\right) \\
        &\hspace{3cm} + \ol\left(K + 2\eta_3 d + 2\log_+\left(\frac{d}{1 + \sigma^2} \right) - 2\log r\right)\bigg] + O(\log n^{\prime})\bigg) \\
        &\hspace{2cm} \frac{1}{(j - 1)!!} \sum_{\widetilde{\ell}} \frac{\Cyc(\ell, \widetilde{\ell})}{(\ell - 1)!!}(e^{\frac{d}{2}J})^{\widetilde{\ell}} \numberthis
    \end{align*}
    
    For the remaining sum over $\widetilde{\ell}$, we use Proposition~\ref{prop:cycles-mgf}.
    We note first that, since $(\ell - 1)!!$ is the total number of matchings on $[\ell]$ and $\frac{\Cyc(\ell, \widetilde{\ell})}{(\ell - 1)!!}$ is the number of such matchings that are disjoint from a fixed matching and whose union with that matching contains $\widetilde{\ell}$ cycles, we may generally bound $\sum_{\widetilde{\ell}}\frac{\Cyc(\ell, \widetilde{\ell})}{(\ell - 1)!!} f(\widetilde{\ell}) \leq \EE f(X_{\ell})$, where $X_{\ell}$ is the random variable from Proposition~\ref{prop:cycles-mgf}.
    If $e^{\frac{d}{2}J} \leq \ell$, then we may bound
    \begin{equation}
        \sum_{\widetilde{\ell}} \frac{\Cyc(\ell, \widetilde{\ell})}{(\ell - 1)!!}(e^{\frac{d}{2}J})^{\widetilde{\ell}} \leq \sum_{\widetilde{\ell}} \frac{\Cyc(\ell, \widetilde{\ell})}{(\ell - 1)!!}(\ell)^{\widetilde{\ell}} \leq (e^3)^{\ell / 4}.
    \end{equation}
    If $e^{\frac{d}{2}J} \geq \ell$, then we have
    \begin{equation}
        \sum_{\widetilde{\ell}} \frac{\Cyc(\ell, \widetilde{\ell})}{(\ell - 1)!!}(e^{\frac{d}{2}J})^{\widetilde{\ell}} \leq \left(\frac{e^{3 + \frac{d}{2}J}}{\ell}\right)^{\ell / 4}.
    \end{equation}
%    \begin{equation}
%        \sum_{\widetilde{\ell}} \frac{\Cyc(\ell, \widetilde{\ell})}{(\ell - 1)!!}(e^{\frac{d}{2}J})^{\widetilde{\ell}} \leq \ell e^{ \frac \ell 4(1 + \frac{d}{2}J)}   \end{equation}
    For the remaining term involving $j$, we bound $(j - 1)!! \geq (j / e)^{j/2}$.
    Combining these estimates,% and recognizing the rate function $R_2$, 
    we may incorporate everything under the exponential as

    \begin{align*}
        \frac{\EE N^2}{(\EE N)^2}
        &\leq \sum_{a, b, c, j, k, \ell}\exp\bigg(m\bigg[5H\left(\oa, \oa, \ob, \oc, \oc, \oj, \ok, \ol\right) \\
        &\hspace{3cm} + (\oa + \ob + \oc + \ok)R_2 \\
        &\hspace{3cm} + \oj\left(1 -\log p - 2\log r + \log\left(\frac{1}{j}\right)\right) \\
        &\hspace{3cm} + \ol\left(K + 2\eta_3 d + 2\log_+\left(\frac{d}{1 + \sigma^2} \right) - 2\log r + 0 \vee\left(\frac{d}{4}J + \frac{1}{2}\log\left(\frac{1}{\ell}\right)\right)\right)\bigg] \\
        &\hspace{2.5cm} + O(\log n^{\prime})\bigg) \numberthis
    \end{align*}
    
    To bound the remaining rates, we first consider the $\oj$ term.
    Recall that $n^{\prime} = \lfloor n / r \rfloor \geq n / 2r$, so $r \geq n / 2n^{\prime}$.
    Thus we have
    \begin{align*}
        -\log p - 2\log r + \log\left(\frac{1}{j}\right)
        &= \log \frac{1}{pr^2 j} \\
        &\leq \log \frac{4n^{\prime^{2}}}{pn^2 j} \\
        &= \log \frac{2n^{\prime^{2}}}{pn^2 m} - \log \oj. \numberthis
    \end{align*}
    Extracting a similar expression by adding and subtracting $\frac{1}{2}\log p$ in the $\ol$ term when the second term of the maximum is greater than zero,
    \begin{align*}
        &K + 2\eta_3 d + 2\log_+\left(\frac{d}{1 + \sigma^2} \right) - 2\log r + \frac{d}{4}J + \frac{1}{2}\log\left(\frac{1}{\ell}\right) \\
        &= K + \left(2\eta_3 + \frac{1}{4}J\right) d + 2\log_+\left(\frac{d}{1 + \sigma^2} \right) + \frac{1}{2}\log p - \log r + \frac{1}{2}\log \frac{2n^{\prime^{2}}}{pn^2 m} - \log \ol \\
        &\leq K + \left(2\eta_3 + \frac{1}{4}J - \frac{1}{4}S_2\right) d + 2\log_+\left(\frac{d}{1 + \sigma^2} \right) - \log r + \frac{1}{2}\log \frac{2n^{\prime^{2}}}{pn^2 m} - \log \ol
        \intertext{and we notice $2\eta_3 + \frac{1}{4}J - \frac{1}{4}S_2 = S_2 - I + I - \frac{1}{4}S_4 - \frac{1}{4}S_2 = \frac{3}{4}S_2 - \frac{1}{4}S_4 = \eta_1$, so, using Proposition~\ref{prop:eta-zeta-bounds} to bound $\eta_1 \leq \eta_3$,}
        &\leq K + \eta_3 d + 2\log\left(\frac{d}{1 + \sigma^2} \right) - \log r + \frac{1}{2}\log \frac{2n^{\prime^{2}}}{pn^2 m} - \log \ol \\
        & \leq \frac 12 (R_1 + R_2) - \log \ol \numberthis
    \end{align*}
    We note also that $-\oj\log\oj \leq H(\oj)$ and likewise $-\ol \log \ol \leq H(\ol)$, and both of these are bounded by $H(\oa, \oa, \ob, \oc, \oc, \oj, \ok, \ol)$ by Proposition~\ref{prop:entropy-properties}.
    
    Applying these observations,
    \begin{align*}
        \frac{\EE N^2}{(\EE N)^2}
        &\leq \sum_{a, b, c, j, k, \ell}\exp\bigg(m\bigg[7H\left(\oa, \oa, \ob, \oc, \oc, \oj, \ok, \ol\right) \\
        &\hspace{3cm} + (\oa + \ob + \oc + \ok)R_2 + \oj R_1 \\
        &\hspace{3cm} + \ol\max\left\{R_2, \frac 12 (R_1 + R_2)\right\}\bigg] \\
        &\hspace{2.5cm} + O(\log n^{\prime})\bigg)\,, \numberthis
    \end{align*}
    which implies the result.
\end{proof}

To complete the proof, it remains to justify Lemma~\ref{lem:M-bound}.
\begin{proof}[Proof of Lemma~\ref{lem:M-bound}]
%    Given $\ell$ and $\widetilde{\ell}$, let $R_0$ be a perfect matching of $[\ell]$, and write $\Cyc(\ell, \widetilde{\ell})$ for the number of perfect matchings $R$ of $[\ell]$ edge-disjoint from $R_0$ such that $R_0 \cup R$ contains exactly $\widetilde{\ell}$ cycles.
    We begin with a combinatorial bound.
    Below, line by line, the factors count the number of ways to choose the vertices of $V(Q) \cap V(Q_0)$, the number of ways to choose the vertices of $V(Q) \setminus V(Q_0)$, the number of ways to draw the edges of $E(Q)$ incident with $V(Q) \cap V(Q_0)$, and the number of ways to draw the edges of $E(Q)$ between pairs of $V(Q) \setminus V(Q_0)$:
    \begin{align*}
        &|\sM_{a, b, c, j, k, \ell, \widetilde{\ell}}| \\
        &\leq \binom{m}{a, c, \frac{b + k}{2}, \frac{j}{2}, \frac{\ell}{2}} 2^{a + c}\cdot \\
        &\phantom{\leq} \binom{n^{\prime} - a - \frac{b}{2} - c - \frac{j}{2} - \frac{k}{2}- \frac{\ell}{2}}{2m - a - b - c - j - k - \ell} r^{2m - a - b - c - j - k - \ell} \cdot \\
        &\phantom{\leq} \binom{b + k}{k} (k + a - 1)!! \binom{2m - a - b - c - j - k - \ell}{b + c}(b + c)! \Cyc(\ell, \widetilde{\ell}) \cdot\\
        &\phantom{\leq} (2m - a - 2b - 2c - j - k - \ell - 1)!!
        \intertext{Let us introduce $C \colonequals n^{\prime} / 2m$, which satisfies $C \geq 2$ by assumption.
    %    We also define normalizations $\oa \colonequals a / 2m$ and likewise $\ob, \oc, \oj, \ok,$ and $\ol$.
        Then, applying the entropy bound for multinomial coefficients wherever possible,}
        &\leq \exp\bigg(m\bigg[H\left(2\oa, 2\oc, \ob + \ok, \oj, \ol\right) \\
        &\phantom{\leq \exp\bigg(m\bigg[} + 2\left(C - \left(\oa + \frac{1}{2}\ob + \oc + \frac{1}{2}\oj + \frac{1}{2}\ok + \frac{1}{2}\ol\right)\right)H\left(\frac{1 - \oa - \ob - \oc - \oj - \ok - \ol}{C - (\oa + \frac{1}{2}\ob + \oc + \frac{1}{2}\oj + \frac{1}{2}\ok + \frac{1}{2}\ol)}\right)\\
        &\phantom{\leq \exp\bigg(m\bigg[} + 2(1 - \oa - \ob - \oc - \oj - \ok - \ol)H\left(\frac{\ob + \oc}{1 - \oa - \ob - \oc - \oj - \ok - \ol}\right) \\
        &\phantom{\leq \exp\bigg(m\bigg[} + 2(\ob + \ok)H\left(\frac{\ok}{\ob + \ok}\right) + 2(\oa + \oc)\log 2\bigg]\bigg)\cdot \\
        &\phantom{\leq } r^{2m - a - b - c - j - k - \ell} (k + a - 1)!! (b + c)! (2m - a - 2b - 2c - j - k - \ell - 1)!! (j - 1)!!(\ell - 1)!!\cdot \\
        &\phantom{\leq } \frac{1}{(j - 1)!!}\frac{\Cyc(\ell, \widetilde{\ell})}{(\ell - 1)!!} \numberthis
    \end{align*}
    We will in particular need to bound the fraction of $\sM$ occupied by each of these subsets.
    To that end, we note that
    \begin{equation}
        |\sM| = \binom{n^{\prime}}{2m} r^{2m} (2m - 1)!! \geq \exp\left(m\left[2C H\left(\frac{1}{C}\right)\right] - O(\log n^{\prime})\right) r^{2m} (2m - 1)!!.
    \end{equation}
    Considering the quotient of factorials and double factorials that will remain, an entropy bound again yields
    \begin{align*}
        &\hspace{-1cm}\frac{(k + a - 1)!! (b + c)! (2m - a - 2b - 2c - j - k - \ell - 1)!! (j - 1)!! (\ell - 1)!!}{(2m - 1)!!} \\
        &\leq \exp\left(-m H\left(\oa + \ok, \ob + \oc, \ob + \oc, \oj, \ol\right) + O(\log n^{\prime})\right). \numberthis
    \end{align*}
    Thus we find
    \begin{align*}
        &\frac{|\sM_{a, b, c, j, k, \ell, \widetilde{\ell}}|}{|\sM|} \\
        &\leq \exp\bigg(m\bigg[H\left(2\oa, 2\oc, \ob + \ok, \oj, \ol\right) \\
        &\hspace{2.5cm} + 2\left(C - \left(\oa + \frac{1}{2}\ob + \oc + \frac{1}{2}\oj + \frac{1}{2}\ok + \frac{1}{2}\ol\right)\right)H\left(\frac{1 - \oa - \ob - \oc - \oj - \ok - \ol}{C - (\oa + \frac{1}{2}\ob + \oc + \frac{1}{2}\oj + \frac{1}{2}\ok + \frac{1}{2}\ol)}\right)\\
        &\hspace{2.5cm} + 2(1 - \oa - \ob - \oc - \oj - \ok - \ol)H\left(\frac{\ob + \oc}{1 - \oa - \ob - \oc - \oj - \ok - \ol}\right) \\
        &\hspace{2.5cm} - H\left(\oa + \ok, \ob + \oc, \ob + \oc, \oj, \ol\right) - 2CH\left(\frac{1}{C}\right) \\
        &\hspace{2.5cm} + 2(\ob + \ok)H\left(\frac{\ok}{\ob + \ok}\right) + 2(\oa + \oc)\log 2\bigg] + O(\log n^{\prime})\bigg) \\
        &\hspace{1cm} r^{- a - b - c - j - k - \ell} \frac{1}{(j - 1)!!} \frac{\Cyc(\ell, \widetilde{\ell})}{(\ell - 1)!!}
        \intertext{and repeatedly use Proposition~\ref{prop:entropy-properties} to bound the entropies,}
        &\leq \exp\bigg(m\bigg[5H\left(\oa, \oa, \ob, \oc, \oc, \oj, \ok, \ol\right) + (\oa + \oc)\log 4\bigg] + O(\log n^{\prime})\bigg) \\
        &\hspace{1cm} r^{- a - b - c - j - k - \ell} \frac{1}{(j - 1)!!} \frac{\Cyc(\ell, \widetilde{\ell})}{(\ell - 1)!!},    \numberthis  \end{align*}
    where we have used that $C \geq 2$ ensures that the sum of the two terms involving $C$ is at most zero.
    \end{proof}

\subsection{Sublinear Error Lower Bound}

We now prove the following application of the above results, which implies the lower bound of Part~3 of Theorem~\ref{thm:d-ll-logn} and of Part~2 of Theorem~\ref{thm:d-theta-logn}.

\begin{lemma}\label{lem:sublin_lb}
    Define $\what{d} \colonequals 1 + d \wedge \log n \geq 2$ and $s \colonequals \what{d}^{1/d}$.
    Suppose that
    \begin{equation}
    \frac{1}{s^{-\omega(1)} n^{4/d} - 1} \leq \sigma^2 \leq \frac{1}{(2s^{\omega(1)}n^{1/d} - 1)^2 - 1}.
    \end{equation}
    Then, there exists an absolute constant $c > 0$ such that, with high probability,
    \begin{equation}
        |\sE| \geq \frac{c}{\sqrt{\what{d}}} \left(1 + \frac{1}{\sigma^2}\right)^{-d/2} n^2.
    \end{equation}
\end{lemma}
\begin{proof}
Let us bound $\what{p}$ from below under these assumptions, which amounts to bounding $\frac{d}{1 + \sigma^2}$ from above.
We always have $\frac{d}{1 + \sigma^2} \leq d$, and, using the lower bound above along with $1 - e^{-x} \leq x$, we have
\begin{equation}
    \frac{d}{1 + \sigma^2} \leq d(1 - s^{\omega(1)}n^{-4/d}) \leq d(1 - n^{-4/d}) \leq 4\log n.
\end{equation}
Thus,
\begin{equation}
    \frac{d}{1 + \sigma^2} \leq d \wedge 4\log n \leq 4\what{d},
\end{equation}
and so, by Proposition~\ref{prop:pair-aug-lb}
\begin{equation}
    \what{p} \geq \frac{1}{1000} \sqrt{\frac{1 + \sigma^2}{d}} \geq \frac{1}{2000}\what{d}^{-1/2}.
\end{equation}

Next, we bound the other term of $p$, $\exp(-\frac{d}{2}S(\sigma^2, 2))$, from above and below.
Since $S(\sigma^2, 2)$ is a decreasing function of $\sigma^2$, by the lower bound on $\sigma^2$ we have
\begin{equation}
    S(\sigma^2, 2) \leq S\left(\frac{1}{s^{-\omega(1)} n^{4/d} - 1}, 2\right) = \log(s^{-\omega(1)} n^{4/d}) \leq \frac{4\log n - \omega(\log \what{d})}{d},
\end{equation}
and thus
\begin{equation}
    \exp\left(-\frac{d}{2}S(\sigma^2, 2)\right) \geq \frac{1}{n^2} \exp(\omega(\log \what{d})),
\end{equation}
whereby $pn^2 = \what{p}\exp(-\frac{d}{2}S(\sigma^2, 2))n^2 \to \infty$ as $n \to \infty$.
On the other hand, by the upper bound on $\sigma^2$ we have
\begin{equation}
    S(\sigma^2, 2) \geq S\left(\frac{1}{(2s^{\omega(1)} n^{1/d} - 1)^2 - 1}, 2\right) = 2\log(2s^{\omega(1)} n^{1/d} - 1) \geq \frac{2\log n}{d},
\end{equation}
whereby
\begin{equation}
    \exp\left(-\frac{d}{2}S(\sigma^2, 2)\right) \leq \frac{1}{n},
\end{equation}
so, since $\what{p} \leq 1$, $pn \leq 1$ as well.

With these properties in mind, let us set up an application of Lemma~\ref{lem:frieze} via Lemmata~\ref{lem:second-moment-a} and \ref{lem:second-moment-b}, with which we will seek to show that $M \gtrsim pn^2$ with high probability.
Fix $c > 0$ a small constant, and take
\begin{align}
    r &\colonequals \left\lfloor \frac{2}{cpn} \right\rfloor \in \left[\frac{1}{cpn}, \frac{4}{cpn}\right], \\
    n^{\prime} &\colonequals \left\lfloor \frac{n}{r} \right\rfloor \in \left[\frac{1}{4}cpn^2, cpn^2\right], \\
    m &\colonequals \left\lfloor \frac{1}{32}cpn^2 \right\rfloor \in \left[\frac{1}{64}cpn^2, \frac{1}{16}cpn^2\right].
\end{align}
Then by Lemma~\ref{lem:second-moment-a}, the type (a) inequality, holds with $\alpha = m / 2n^{\prime} \geq \frac{1}{128}$.
    
For the type (b) inequality, we have $n^{\prime} \geq 4m$ by our choice, and by the upper bound on $\sigma^2$,
\begin{equation}
    \sigma^2 \leq \frac{1}{4n^{1/d}(n^{1/d} - 1)} \leq \frac{1}{4(n^{1/d} - 1)} \leq \frac{d}{4\log n},
\end{equation}
so, for sufficiently large $n$, the conditions of Lemma~\ref{lem:second-moment-b} are satisfied.
It remains to control the rates $R_1$ and $R_2$ appearing in the Lemma, and thus to bound $\beta$.

We note in advance that, by the upper bound on $\sigma^2$ and since $I(\sigma^2)$ is a decreasing function,
\begin{equation}
    I(\sigma^2) \geq I\left(\frac{1}{(2s^{\omega(1)}n^{1/d} - 1)^2 - 1}\right) = \frac{2\log n}{d} + \omega(\log s) = \frac{2\log n + \omega(\what{d})}{d}.
\end{equation}
The quantities appearing in these rates satisfy
\begin{align*}
        \frac{n^{\prime^2}}{pn^2 m} &\leq 64c \numberthis \\
        d(S(\sigma^2, 2) - I(\sigma^2)) - 2\log r &\leq -dI(\sigma^2) + 2\log n + 2\log c \\
        &\leq -\log n\left(\frac{d}{\log n}I(\sigma^2) - 2\right) \\
        &\leq -\omega(\log \what{d}). \numberthis
    \end{align*}
    We thus have
    \begin{align}
        R_1 &= K + \log\left(\frac{n^{\prime^2}}{pn^2 m}\right) \leq K + \log 64c, \\
        R_2 &= K + d(S(\sigma^2, 2) - I(\sigma^2)) + 4\log_+\left(\frac{d}{1 + \sigma^2}\right) - 2\log r = -\omega(\log \what{d}),
    \end{align}
    using in the latter our earlier result that $\frac{d}{1 + \sigma^2} = O(\what{d})$.
    
    For any $D > 0$, we may therefore choose $c$ small enough that $R_i \leq -D$ for $i = 1, 2$, so the whole rate function $F$ in Lemma~\ref{lem:second-moment-b} satisfies
    \begin{align*}
        F(\oa, \ob, \oc, \oj, \ok, \ol)
        &\leq 7H(\oa, \oa, \ob, \oc, \oc, \oj, \ok, \ol) - D(\oa + \ob + \oc + \oj + \ok + \ol)
    \end{align*}
    The first term is bounded uniformly by $7\log 9$, so sufficiently large $D$ we may ensure that $F$ is negative if any of $\oa, \ob, \oc, \oj, \ok,$ or $\ol$ is at least $\epsilon$.
    On the other hand if all of the parameters are at most $\epsilon$, then the first term is at most $-7(8\epsilon \log \epsilon + (1 - \epsilon)\log(1 - \epsilon))$, which tends to zero as $\epsilon \to 0$.
    Thus for sufficiently small $c$ we may make the supremum $\sup_{\bx \in \sA} F(\bx)$ bounded by any arbitrarily small positive number.
    In particular, for any $\epsilon > 0$ there exists $c > 0$ such that the type (b) inequality holds with $\beta < \epsilon$.
    For $c$ sufficiently small we may thus ensure, e.g., $\beta \leq \frac{1}{512} \leq \frac{\alpha}{4}$.
    Thus we may take $\gamma = \frac{1}{4}$ in Lemma~\ref{lem:frieze}, which gives that, with high probability, $M \geq \frac{1}{2048}cpn^2$.
    Substituting our lower bound on $\what{p}$ then gives the result as stated.
\end{proof}
    
\subsection{Linear or Nearly-Linear Error Lower Bound}

Finally, we prove the following result, which yields Part~4 of Theorem~\ref{thm:d-ll-logn}.
\begin{lemma}
	Let $s \colonequals d^{1/d}$.
    Suppose that $1 \leq d \ll \log n$, and that for some $a \in \RR$,
    \begin{equation}
    \sigma^2 \geq \frac{1}{(2s^a n^{1/d} - 1)^2 - 1}.
\end{equation}
    Then, there exists $c = c(a) > 0$ such that
    \begin{equation}
        |\sE| \geq e^{-cd}n
    \end{equation}
    with high probability.
\end{lemma}
\begin{proof}
We first produce similar preliminary bounds to before.
As before we have $\frac{d}{1 + \sigma^2} \leq d$, and thus
%
%and using that $1 - e^{-x} \leq x$
%\begin{align*}
%    \frac{d}{1 + \sigma^2} 
%    &\leq d(1 - (2s^a n^{1/d} - 1)^{-2}) \\
%    &= d(1 - (2(d^an)^{1/d} - 1)^{-2}) \\
%    &\leq d\left(1 - \frac{1}{4}(d^an)^{-2/d}\right) \\
%    &\leq d \log 4 + 2\log n + 2a\log d \\
%    &\leq 4\log n
%\end{align*}
%for sufficiently large $n$.
%Thus, again for sufficiently large $n$, we obtain the same bound as before,
%\begin{equation}
%    \frac{d}{1 + \sigma^2} \leq 4\what{d},
%\end{equation}
%and the same lower bound for $\what{p}$,
\begin{equation}
\what{p} \geq \frac{1}{1000} d^{-1/2}.
\end{equation}

For the bounds on the other, exponential factor in $p$, we have, again assuming $n$ is sufficiently large,
\begin{equation}
    S(\sigma^2, 2) \leq S\left(\frac{1}{(2s^a n^{1/d} - 1)^2 - 1}, 2\right) = 2\log(2s^{a} n^{1/d} - 1) \leq 2\log(2s^a n^{1/d}),
\end{equation}
whereby
\begin{equation}
    \exp\left(-\frac{d}{2}S(\sigma^2, 2)\right) \geq \frac{1}{2^ds^{ad}n} = \frac{d^{-a}}{2^d n}.
\end{equation}
Thus we may bound
\begin{equation}
    pn = \what{p}\exp\left(-\frac{d}{2}S(\sigma^2, 2)\right) n \geq \frac{1}{2000}d^{-a-1/2} 2^{-d}.
\end{equation}

We take a similar choice of parameters to the previous proof, only now taking $r$ a constant not depending on $n$.
Let $c = c(a) > 0$ be a constant to be fixed later, and take
\begin{align}
    r &\colonequals \left\lfloor 2e^{cd} \right\rfloor \in \left[e^{cd}, 4e^{cd}\right], \\
    n^{\prime} &\colonequals \left\lfloor \frac{n}{r} \right\rfloor \in \left[\frac{1}{4}e^{-cd}n, e^{-cd}n\right], \\
    m &\colonequals \left\lfloor \frac{1}{32}e^{-cd}n \right\rfloor \in \left[\frac{1}{64}e^{-cd}n, \frac{1}{16}e^{-cd}n\right].
\end{align}
Then by Lemma~\ref{lem:second-moment-a}, the type (a) inequality, holds as before with $\alpha = m / 2n^{\prime} \geq \frac{1}{128}$.
The conditions of Lemma~\ref{lem:second-moment-b} are again satisfied.
To control the rates appearing there, we again have
\begin{align}
    \frac{n^{\prime^2}}{pn^2 m} &\leq 64 \frac{e^{-cd}}{pn} = O((2e^{-c})^d d^{a + 1/2}),
    \intertext{and for the other rate we use that, by Proposition~\ref{prop:eta-zeta-bounds}, for all $\sigma^2$ we have $S(\sigma^2, 2) - I(\sigma^2) \leq \frac{3}{2}$, so}
    d(S(\sigma^2, 2) - I(\sigma^2)) - 2\log r &\leq d\left(\frac{3}{2} - 2c\right).
\end{align}
Thus, for $c$ sufficiently large the rates appearing in $F$ are
\begin{align}
    R_1 &= K + \log\left(\frac{n^{\prime^2}}{pn^2 m}\right) \leq K - \frac{c}{2}d, \\
    R_2 &= K + d(S(\sigma^2, 2) - I(\sigma^2)) + 4 \log_+\left(\frac{d}{1 + \sigma^2}\right) - 2\log r \leq K - \frac{c}{2}d,
\end{align}
thus choosing $c$ large enough we may again make both rates arbitrarily negative, and the remainder of the proof goes through as before.
\end{proof}

\addcontentsline{toc}{section}{Acknowledgments}
\section*{Acknowledgments}

We thank Cristopher Moore for helpful discussions about the second moment method.

\addcontentsline{toc}{section}{References}
\bibliographystyle{abbrvnat}
\bibliography{main}

\newpage

\appendix

\section{Greedy Algorithms and a Gaussian Limit}
\label{app:greedy}

To supplement our discussion of the MLE, we describe two natural greedy algorithms for estimating the planted permutation $\pi^{\star}$ and discuss their performance heuristically.
We believe the computations presented here are accurate, but for the sake of brevity we will not give detailed proofs.
The algorithms we analyze here are also \emph{improper} in the sense that they do not return a permutation; rather, they output an assignment of each $\bx_i$ to an element of $\{\by_1, \dots, \by_n\}$ with no restriction that each element is matched exactly once.
We leave more careful analysis of greedy algorithms which output a permutation to future work.

To summarize before presenting the details, the first greedy algorithm, where distance is measured as ordinary $\ell^2$ distance, will match the performance of the MLE when $d = o(\log n)$ but will make $n - o(n)$ errors once $d = \omega(\log n)$.
The second, which greedily selects the point with largest inner product with $\bx_i$, will make $n - o(n)$ errors when $d = o(\log n)$ but will sometimes (though not always) improve on the MLE in the $d = \omega(\log n)$ regime.
It is unclear what simplifying assumptions are reasonable when $d = \Theta(\log n)$, so we leave this case aside here; numerical evidence suggests that all three algorithms are competitive and none strictly dominates another in this regime.

\subsection{Algorithm 1: Greedy Distance}

The first algorithm we consider is perhaps the most immediately appealing greedy algorithm, which attempts to match each point to its nearest neighbor.
This may be viewed as greedily matching rows to columns in the matrix $\bW^{(0)}$ of pairwise squared distances between the $\bx_i$ and $\by_j$ formed as an intermediate step in our derivation of the MLE.
As a proxy for the error incurred by such an algorithm, we consider the number of $\bx_i$ whose nearest neighbor among the $\{\by_j\}_{j = 1}^n$ is not equal to $\by_i$, the set of which we denote
\begin{equation}
    \sE^{\dist} \colonequals \{i \in [n]: \|\bx_i - \by_i\|^2 > \|\bx_i - \by_j\|^2 \text{ for some } j \neq i\}.
\end{equation}
As another, simpler variant, we may also consider
\begin{equation}
    \widetilde{\sE}^{\dist} \colonequals \{(i, j) \in [n]^2: \|\bx_i - \by_i\|^2 > \|\bx_i - \by_j\|^2\},
\end{equation}
which satisfies
\begin{equation}
    \frac{1}{n-1}|\widetilde{\sE}^{\dist}| \leq |\sE^{\dist}| \leq |\widetilde{\sE}^{\dist}|.
\end{equation}

By linearity of expectation, we have
\begin{align*}
    \EE |\widetilde{\sE}^{\dist}|
    &= n(n-1) \cdot \PP[\|\bx_1 - \by_1\|^2 > \|\bx_1 - \by_2\|^2]
    \intertext{Here, we note that $\bx_1 - \by_1 = \bz_1$ whose squared norm has law $\chi^2(d)$ scaled by $\sigma^2$, and is independent from $\bx_1 - \by_2 = \bx_1 - \bx_2 - \bz_2$, whose squared norm has law $\chi^2(d)$ scaled by $2 + \sigma^2$. Thus, we may rewrite this probability in terms of two independent $A, B \sim \chi^2(d)$ as}
    &= n(n-1) \cdot \PP\left[\frac{A}{B} < \frac{\sigma^2}{2 + \sigma^2}\right]
    \intertext{The ratio $A / B$ has the $F$ distribution $F(\frac{d}{2}, \frac{d}{2})$, whose density is  $\frac{\Gamma(d)}{\Gamma(d/2)^2}x^{d/2 - 1} (1 + x)^{-d}dx$. Thus, so long as $\sigma^2 = o(1)$, we will have from integrating an initial segment of this density that}
    &\lesssim \sigma^d n^2, \numberthis
\end{align*}
which, up to lower-order terms, is the same as the expected number of augmenting 2-cycles for the MLE.
In particular, $\EE |\sE^{\dist}|$ is bounded by the same quantity and it is reasonable to believe, so long as this is $O(n)$, that this bound is tight.
We also expect this first moment computation to be an accurate estimate of the typical size of $\sE^{\dist}$.
Thus, we find that the error rate of a greedy distance algorithm asymptotically that of the MLE so long as $d = o(\log n)$.

However, once $d = \omega(\log n)$ this algorithm is much less effective than the MLE.
In that case, in the critical scaling for the MLE we have $\sigma^2 = \Theta(\frac{d}{\log n}) = \omega(1)$, so $\frac{\sigma^2}{2 + \sigma^2} \to 1$ as $n \to \infty$.
By evaluating the probability as a Laplace integral, we thus find
\begin{align*}==
    \PP\left[\frac{A}{B} < \frac{\sigma^2}{2 + \sigma^2}\right] 
    &\approx \frac{\Gamma(d)}{\Gamma(\frac{d}{2})^2} \exp\left(d\left[\frac{1}{2}\log\left(\frac{\sigma^2}{2 + \sigma^2}\right) - \log\left(1 + \frac{\sigma^2}{2 + \sigma^2}\right)\right]\right) \\
    &\approx \exp\left(d\left[\frac{1}{2}\log\left(\frac{\sigma^2}{2 + \sigma^2}\right) - \log\left(\frac{2 + 2\sigma^2}{2 + \sigma^2}\right) + \log 2 \right]\right) \\
    &= \exp\left(d\log \frac{\sqrt{\sigma^2(2 + \sigma^2)}}{1 + \sigma^2}\right), \numberthis
\end{align*}
whereby
\begin{equation}
    \frac{\log(1 \vee \EE |\widetilde{\sE}^{\dist}|)}{\log n} \approx 2 - \frac{d}{2\log n} \log \frac{(1 + \sigma^2)^2}{\sigma^2(2 + \sigma^2)} \approx 2 - \frac{d}{2\log n}\log\left(1 + \frac{1}{\sigma^4}\right) \approx 2 - \frac{d}{2\sigma^4 \log n}.
\end{equation}
So, while the MLE achieves perfect recovery for some $\sigma^2 = \Theta(d / \log n)$, the greedy distance algorithm only achieves perfect recovery for the asymptotically smaller $\sigma^2 = O(\sqrt{d / \log n})$.

A more informal way to make the same prediction is to first observe that, for large $d$, we have the distributional approximation $\chi^2(d) \approx \sN(d, 2d)$.
Then, when $\sigma^2 \gg 1$ the distances $\|\bx_i - \by_j\|^2$ are distributed approximately as $\sN((2 + \sigma^2)d, 2(2 + \sigma^2)^2d)$.
Likewise the $\|\bx_i - \by_i\|^2 = \|\bz_i\|^2$ are distributed approximately as $\sN(\sigma^2d, 2\sigma^4 d)$.
Moreover, we may make the simplifying assumption of thinking of these distances as independent.
Then, we expect strong recovery to only be possible when $\min_{j \neq i} \|\bz_j\|^2 \approx (2 + \sigma^2) d - \sqrt{4(2 + \sigma^2)^2d\log n}$ is at least the typical $\|\bz_i\|^2 \approx \sigma^2 d$.
This gives $(2 + \sigma^2)\sqrt{d\log n} \lesssim d$, or $\sigma^2 \lesssim \sqrt{d / \log n}$, as claimed above.

\begin{figure}
    \centering
    \includegraphics[scale=0.6]{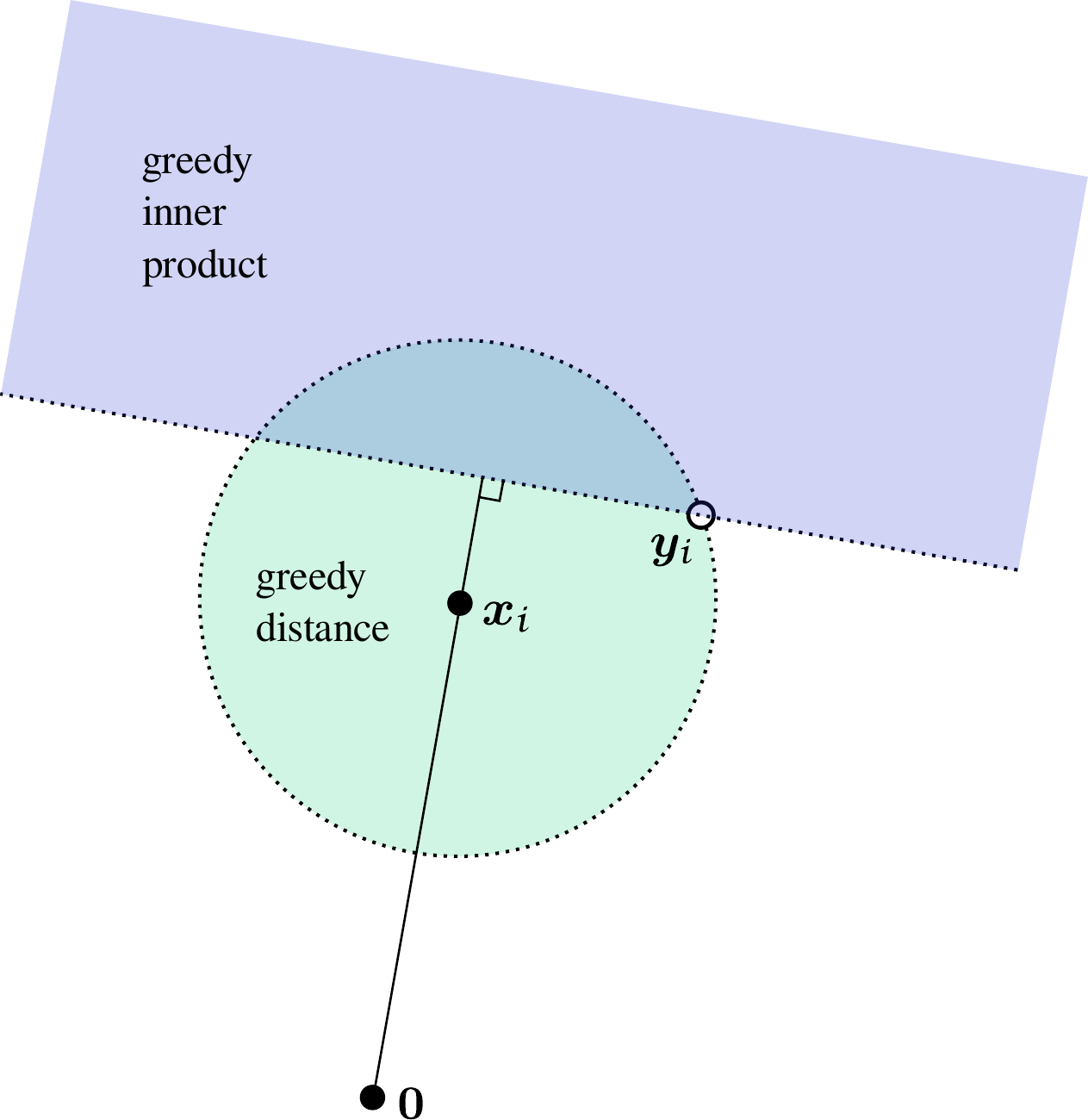}
    \caption{\textbf{Greedy algorithm proximal regions.} We illustrate the regions which must not contain any other $\by_j$ in order for $\bx_i$ to be matched with $\by_i$ in the two greedy algorithms we consider in Appendix~\ref{app:greedy}.}
    \label{fig:greedy}
\end{figure}

\subsection{Algorithm 2: Greedy Inner Product}

The second algorithm we consider applies the same greedy matching approach to the cost matrix $\bW$ formed for the MLE by subtracting out the norm terms when the squared distances are expanded.
The analogous error set is then
\begin{equation}
    \sE^{\pprod} \colonequals \{i \in [n]: \langle \bx_i, \by_i \rangle < \langle \bx_i, \by_j \rangle \text{ for some } j \neq i\}.
\end{equation}

Here a useful shortcut allows us to dispense with the low-dimensional case easily: as is apparent from Figure~\ref{fig:greedy}, if $i \notin \sE^{\pprod}$ then $\by_i$ is a vertex of the convex hull of $\by_1, \dots, \by_n$.
In particular then, this algorithm will only achieve even weak recovery when the convex hull of $n$ i.i.d.\ standard Gaussian vectors in $\RR^d$ has $\Omega(n)$ vertices with high probability.
As has been shown in the literature on this so-called \emph{Gaussian polytope} (e.g., \cite{HMR-2004-GaussianPolytopes}), the expected number of vertices is $(\log n)^{O(d)}$, whereby whenever $d = o(\log n / \log \log n)$ the greedy inner product algorithm will not achieve even weak recovery, having $|\sE^{\pprod}| = n - o(n)$.

On the other hand, when $d = \omega(\log n)$ we believe that the instance $\bW$ should, loosely speaking, behave in law like a matrix with independent entries (we will say more about how our computations here relate to prior work on such models below).
In this case, the ``planted'' or diagonal and ``null'' or off-diagonal distributions should be approximately
\begin{align}
    W_{ii} &= \langle \bx_i, \by_i \rangle = \|\bx_i\|^2 + \langle \bx_i, \bz_i \rangle \stackrel{(d)}{\approx} \sN(d, \sigma^2 d) \equalscolon \sP, \\
    W_{ij} &= \langle \bx_i, \by_j \rangle = \langle \bx_i, \bx_j + \bz_j \rangle \stackrel{(d)}{\approx} \sN(0, \sigma^2d) \equalscolon \sQ,
\end{align}
where we use that, because $\sigma^2 \gg 1$, we may neglect the fluctuations coming from terms not involving any $\bz_i$.

If this approximation is sound, then we expect $\max_{j \neq i}W_{ij} \approx \sqrt{2\sigma^2 d \log n}$.
On the other hand, $n - o(n)$ of the diagonal terms are of size $\Theta(d)$.
Therefore we expect the strong recovery regime to be when $\sqrt{2\sigma^2 d \log n} < d$, or $\sigma^2 < \frac{1}{2}\frac{d}{\log n}$.
This is strictly larger than the strong recovery regime $\sigma^2 < \frac{1}{4}\frac{d}{\log n}$ of the MLE.
We note that the former threshold $\sigma^2 = \frac{1}{2}\frac{d}{\log n}$ as the threshold we illustrated in Figure~\ref{fig:cycle-counts} when the number of augmenting 2-cycles for the MLE becomes macroscopic.

On the other hand, we also expect $\min_i W_{ii} \approx d - \sqrt{2\sigma^2 d \log n}$, so we expect the \emph{perfect} recovery regime for the greedy inner product algorithm to be when $d - \sqrt{2\sigma^2 d \log n} < \sqrt{2\sigma^2 d\log n}$, or $\sigma^2 < \frac{1}{8}\frac{d}{\log n}$.
This is strictly \emph{smaller} than the perfect recovery regime $\sigma^2 < \frac{1}{4}\frac{d}{\log n}$ of the MLE (which is the same as the strong recovery regime of the MLE).
Thus the final picture that emerges for the greedy inner product algorithm when $d = \omega(\log n)$ is that it achieves strong recovery for a greater range of $\sigma^2$, but has a region of sublinear error $\frac{1}{8}\frac{d}{\log n} < \sigma^2 < \frac{1}{2} \frac{d}{\log n}$, while the MLE has no region of sublinear error on this scale, instead achieving perfect recovery when $\sigma^2 < \frac{1}{4}\frac{d}{\log n}$; from the point of view of the polynomial error rate, the two algorithms are thus incomparable.

\subsection{Gaussian Limit in High Dimension}

The independent Gaussian limit discussed above falls in the range of models treated by previous works \cite{MMX-2019-PlantedMatching,SSZ-2020-SparsePlantedMatching,DWXY-2021-PlantedMatchingInfiniteOrder}.
In particular, it was predicted in \cite{MMX-2019-PlantedMatching,SSZ-2020-SparsePlantedMatching} and proved for certain models (not including the Gaussian model of $\sP$ and $\sQ$ above) in \cite{DWXY-2021-PlantedMatchingInfiniteOrder} that the strong recovery threshold in such a model should correspond to $\sqrt{n} B(\sP, \sQ) = 1$, where $B(\sP, \sQ)$ is the Bhattacharyya coefficient.
This may be computed in closed form for Gaussian distributions, which gives that the critical $\sigma^2$ should satisfy
\begin{equation}
    n = \frac{3 + 2\sigma^2}{2\sqrt{(2 + \sigma^2)(1 + \sigma^2)}}\exp\left(\frac{d}{2(3 + 2\sigma^2)}\right).
\end{equation}
As $n \to \infty$ the prefactor and the constant term in the exponent denominator are irrelevant, so this predicts a critical transition at $n = \exp(d / 4\sigma^2)$, or $\sigma^2 = \frac{1}{4}\frac{d}{\log n}$.

Per our discussion above, this is the correct strong recovery threshold for the MLE; indeed, the proof of the positive results in \cite{DWXY-2021-PlantedMatchingInfiniteOrder} goes by analyzing the MLE, so this is not surprising.
However, the greedy algorithm applied to this model (to agree with the setting of \cite{DWXY-2021-PlantedMatchingInfiniteOrder}, we should think of the input as the matrix $\bW$ with entries distributed roughly according to $\sP$ and $\sQ$, rather than the ``raw'' point sets $\{\bx_i\}$ and $\{\by_i\}$) achieves a better strong recovery threshold of $\sigma^2 = \frac{1}{2}\frac{d}{\log n}$.
Essentially the same is noted in Remark~1 of \cite{DWXY-2021-PlantedMatchingInfiniteOrder}, where the authors bring up a similar independent Gaussian model as an instance where the Bhattacharyya coefficient does \emph{not} give a correct prediction.

Our discussion above, however, gives some further nuance to this point if one is interested in sublinear error rates in addition to just strong recovery.
Namely, both in our model for high dimension and in the independent Gaussian model, the greedy algorithm achieves an inferior perfect recovery threshold, and, more generally, an inferior sublinear error rate to the MLE whenever $\sigma^2 < \frac{1}{4} \frac{d}{\log n}$, but a superior rate whenever $\frac{1}{4}\frac{d}{\log n} < \sigma^2 < \frac{1}{2}\frac{d}{\log n}$.

\section{Evaluation of Integral: Proof of Proposition~\ref{prop:I-eval}}
\label{app:I}

Recall our claim,
\begin{equation}
I(\sigma^2) \colonequals \int_0^1 \log\left(1 + \frac{1}{2\sigma^2}(1 - \cos(2\pi x))\right)dx = 2\log\left(\frac{1 + \sqrt{1 + \sigma^{-2}}}{2}\right).
\end{equation}
To lighten the notation, let us set $\lambda = \sigma^2$.
Differentiating under the integral sign, we have
    \begin{align*}
        I^{\prime}(\lambda) &= -\frac{1}{\lambda} \int_0^1 \frac{1 - \cos(2\pi x)}{(2\lambda + 1) - \cos(2\pi x)}dx
        \intertext{which we may write as a contour integral over $C$ the complex unit circle}
        &= -\frac{1}{2\pi \lambda} \oint_C \frac{1 - \frac{z + z^{-1}}{2}}{(2\lambda + 1) - \frac{z + z^{-1}}{2}} \frac{dz}{iz} \\
        &= -\frac{1}{2\pi i\lambda} \oint_C \frac{(z - 1)^2}{z(z^2 - (4\lambda + 2)z + 1}dz
        \intertext{where the integrand has poles at $z = 0, \rho_-, \rho_+$ for $\rho_{\pm} = 2\lambda + 1 \pm 2\sqrt{\lambda^2 + \lambda}$. Only $z = 0, \rho_-$ lie inside $C$, so by the residue theorem we have}
        &= -\frac{1}{\lambda}\left(\frac{1}{\rho_+\rho_-} + \frac{(\rho_- - 1)^2}{\rho_-(\rho_- - \rho_+)}\right)
        \intertext{which after some algebra reduces to}
        &= -\frac{1}{\lambda}\left(1 - \sqrt{\frac{\lambda}{\lambda + 1}}\right). \numberthis
    \end{align*}
    Since $\lim_{\lambda \to \infty} I(\lambda) = 0$, we then have
    \begin{align*}
        I(\lambda) &= \int_{\lambda}^{\infty} \frac{1}{t}\left(1 - \sqrt{\frac{t}{t + 1}}\right)dt
        \intertext{where the integrand has the explicit antiderivative $\log(t) - \log(1 + \sqrt{\frac{t}{t + 1}}) + \log(1 - \sqrt{\frac{t}{t + 1}})$ whose limit as $t \to \infty$ is $-\log(4)$, whereby we finish}
        &= -\log(4) - \log(\lambda) + \log\left(1 + \sqrt{\frac{t}{t + 1}}\right) - \log\left(1 - \sqrt{\frac{t}{t + 1}}\right) \\
        &= \log\left(\frac{1 + \sqrt{\frac{\lambda}{1 + \lambda}}}{4\lambda(1 - \sqrt{\frac{\lambda}{1 + \lambda}})}\right)
        \intertext{which after some algebra reduces to}
        &= \log\left(\left(\frac{1 + \sqrt{1 + \lambda^{-1}}}{2}\right)^2\right), \numberthis
    \end{align*}
    as claimed.

\section{Riemann Sum Analysis}
\label{app:riemann}

In this appendix we prove our bounds on the Riemann sums $S(\sigma^2, t)$.
We will proceed by relating $S(\sigma^2, t)$ to the following well-known family of polynomials.
\begin{definition}[Lucas polynomials]
    The \emph{Lucas polynomials} $L_k(x) \in \RR[x]$ for $k \geq 0$ are defined by the recursion
    \begin{align}
        L_0(x) &= 2, \\
        L_1(x) &= x, \\
        L_{k}(x) &= xL_{k - 1}(x) + L_{k - 2}(x) \text{ for } k \geq 2.
    \end{align}
\end{definition}
\noindent
The recursion may be solved as follows, a version of the usual approach for a second-order recurrence, only now parametrized by $x$ (see, e.g., \cite{HM-1985-PellLucasPolynomials}).
\begin{proposition}[Binet's formula]
    Let $\alpha(x), \beta(x)$ be the roots of $t^2 - tx - 1$, i.e.,
    \begin{align}
        \alpha(x) &= \frac{x + \sqrt{x^2 + 4}}{2}, \\
        \beta(x) &= \frac{x - \sqrt{x^2 + 4}}{2}.
    \end{align}
    Then, $L_n(x) = \alpha(x)^n + \beta(x)^n$.
\end{proposition}

The following is then the key statement relating the Lucas polynomials to our Riemann sums.
\begin{lemma}
    \label{lem:S-lucas}
    For any $t \geq 3$, $\exp(S(\sigma^2, t)) = (4\sigma^2)^{-t}(L_{2t}(2\sigma) - 2)$.
\end{lemma}
\noindent
We note that the same formula does not hold for $t = 2$: the left-hand side is $1 + \sigma^{-2}$, while the right-hand side is $1 + \sigma^{-2}/4$.
As we will see in the course of the proof, that is because the formula depends on the eigenvalues of the $t$-cycle graph $C_3$ appearing in the summation in $S(\sigma^2, t)$.

 \begin{figure}
        \centering
        \includegraphics[scale=0.8]{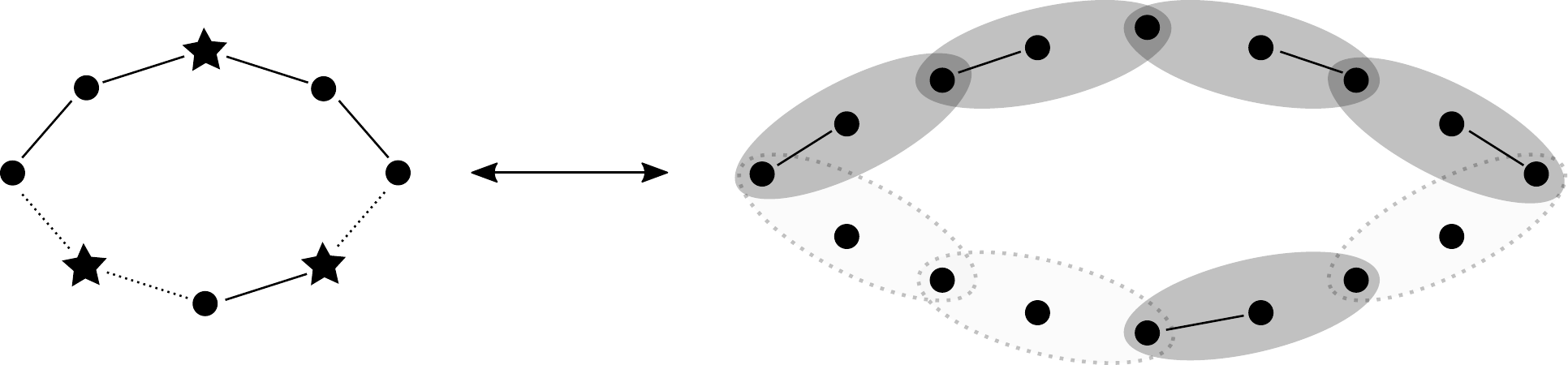}
        \caption{\textbf{Spanning forests and matchings.} We illustrate the bijection between rooted spanning forests on the $t$-cycle and matchings on the $2t$-cycle used in the proof of Lemma~\ref{lem:S-lucas}.}
        \label{fig:forest-matching-bijection}
    \end{figure}

\begin{proof}[Proof of Lemma~\ref{lem:S-lucas}]
    Recall that we denote by $C_t$ the cycle on $t$ vertices and by $\bL^{C_t}$ its graph Laplacian.
    Let $\lambda_{t, 1}, \dots, \lambda_{t, t}$ denote the eigenvalues of $\bL^{C_t}$.
    Then, we have
    \begin{equation}
        \exp(S(\sigma^2, t))
        = \prod_{j = 1}^t \left(1 + \frac{1}{4\sigma^2}\lambda_{t, j}\right)
        = \sum_{k = 0}^t E_k(\lambda_{t, 1}, \dots, \lambda_{t, t})(4\sigma^2)^{-k},
    \end{equation}
    where
    \begin{equation}
        E_k(a_1, \dots, a_t) \colonequals \sum_{1 \leq i_1 < \cdots < i_k \leq t} a_{i_1} \cdots a_{i_k}
    \end{equation}
    are the elementary symmetric polynomials.
    
    By a generalization of the matrix-tree theorem (see Theorem 7.5 of \cite{Biggs-1993-AlgebraicGraphTheory}), $E_k(\lambda_{t, 1}, \dots, \lambda_{t, t})$ is equal to the number of spanning forests of $C_t$ containing $t - k$ connected components and with each connected component having an assigned root vertex.
    The condition of a spanning forest having $t - k$ connected components is, for the specific case of the graph $C_t$, also equivalent to the forest containing $k$ edges.
    
    When $k < t$, then these forests are in bijection with the matchings on $C_{2t}$ also containing $k$ edges.
    An explicit bijection is as follows.
    Suppose the vertices of $C_{2t}$ are labelled $0, \dots, 2t - 1$ and the vertices of $C_t$ labelled $0, \dots, t - 1$.
    Suppose $M$ is a matching in $C_{2t}$.
    We build a spanning forest $F$ of $C_t$ by including the edge $\{i, i + 1\}$ whenever either $\{2i, 2i + 1\}$ or $\{2i + 1, 2i + 2\}$ is included in $M$, and by declaring $i$ a root vertex in $C_t$ if $2i$ is not adjacent to any edges of $M$ in $C_{2t}$.
    We illustrate this mapping in Figure~\ref{fig:forest-matching-bijection}, which shows that every connected component of $F$ formed this way indeed has a unique root vertex (located where the pairs of consecutive edges containing edges of $M$ switch from ``leaning'' counterclockwise to clockwise), and that the knowledge of the connected components and the root vertices of $F$ uniquely determines the preimage $M$.
    This holds so long as $k < t$; however, when $k = t$, then there are two matchings on $C^{2t}$ with $k = t$ edges, while $E_t(\lambda_{t, 1}, \dots, \lambda_{t, t}) = \det(\bL^{C_t}) = 0$.
    
    Let $M_{t, k}$ denote the number of matchings of $k$ edges in $C_t$.
    The result then follows from showing that $L_t(x)$ is the \emph{matching polynomial} of $C_t$ for $t \geq 3$:
    \begin{equation}
        L_t(x) = \sum_{k = 0}^{\lfloor t / 2 \rfloor} M_{t, k} x^{t - 2k}.
    \end{equation}
    This fact is known, though often phrased differently (e.g., Section 6 of \cite{Farrell-1979-MatchingPolynomials}), but we give the simple proof here for the sake of completeness.
    The statement is easily verified for $t = 3, 4$, and then it suffices to show the coefficient recursion for $t \geq 5$ that
    \begin{equation}
        M_{t, k} = M_{t - 1, k} + M_{t - 2, k - 1}.
    \end{equation}
    We present a bijective proof of this recursion in Figure~\ref{fig:matching-recursion}: fixing a sequence of three consecutive edges in $C_t$, we map any matching in $C_t$ to a matching in either $C_{t - 1}$ or $C_{t - 2}$ by a suitable replacement of this sequence by either two edges or one edge, and this mapping is visibly a bijection.
\end{proof}

\begin{figure}
    \centering
    \includegraphics[scale=0.8]{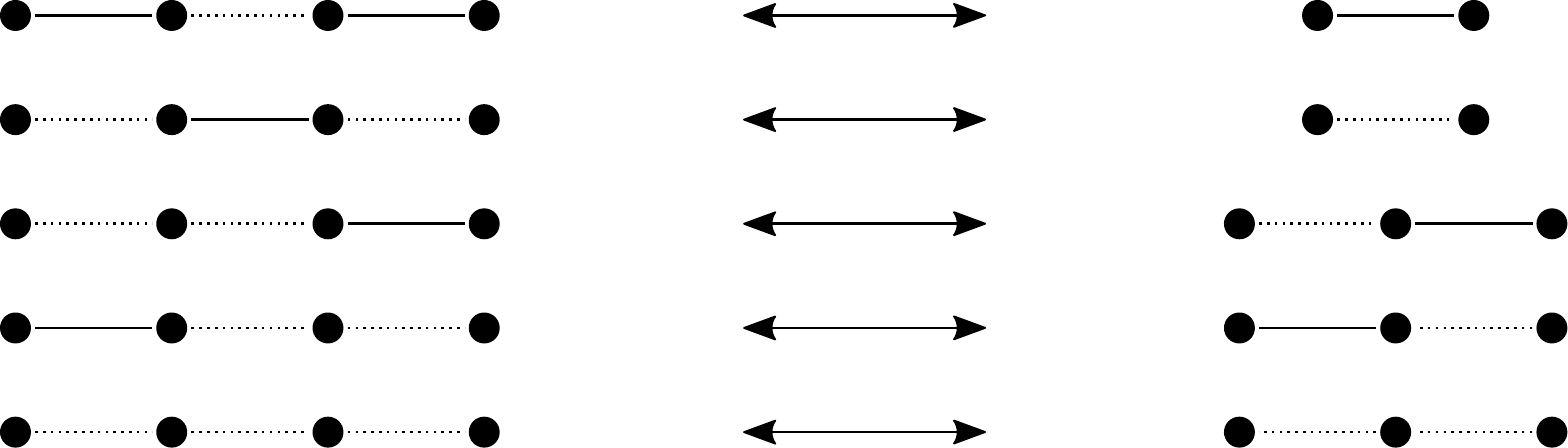}
    \caption{\textbf{Contracting matchings in cycles.} We illustrate the bijection between matchings in a $t$-cycle and those in either a $(t-1)$- or $(t - 2)$-cycle used in the proof of Lemma~\ref{lem:S-lucas}.}
    \label{fig:matching-recursion}
\end{figure}

\subsection{Discrete Concavity: Proof of Lemma~\ref{lem:S-lower}}

It suffices to show that $S(\sigma^2, t) - S(\sigma^2, t - 1)$ is decreasing in $t$, or equivalently that, for all $t \geq 3$,
\begin{equation}
    2S(\sigma^2, t) \stackrel{?}{>} S(\sigma^2, t - 1) + S(\sigma^2, t + 1).
\end{equation}

We consider separately the case $t = 3$.
Introducing for the sake of convenience a new variable $y = \sigma^{-2}$, we have the polynomials
\begin{align*}
    \exp(S(y^{-1}, 2)) &= 1 + y\sin^2\left(\frac{1}{2}\pi\right) \\
    &= 1 + y, \numberthis \label{eq:expS2} \\
    \exp(S(y^{-1}, 3)) &= \left(1 + y\sin^2\left(\frac{1}{3}\pi\right)\right)\left(1 + y\sin^2\left(\frac{2}{3}\pi\right)\right) \\
        &= \left(1 + \frac{3}{4}y\right)^2, \numberthis \label{eq:expS3} \\
        \exp(S(y^{-1}, 4)) &= \left(1 + y\sin^2\left(\frac{1}{4}\pi\right)\right)\left(1 + y\sin^2\left(\frac{2}{4}\pi\right)\right)\left(1 + y\sin^2\left(\frac{3}{4}\pi\right)\right) \\
        &= \left(1 + \frac{1}{2}y\right)^2(1 + y) \label{eq:expS4} \numberthis
\end{align*}
Thus it suffices to show that, for all $y > 0$,
\begin{equation}
    \left(1 + \frac{3}{4}y\right)^4 \stackrel{?}{>} \left(1 + \frac{1}{2}y\right)^2(1 + y)^2,
\end{equation}
which follows by the arithmetic-geometric mean inequality which gives $(1 + \frac{1}{2}y)(1 + y) < (\frac{1}{2}(1 + \frac{1}{2}y + 1 + y))^2 = (1 + \frac{3}{4}y)^2$.

Now, suppose $t \geq 4$.
Then, exponentiating both sides and applying Lemma~\ref{lem:S-lucas}, it suffices to show that
\begin{equation}
    (L_{2t}(x) - 2)^2 \stackrel{?}{>} (L_{2t - 2}(x) - 2)(L_{2t + 2}(x) - 2)
\end{equation}
for all $x > 0$.
Letting $a = \alpha(x)^2 > 1$, we have $\beta(x)^2 = a^{-1}$.
In terms of $a$, we may then expand either side as
\begin{align}
    (L_{2t}(x) - 2)^2 &= a^{2t} + a^{-2t} + 2 - 4a^{t} - 4a^{-t}, \\
    (L_{2t - 2}(x) - 2)(L_{2t + 2}(x) - 2) &= a^{2t} + a^{-2t} + a^2 + a^{-2} - 2a^{t - 1} - 2a^{-t + 1} - 2a^{t + 1} - 2a^{-t - 1}.
\end{align}
Therefore, it suffices to show that, for all $a > 0$,
\begin{align*}
    0 &\stackrel{?}{<} 2 - 4a^t - 4a^{-t} - a^2 - a^{-2} + 2a^{t - 1} + 2a^{-t + 1} + 2a^{t + 1} + 2a^{-t - 1} \\
    &= 2 + 2a^{t - 1}(1 - a)^2 + 2a^{-t + 1}(1 - a^{-1})^2 - a^2 - a^{-2}. \numberthis
\end{align*}

Viewing $t$ for a moment as a continuous parameter, we note that the derivative of the above expression with respect to $t$ is $2\log a (a^{t - 1}(1 - a)^2 - a^{-t + 1}(1 - a^{-1})^2) > 0$ for any $t > 1$ and $a > 1$, since $a^{t - 1} > a^{-t + 1}$ and $(1 - a)^2 > (1 - a^{-1})^2 = (\frac{a - 1}{a})^2$.
Thus this expression is increasing in $t$, so it suffices to consider $t = 4$.
In that case, we have the factorization
\begin{equation}
    2 + 2a^{3}(1 - a)^2 + 2a^{-3}(1 - a^{-1})^2 - a^2 - a^{-2} = \frac{(a - 1)^4(2a^6 + 4a^5 + 6a^4 + 7a^3 + 6a^2 + 4a + 2)}{a^5},
\end{equation}
which shows strict positivity for any $a > 1$.

\subsection{Upper Bound: Proof of Lemma~\ref{lem:S-upper}}

We want to show $S(\sigma^2, t) < tI(\sigma^2)$.
Exponentiating either side, we observe that
\begin{align*}
    \exp(tI(\sigma^2)) &= \left(\frac{1 + \sqrt{1 + \sigma^{-2}}}{2}\right)^{2t}, \numberthis \\
    \exp(S(\sigma^2, t))
    &= (4\sigma^2)^{-t}\left(\alpha(\sqrt{4\sigma^2})^{2t} + \beta(\sqrt{4\sigma^2})^{2t} - 2\right) \\
    &= (4\sigma^2)^{-t}\left(\left(\frac{\sqrt{4\sigma^2} + \sqrt{4\sigma^2 + 4}}{2}\right)^{2t} + \left(\frac{\sqrt{4\sigma^2} - \sqrt{4\sigma^2 + 4}}{2}\right)^{2t} - 2\right) \\
    &= \left(\frac{1 + \sqrt{1 + \sigma^{-2}}}{2}\right)^{2t} + \left(\frac{1 - \sqrt{1 + \sigma^{-2}}}{2}\right)^{2t} - 2(4\sigma^2)^{-t} \\
    &= \exp(tI(\sigma^2)) + \left(\frac{1 - \sqrt{1 + \sigma^{-2}}}{2}\right)^{2t} - \left(\frac{2^{1/t}}{4\sigma^2}\right)^t. \numberthis
\end{align*}
Thus it suffices to show that
\begin{equation}
    \frac{2^{1/t}}{4\sigma^2} \stackrel{?}{>} \left(\frac{\sqrt{1 + \sigma^{-2}} - 1}{2}\right)^2,
\end{equation}
which we verify as
\begin{equation}
    \frac{2^{1/t}}{4\sigma^2} - \left(\frac{\sqrt{1 + \sigma^{-2}} - 1}{2}\right)^2 > \frac{1}{4\sigma^2} - \left(\frac{\sqrt{1 + \sigma^{-2}} - 1}{2}\right)^2 = \frac{1}{2}\left(\sqrt{1 + \sigma^{-2}} - 1\right) > 0.
\end{equation}

\subsection{Miscellaneous Rate Functions: Proof of Proposition~\ref{prop:eta-zeta-bounds}}
\label{app:pf:prop:eta-zeta-bounds}

Again introducing $y = \sigma^{-2}$ and using our expressions for $S_t$ for $t = 2, 3,$ and 4 from \eqref{eq:expS2}, \eqref{eq:expS3}, and \eqref{eq:expS4}, respectively, as well as the expression
\begin{equation}
    \exp(I(y^{-1})) = \left(\frac{1 + \sqrt{1 + y}}{2}\right)^2 \in \left[1 + \frac{y}{4}, 1 + \frac{y}{2}\right],
\end{equation}
we may compute as follows:
\begin{align}
    \exp(\eta_1) &= \frac{(1 + y)^{3/4}}{(1 + \frac{1}{2}y)^{1/2}(1 + y)^{1/4}} \nonumber \\
    &= \sqrt{\frac{1 + y}{1 + \frac{1}{2}y}}, \\
    \exp(\eta_2) &= \frac{1 + y}{1 + \frac{3}{4}y}, \\
    \exp(\eta_3) &= \frac{2\sqrt{1 + y}}{1 + \sqrt{1 + y}}.
\end{align}
Thus we find $\eta_3 \geq \eta_1$, since by concavity of the square root $\frac{1}{2}(1 + \sqrt{1 + y}) \leq \sqrt{1 + \frac{1}{2}y}$.
We also have, by the arithmetic-geometric mean inequality,
\begin{equation}
    \frac{\exp(\eta_1)}{\exp(\eta_2)} = \frac{1 + \frac{3}{4}y}{\sqrt{(1 + y)(1 + \frac{1}{2}y)}} \geq 1,
\end{equation}
so $\eta_3 \geq \eta_1 \geq \eta_2$.
For the upper bound on $\eta_3$, we have
\begin{align*}
    \exp(\eta_3)
    &\leq \sqrt{\frac{1 + y}{1 + \frac{1}{4}y}} \\
    &= \sqrt{1 + \frac{\frac{3}{4}y}{1 + \frac{1}{4}y}} \\
    &\leq \exp\left(\frac{\frac{3}{4}y}{2 + \frac{1}{2}y}\right) \\
    &= \exp\left(\frac{1}{\frac{2}{3} + \frac{8}{3}\sigma^2}\right). \numberthis
\end{align*}

\section{Edge Probability Prefactor: Proof of Proposition~\ref{prop:pair-aug-lb}}
\label{app:pf:prop:pair-aug-lb}
    We begin by producing the following formula for $p$: let $g \sim \sN(0, \sigma^2)$ and $u \sim \chi^2(d)$ be independent.
    Then,
    \begin{equation}
        p = \PP\left[g \geq \sqrt{u} \right].
    \end{equation}
    We work directly from the earlier expression, in the special case $t = 2$:
    \begin{align*}
        \PP[(1, 2) \text{ is augmenting}] &= \PP\left[ \langle \bz_1 - \bz_2, \bx_2 - \bx_1 \rangle \geq \|\bx_2 - \bx_1\|^2 \right] \\
        &= \PP\left[ \left\langle \frac{\bz_1 - \bz_2}{\sqrt{2}}, \frac{\bx_2 - \bx_1}{\|\bx_2 - \bx_1\|} \right\rangle \geq \left\|\frac{\bx_2 - \bx_1}{\sqrt{2}}\right\| \right], \numberthis
    \end{align*}
    where we observe now that $(\bx_2 - \bx_1) / \|\bx_2 - \bx_1\|$ has the law of a uniformly distributed unit vector, and is independent from $\|\bx_2 - \bx_1\| / \sqrt{2}$, which has the law of the norm of a standard gaussian vector, which is that of $\sqrt{u}$.
    Moreover, $(\bz_1 - \bz_2) / \sqrt{2}$ has the law $\sN(\bm 0, \sigma^2 \bm I_d)$, so the left-hand side of the probability has law $\sN(0, \sigma^2)$, giving the claim.
        
    Note that the upper bound on $\hat{p}$ follows from Proposition~\ref{prop:cycle-aug-prob}.
    For the lower bounds, we first give two quantitative lower bounds, a looser one that holds for all $d \geq 1$ and a tighter one that holds for all $d \geq 4$.
    We will use the following ``Mills' ratio'' lower bounds on Gaussian tails (see, e.g., \cite{Duembgen-2010-MillsRatio}): for all $t > 0$,
    \begin{equation}
        \Px_{g \sim \sN(0, 1)}[g \geq t] \geq \frac{t}{1 + t^2} \frac{1}{\sqrt{2\pi}} \exp\left(-\frac{t^2}{2}\right) \geq \left(\frac{1}{t} - \frac{1}{t^3}\right) \frac{1}{\sqrt{2\pi}} \exp\left(-\frac{t^2}{2}\right) \label{eq:mills}
    \end{equation}
    
    For our first lower bound, we use the first lower bound of \eqref{eq:mills}:
    \begin{align*}
        p 
        &= \Ex_{u \sim \chi^2(d)}\Px_{g \sim \sN(0, 1)}\left[g \geq \sqrt{\frac{u}{\sigma^2}}\right] \\
        &\geq \sqrt{\frac{\sigma^2}{2\pi}} \Ex_u \frac{\sqrt{u}}{\sigma^2 + u}\exp\left(-\frac{u}{2\sigma^2}\right) \\
        &= \sqrt{\frac{\sigma^2}{2\pi}} \frac{1}{2^{d/2} \Gamma(\frac{d}{2})} \int_0^{\infty} \frac{u^{\frac{d - 1}{2}}}{\sigma^2 + u} \exp\left(-\frac{1 + \sigma^{-2}}{2}u\right)du \\
        &= \sqrt{\frac{\sigma^2}{2\pi}} \frac{1}{2^{d/2} \Gamma(\frac{d}{2})} \left(\frac{1 + \sigma^{-2}}{2}\right)^{-\frac{d - 1}{2}} \int_0^{\infty} \frac{v^{\frac{d - 1}{2}}}{\frac{1 + \sigma^2}{2} + v} e^{-v}dv \\
        &= \exp\left(-\frac{d}{2}S(\sigma^2, 2)\right) \frac{\sqrt{1 + \sigma^2}}{2\sqrt{\pi}} \frac{1}{\Gamma(\frac{d}{2})} \int_0^{\infty} \frac{v^{\frac{d - 1}{2}}}{\frac{1 + \sigma^2}{2} + v} e^{-v}dv. \numberthis
    \end{align*}
    Working now with the remaining integral,
    \begin{align*}
        \int_0^{\infty} \frac{v^{\frac{d - 1}{2}}}{\frac{1 + \sigma^2}{2} + v} e^{-v}dv
        &\geq \frac{1}{1 + \sigma^2}\int_0^{\infty} \frac{v^{\frac{d - 1}{2}}}{1 + v} e^{-v}dv \\
        &\geq \frac{1}{1 + \sigma^2}\int_1^{\infty} \frac{v^{\frac{d - 1}{2}}}{2v} e^{-v}dv \\
        &= \frac{1}{2(1 + \sigma^2)} \Gamma\left(\frac{d - 1}{2}, 1\right), \numberthis
    \end{align*}
    and we find that, for all $d \geq 1$,
    \begin{align*}
        \what{p} 
        &\geq \frac{1}{4\sqrt{\pi(1 + \sigma^2)}} \frac{\Gamma(\frac{d - 1}{2}, 1)}{\Gamma(\frac{d}{2})}
        \intertext{and, bounding the $\Gamma$ functions,}
        &\geq \frac{1}{40\sqrt{\pi d(1 + \sigma^2)}}. \numberthis
    \end{align*}

    For our second lower bound, we suppose that $d \geq 4$ and use the second lower bound of \eqref{eq:mills}:
    \begin{align*}
        p &\geq \frac{1}{\sqrt{2\pi}}\Ex_u \left(\left(\frac{\sigma^2}{u}\right)^{1/2} - \left(\frac{\sigma^2}{u}\right)^{3/2}\right)\exp\left(-\frac{u}{2\sigma^2}\right) \\
        &= \frac{1}{\sqrt{2\pi}} \frac{1}{2^{d/2}\Gamma(\frac{d}{2})}\left(\sigma\int_0^{\infty} u^{\frac{d - 3}{2}}\exp\left(-\frac{1 + \sigma^{-2}}{2}u\right)du - \sigma^3\int_0^{\infty} u^{\frac{d - 5}{2}}\exp\left(-\frac{1 + \sigma^{-2}}{2}u\right)du\right)
        \intertext{where both integrals converge due to our assumption that $d \geq 4$. Performing the same change of various as before, we find}
        &= \frac{1}{\sqrt{2\pi}} \frac{1}{\Gamma(\frac{d}{2})}\left(\frac{\Gamma(\frac{d - 1}{2})}{\sqrt{2}}(1 + \sigma^2)^{1/2} - \frac{\Gamma(\frac{d - 3}{2})}{2\sqrt{2}}(1 + \sigma^2)^{3/2}\right) \exp\left(-\frac{d}{2}S(\sigma^2, 2)\right), \numberthis
    \end{align*}
    and thus, rearranging, we find that
    \begin{align*}
        \what{p} 
        &\geq \frac{1}{2\sqrt{\pi}} \cdot \frac{\Gamma(\frac{d - 1}{2})}{\Gamma(\frac{d}{2})}(\sigma^2 + 1)^{\frac{1}{2}} - \frac{1}{4\sqrt{\pi}} \cdot \frac{\Gamma(\frac{d - 3}{2})}{\Gamma(\frac{d}{2})} (\sigma^2 + 1)^{3/2}
        \intertext{and bounding the $\Gamma$ function ratios from above and below,}
        &\geq \frac{1}{\sqrt{2\pi}} \left(\frac{1 + \sigma^2}{d}\right)^{\frac{1}{2}} - 2\left(\frac{1 + \sigma^2}{d}\right)^{3/2} \numberthis
    \end{align*}
    
    For $1 \leq d \leq 40$, by assumption we have $\sigma^2 \leq 1$, so by our first bound we find
    \begin{equation}
        \what{p} \geq \frac{1}{40 \sqrt{80\pi}} \geq \frac{1}{1000} \sqrt{\frac{1 + \sigma^2}{d}},
    \end{equation}
    using that $1 + \sigma^2 \leq 2$ and $d \geq 1$.
    For $d \geq 40$, we have $1 \leq \frac{d}{40}$, so $1 + \sigma^2 \leq \frac{d}{20}$.
    On the interval $x \in [0, \frac{1}{20}]$ we have $\frac{1}{\sqrt{2\pi}}x^{1/2} - 2x^{3/2} \geq \frac{1}{4}x^{1/2}$, so by our second bound we have
    \begin{equation}
        \what{p} \geq \frac{1}{4} \sqrt{\frac{1 + \sigma^2}{d}}.
    \end{equation}
    Combining the two cases gives the result.

\end{document}